%% file: main.tex
\documentclass[sigconf,authorversion]{acmart}

\AtBeginDocument{%
  }

\copyrightyear{2023} 
\acmYear{2023} 
\setcopyright{acmlicensed}\acmConference[KDD '23]{Proceedings of the 29th ACM SIGKDD Conference on Knowledge Discovery and Data Mining}{August 6--10, 2023}{Long Beach, CA, USA}
\acmBooktitle{Proceedings of the 29th ACM SIGKDD Conference on Knowledge Discovery and Data Mining (KDD '23), August 6--10, 2023, Long Beach, CA, USA}
\acmPrice{15.00}
\acmDOI{10.1145/3580305.3599325}
\acmISBN{979-8-4007-0103-0/23/08}

\input{notation}

\begin{document}

\title{Efficient Centrality Maximization with Rademacher Averages}


\author{Leonardo Pellegrina}
\affiliation{%
  \institution{Dept. of Information Engineering, University of Padova}
  \streetaddress{Via Gradenigo 6b}
  \city{Padova}
  \country{Italy}
  \postcode{35131}
}
\email{leonardo.pellegrina@unipd.it}

%


\begin{abstract}

The identification of the set of $k$ most central nodes of a graph, or centrality maximization, is a key task in network analysis, with various applications ranging from finding communities in social and
biological networks to understanding which seed nodes are important to diffuse information in a graph. 
As the exact computation of centrality measures does not scale to modern-sized networks, the most practical solution is to resort to rigorous, but efficiently computable, randomized approximations. 
In this work 
we present \algname, the first algorithm based on progressive sampling 
to compute high-quality approximations of the set of $k$ most central nodes. 
\algname\ is based on a novel approach to efficiently estimate Monte Carlo Rademacher Averages, a powerful tool from statistical learning theory to compute sharp \emph{data-dependent} approximation bounds. 
Then, 
we study the sample complexity of centrality maximization using the VC-dimension, a key concept from statistical learning theory. 
We show that the number of random samples required to compute high-quality approximations 
scales with finer characteristics of the graph, such as its vertex diameter, 
or of the centrality of interest, 
significantly improving looser bounds derived from standard techniques. 
We apply \algname\ to analyze large real-world networks, showing that it significantly outperforms the state-of-the-art approximation algorithm in terms of number of samples, running times, and accuracy. 
\end{abstract}

\ccsdesc[500]{Information systems~Data mining}
\ccsdesc[500]{Mathematics of computing~Probabilistic algorithms}
\ccsdesc[500]{Theory of computation~Sketching and sampling}

\keywords{Centrality Maximization, Rademacher Averages, Random Sampling}

\maketitle

\input{centrmax}

\end{document}

%% file: notation.tex
\usepackage{graphicx}
\usepackage{bm}
\usepackage{dsfont}
\usepackage{subcaption}
\usepackage[ruled,vlined,linesnumbered]{algorithm2e}
\usepackage{enumitem}

\newcommand\numberthis{\addtocounter{equation}{1}\tag{\theequation}}

\newcommand{\bc}{betweenness centrality}
\newcommand{\algname}{\textsc{CentRA}}

\newcommand{\pars}[1]{\left( #1 \right)}
\newcommand{\sqpars}[1]{\left[ #1 \right]}
\newcommand{\brpars}[1]{\left\lbrace #1 \right\rbrace}
\newcommand{\qt}[1]{\lq\lq#1\rq\rq}
\newcommand{\qtm}[1]{\text{\lq\lq}#1\text{\rq\rq}}
\newcommand{\ind}[1]{\mathds{1} \left[ #1 \right] }

\newcommand{\sample}{\mathcal{H}}
\newcommand{\nodeset}{\mathcal{S}}
\newcommand{\centr}{\mathcal{C}}
\newcommand{\hyperg}{\mathfrak{H}}

\newcommand{\BO}[1]{\mathcal{O}\left(#1\right)}
\newcommand{\BOi}[1]{\mathcal{O}( #1 )}
\newcommand{\BTi}[1]{\Theta( #1 )}

\newcommand{\E}{\mathbb{E}}
\newcommand{\F}{\mathcal{F}}
\newcommand{\X}{\mathcal{X}}

\DeclareMathOperator*{\argmax}{arg\,max}

\newcommand{\abs}[1]{\lvert#1 \rvert}

\newcommand{\rc}{\rade(\F, m)}

\newcommand{\era}{\erade\left(\F, \sample\right)}

\newcommand{\sd}{\mathsf{D}}
\newcommand{\supdev}{\sd(\F, \sample)}
\newcommand{\probdist}{\gamma}

\newcommand{\vsigma}{{\bm{\sigma}}}

\newcommand{\mera}{\rade^{t}_{m}(\F, \sample, \vsigma)}
\newcommand{\amera}{\trade^{t}_{m}(\F, \sample, \vsigma)}

\newcommand{\ewvar}{\tilde{w}}

\newcommand{\R}{\mathbb{R}}
\newcommand{\rade}{\mathsf{R}}
\newcommand{\erade}{\hat{\rade}}
\newcommand{\trade}{\tilde{\rade}}

\newcommand{\unionbdelta}{\ln \bigl( \frac{5}{\delta} \bigr)}

%% file: centrmax.tex

\section{Introduction}
\label{sec:intro}
Measuring the importance of nodes of a graph is a fundamental task in 
graph analytics~\cite{newman2018networks}. 
To this aim, several \emph{centrality} measures have been proposed to quantify the importance of nodes and \emph{sets} of nodes, such as the 
betweenness~\cite{freeman1977set}, closeness~\cite{bavelas1948mathematical}, pagerank~\cite{page1999pagerank}, random walk~\cite{newman2005measure}, and harmonic~\cite{boldi2014axioms} centralities. 
While the specific notion of centrality to use depends on the application~\cite{ghosh2014interplay}, 
identifying important components of a network is crucial for many relevant tasks, such as 
finding communities in social and
biological networks~\cite{girvan2002community},
identifying vulnerable nodes that may be attacked to disrupt the functionality of a network~\cite{iyer2013attack}, 
and
for influence maximization~\cite{kempe2003maximizing}. 
As in many applications the goal of the analysis is to identify important \emph{sets} of central nodes, 
recent works proposed extensions and generalizations of centrality measures from single nodes to sets of nodes, such as the set betweenness~\cite{ishakian2012framework}, coverage~\cite{yoshida2014almost}, $\kappa$-path~\cite{alahakoon2011k} and set closeness centrality~\cite{bergamini2019computing}. 
A key problem is \emph{set centrality maximization}, that consists in the task of finding the most central set of nodes of cardinality at most $k$, for some integer $k \geq 1$. 
For example, $k$ may represents the budget of an attacker with the intention of disrupting a network, or the number of seed nodes to effectively diffuse relevant information in a graph. 

To address the problems mentioned above, 
many algorithms have been proposed to compute centrality metrics exactly~\cite{brandes2001faster,erdHos2015divide}. 
For most centralities, however, these algorithms poorly scale when applied to large graphs~\cite{RiondatoK15,borassi2019kadabra}. 
Therefore, often the only viable solution is to resort to approximate, but rigorous, estimates. 
In fact, several works~(e.g.,~\cite{RiondatoK15,RiondatoU18,borassi2019kadabra,cousins2021bavarian,pellegrina2021silvan} for the \bc, see also Section~\ref{sec:relwork}) have recently proposed sampling approaches that provide approximations with rigorous guarantees of the centralities of individual nodes. 
In all these methods, the most critical challenge to address is to relate the \emph{size} of the random sample with the \emph{accuracy} of the approximation, i.e. studying the trade-off between the time to process the sample and the probabilistic guarantees that the sample provides w.r.t. the exact analysis. 
To achieve this goal, these works make use of sophisticated probabilistic and sampling techniques (described in more detail in Section~\ref{sec:relwork}).

However, as anticipated before, in many applications the interest of the analysis is on the centrality of \emph{sets} of nodes, a task that is considerably more challenging, while yielding results that are substantially different. 
In fact, the centrality of a set of nodes is weakly related to the centralities of the individual nodes making part of the set. 
This implies that knowing the centralities of several nodes does not provide  information on the centrality of the corresponding set of nodes as a whole. 
More importantly, all the most refined approaches mentioned above~\cite{RiondatoK15,RiondatoU18,borassi2019kadabra,cousins2021bavarian,pellegrina2021silvan} are specifically tailored to provide guarantees for the centralities of individual nodes, and do not generalize to sets of nodes of size $> 1$.
For this reason, they cannot be applied to 
centrality maximization or other tasks related to node sets.

Given the difficulty of set centralities approximation, this problem received only scant attention. 
In fact, only a few recent works 
proposed sampling-based algorithms to compute rigorous approximations of the centrality maximization task~\cite{yoshida2014almost,mahmoody2016scalable}. 
As for the individual nodes case, the main technical challenge of these methods is to relate the accuracy of the approximation to the \emph{size} of the random sample to process, that impacts both the space and running time of the algorithm, imposing a severe trade-off. 
As the state-of-the-art techniques~\cite{yoshida2014almost,mahmoody2016scalable} typically provide loose guarantees on the approximation quality, and 
require unrealistic a-priori knowledge of underlying graph, this problem remains difficult to address and computationally expensive in practice.
Set centrality approximation is a challenging problem that requires new techniques to provide rigorous, yet efficiently computable, accuracy guarantees. 
This is the main goal of this work.

\paragraph{Our contributions} 
We introduce a new algorithm, called \algname\ (\underline{Cent}rality Maximization with \underline{R}ademacher \underline{A}verages) for efficient centrality maximization using progressive sampling, that significantly outperforms previous works. 
\begin{itemize}[leftmargin=.15in]
\item \algname\ is the first algorithm for approximate centrality maximization with data-dependent bounds and progressive sampling (Section~\ref{sec:boundsupdevrade} and~\ref{sec:algorithm}). 
The progressive sampling strategy of \algname\ makes it \emph{oblivious} to the (unknown) set centrality of the optimal solution to approximate, in strong contrast with previous works (that instead requires a-priori knowledge of the optimal set centrality, see Section~\ref{sec:prelimsetcentr}). 
To provide sharp data-dependent bounds, the central contribution at the core of \algname\ is a new algorithmic strategy to estimate Monte Carlo Rademacher Averages, a key tool from statistical learning theory. 
\algname\ enables the efficient computation of tight \emph{graph-} and \emph{data-dependent} probabilistic upper bounds to the Supremum Deviation, a key component of its strategy to provide rigorous high-quality approximations of set centralities. 
\item We study the sample complexity of centrality maximization with the VC-dimension, a fundamental concept from statistical learning theory. 
We derive new bounds to the number of sufficient samples to obtain high-quality approximations of set centralities with high probability (Section~\ref{sec:samplecomplexity}). 
Such bounds scale with granular properties of the graph or of the centrality measure of interest. Therefore, 
our novel bounds improve the analysis based on standard techniques (e.g., based on a union bound),
typically offering a refined dependence on the graph size. 
In practice, these results can be naturally combined with the progressive sampling approach of \algname\ by providing an upper limit to the number of samples it needs to process. 
\item We perform an extensive experimental evaluation, testing \algname\ on several large real-world graphs (Section~\ref{sec:experiments}). 
Compared to the state-of-the-art, 
\algname\ computes much sharper approximations for the same amount of work, or speeds up the analysis, of up to two orders of magnitude, for obtaining approximations of comparable quality. 
\end{itemize}

\section{Related Work}
\label{sec:relwork}
We first introduce related work for centrality approximation of individual nodes. 
We focus on methods providing rigorous guarantees, a necessary requirement for many tasks and downstream analysis. 
Riondato and Kornaropoulos~\cite{RiondatoK15} study the VC-dimension~\cite{Vapnik:1971aa}, a key notion from statistical learning theory, of shortest paths to obtain approximation bounds for the betweenness centrality with random sampling. 
Riondato and Upfal~\cite{RiondatoU18} present an improved method based on deterministic upper bounds to Rademacher Averages~\cite{KoltchinskiiP00} and pseudodimension~\cite{pollard2012convergence}. 
KADABRA~\cite{borassi2019kadabra} furtherly improved \bc\ approximations with adaptive sampling and a weighted union bound, 
while BAVARIAN~\cite{cousins2021bavarian} used Monte Carlo Rademacher Averages~\cite{BartlettM02,pellegrina2022mcrapper} as a framework to fairly compare the accuracy of different estimators of the \bc. 
SILVAN~\cite{pellegrina2021silvan} furtherly improved upon~\cite{cousins2021bavarian} leveraging non-uniform approximation bounds. 
As anticipated in Section~\ref{sec:intro}, all these methods are specifically designed for individual nodes, and do not provide information or guarantees for node sets. 
Our algorithm \algname\ leverages novel data-dependent bounds based on Monte Carlo Rademacher Averages. 
We remark that, while these techniques have been previously applied to centrality approximation of \emph{individual} nodes~\cite{cousins2021bavarian,pellegrina2021silvan} and other pattern mining problems~\cite{pellegrina2022mcrapper,simionatobounding}, 
embedding and generalizing them to 
the context of \emph{set} centrality approximation
is extremely challenging, as we discuss in Section~\ref{sec:rad_ave}. 
The main technical contributions of \algname\ are to prove new concentration results, that are general and may be of independent interest, and to develop a new algorithmic approach, based on the contraction principle of Rademacher Averages (\cite{ledoux1991probability,BartlettM02}, see Section~\ref{sec:boundsupdevrade} of \cite{boucheron2013concentration}) and progressive sampling to obtain efficiently computable, but still accurate, approximations for set centrality maximization. 
Interestingly, contraction approaches are fundamental theoretical tools that have been applied to prove generalization bounds for several complex Machine Learning tasks, such as multi-class classification and ranking~\cite{ShalevSBD14}, and models~\cite{BartlettM02,mohri2018foundations}, such as SVMs and neural networks~\cite{cortes1995support,anthony1999neural}. 
In this work we show their relevance and practical impact within graph analysis.

For the problem of approximate set centrality maximization, 
the state-of-the-art method is from Mahmoody et al.~\cite{mahmoody2016scalable}. 
This work presents HEDGE, an algorithm based on a general framework for centralities that can be defined as submodular set functions, and that admit an appropriate randomized sampling oracle (an hyper-edge sampler, see Section~\ref{sec:prelimsetcentr}). 
In our work we adhere to a similar random sampling framework, but we develop a novel progressive and adaptive algorithm, 
and derive sharper sample complexity and data-dependent bounds. 
We show that our new algorithm \algname\ significantly outperforms HEDGE~\cite{mahmoody2016scalable} in terms of running time, accuracy, and approximation guarantees. 

Other works considered efficient algorithms for different tasks or to approximate different notions of group centralities.
\cite{bergamini2018scaling}~considers the problem of closeness centrality maximization. 
\cite{medya2018group}~studies the problem of maximizing the centrality of a set of nodes by adding new edges to the graph. 
\cite{angriman2020group}~proposes a new group centrality measure inspired by Katz centrality that can be efficiently approximated.  
\cite{angriman2021group}~studies the problem of approximating the group harmonic and group closeness centralities~\cite{boldi2013core,boldi2014axioms}. 
\cite{de2020estimating}~approximates percolation centrality with pseudodimension, while~\cite{Chechik2015} is based on probability proportional to size sampling to estimate closeness centralities.
Other recent works extended the computation of the \bc\ to dynamic~\cite{bergamini2014approximating,bergamini2015fully,hayashi2015fully}, uncertain~\cite{saha2021shortest}, and temporal networks~\cite{santoro2022onbra}.

\section{Preliminaries}
In this Section we introduce the notation and the most important concepts for our algorithm \algname. 
\label{sec:prelims}

\subsection{Set Centralities}
\label{sec:prelimsetcentr}

Let a graph $G = (V , E)$ with $n = |V|$ nodes. 
Let $\centr : 2^V \rightarrow \R$ be a \emph{set centrality} function from the set of all possible subsets of $V$ to $\R$. 
Let $\hyperg \subseteq 2^V$ be a space of sets of nodes, or \emph{hyper-edges}~\cite{yoshida2014almost,mahmoody2016scalable}, such that each $h \in \hyperg$ is a subset of $V$.   
For a parameter $k \geq 1$, we define a family of functions $\F$ as 
$\F = \brpars{ f_\nodeset : \nodeset \subseteq V , |\nodeset| \leq k } $, 
where each $f_\nodeset : \hyperg \rightarrow \{ 0 , 1 \}$ is a function from $\hyperg$ to $\{ 0,1 \}$ such that $f_\nodeset(h) = 1$ if at least one of the nodes of $\nodeset$ belongs to $h$, $0$ otherwise; more formally, we have $f_\nodeset(h) = \ind{ \nodeset \cap h \neq \emptyset }$. 
Note that, to simplify notation, we will also denote $f_u$ for all nodes $u \in V$ as equivalent of $f_{ \{ u \} }$.  
Under this setting, it is possible to define various set centralities $\centr(\cdot)$ as the average value of $f_\nodeset$ over $\hyperg$
\begin{align*}
\centr(\nodeset) = \frac{1}{| \hyperg |} \sum_{h \in \hyperg} f_\nodeset (h) , 
\end{align*}
using different notions of the space $\hyperg$. 
For instance, we define the \emph{set \bc} by taking $\hyperg$ as the family of sets of nodes that are internal to shortest paths of the graph $G$, where each shortest path is from the node $u$ to the node $v$, for all pairs of nodes $u,v \in V , u \neq v$. 
A large value of $\centr(\nodeset)$ denotes that many shortest paths of $G$ pass along any of the nodes of $\nodeset$. 
Analogously, for the $\kappa$-path centrality~\cite{alahakoon2011k}, $\hyperg$ is the family of sets of nodes traversed by simple paths of length at most $\kappa$; for the triangle centrality~\cite{mahmoody2016scalable}, $\hyperg$ contains sets of nodes incident to triangles of $G$. 
Intuitively, the centrality $\centr(\nodeset)$ of a node set $\nodeset$ is large if the nodes of $\nodeset$ \qt{cover} a large fraction of the space $\hyperg$. 
Finally, note that $\centr(\nodeset)$ cannot be obtained from the values $\{ \centr(\{u\}) : u \in \nodeset \}$, i.e. the values of the centralities of individual nodes making part of the set $\nodeset$.

The problem of \emph{centrality maximization} is, for a given $k \geq 1$, to identify a set $\nodeset^* \subseteq V$ of size at most $k$ maximizing $\centr(\cdot)$, such that 
\begin{align*}
\nodeset^* = \argmax_{\nodeset \subseteq V , |\nodeset| \leq k} \centr(\nodeset) . 
\end{align*}
As the computation of $\nodeset^*$ is NP-Hard for most set centrality functions~\cite{fink2011maximum}, the only viable solution is to obtain approximations that are efficient to compute, but with rigorous guarantees in terms of solution quality. 
We are interested in the following standard notion of $\alpha$-approximation. 
\begin{definition}
For any $\alpha > 0$, a set $\nodeset \subseteq V$ with $|\nodeset| \leq k$ provides an $\alpha$-approximation of $\nodeset^*$ if it holds 
$\centr(\nodeset) \geq \alpha \centr(\nodeset^*)$. 
\end{definition}
The state-of-the-art approach for the approximation of centrality maximization relies on the \emph{submodularity} of the set centrality function~\cite{schrijver2003combinatorial}. 
The well known greedy algorithm that, starting from an empty set of nodes $\nodeset = \emptyset$, iteratively picks and inserts in $\nodeset$ a new node $u^*$ maximizing the set centrality $u^* = \argmax_u \centr(\nodeset \cup \{ u\})$, for $k$ iterations, achieves an $(1 - 1/e)$-approximation. 
Such approximation ratio is essentially the best possible unless $P=NP$~\cite{feige1998threshold,fink2011maximum}.  

When considering large graphs, this approach is not efficient. 
In fact, simply computing the exact \bc\ of individual nodes requires $\BOi{ n |E| }$ time with Brandes' algorithm~\cite{brandes2001faster} (with a matching lower bound~\cite{borassi2016into}); furthermore, each iteration of the greedy algorithm would require to evaluate and update the centralities $\centr(\nodeset \cup \{ u\})$ for nodes $u \not\in \nodeset$, which is clearly infeasible in reasonable time on large graphs. 
For this reason, approaches based on random sampling have been proposed~\cite{yoshida2014almost,mahmoody2016scalable} to scale the approximation of set centralities. 

Let a sample $\sample = \{ h_1 , \dots , h_m \}$ be a multiset of size $m$ of hyper-edges from $\hyperg$,  
where each $h \in \sample$ is taken independently and uniformly at random from $\hyperg$. 
The estimate $\centr_\sample(\nodeset)$ of $\centr(\nodeset)$ for the set $\nodeset$ computed on $\sample$ is defined as 
\begin{align*}
\centr_\sample(\nodeset) = \frac{1}{m} \sum_{ s = 1 }^m f_\nodeset (h_s) . 
\end{align*}
The key observation of sampling approaches~\cite{yoshida2014almost,mahmoody2016scalable} is that the greedy algorithm can be applied to a random sample, rather than considering the entire space $\hyperg$. 

Let \texttt{greedyCover($k , \sample$)} be the algorithm defined as follows: starting from $\nodeset = \emptyset$, at every iteration $i \in [1,k]$, insert the node $u^* = \argmax_u \centr_\sample (\nodeset \cup \{ u \})$ in $\nodeset$; after all iterations, return $\nodeset$. 

It is clear that the \emph{size} $m$ of the random sample $\sample$ imposes a trade-off between the \emph{accuracy} of the approximation and the running time of \texttt{greedyCover}. Our goal is to precisely and sharply quantify this trade-off. 
First, note that, by definition, $\centr_\sample(\nodeset)$ is an \emph{unbiased} estimator of $\centr(\nodeset)$, as  
$\E_\sample[ \centr_\sample(\nodeset) ] = \centr(\nodeset) , \forall \nodeset \subseteq V$.
However, in order to provide guarantees in terms of solution quality, it is necessary to properly quantify $| \centr(\nodeset) - \centr_\sample(\nodeset) |$, i.e.  the deviation of $\centr_\sample(\nodeset)$ w.r.t. its expectation $\centr(\nodeset)$, for all $\nodeset \subseteq V$.

A key quantity to study in order to control the approximation accuracy is the \emph{Supremum Deviation (SD)} $\supdev$, defined as 
\begin{equation*}
\supdev = \sup_{f_\nodeset \in \F} | \centr(\nodeset) - \centr_\sample(\nodeset) |.
\end{equation*}
Note that it is not possible to evaluate $\supdev$ explicitly, as the values of $\centr(\nodeset)$ are unknown. 
Therefore, our goal is to obtain tight \emph{upper bounds} to the SD $\supdev$. 

The state-of-the-art method for approximate centrality maximization is HEDGE~\cite{mahmoody2016scalable}. 
This algorithm computes, for an $\varepsilon \in (0 , 1 - 1/e)$, an $(1 - 1/e - \varepsilon)$-approximation with high probability. 
More precisely, their analysis shows that if \texttt{greedyCover($k,\sample$)} is applied to a random sample $\sample$ of size 
\begin{align}
m \in \BO{ \frac{ k\log ( n ) + \log( 1/\delta ) } {\varepsilon^2 \centr(\nodeset^*) } } , \label{eq:unionboundsamples}
\end{align}
then it holds $\supdev \leq \varepsilon \centr(\nodeset^*) / 2$ 
with probability $\geq 1 - \delta$;  
moreover, this is a sufficient condition to guarantee that the output $\nodeset$ of \texttt{greedyCover($k , \sample$)} is an $(1 - 1/e - \varepsilon)$-approximation of $\nodeset^*$. 
The analysis of~\cite{mahmoody2016scalable} (Lemma 2 and Theorem~1) combines the Chernoff bound (to upper bound the deviation $| \centr(\nodeset) - \centr_\sample(\nodeset) |$ of a set $\nodeset$ with high probability) with an union bound over $n^k$ events (an upper bound to the number of non-empty subsets of $V$ of size $\leq k$). 

As introduced in previous Sections, there are two key limitation of this approach. First, the union bound leads to loose guarantees, as it ignores any information of the input graph (apart from its size, which is very large for real-world graphs).
Then, the number of samples to generate to compute an high-quality approximation of the most central set of nodes (given by \eqref{eq:unionboundsamples}) depends on $\centr(\nodeset^*)$, that is \emph{unknown} a-priori.

Our new algorithm \algname\ tackles both these issues. Our first goal is to obtain a finer characterization of the trade-off between the size $m$ of the random sample and bounds to the SD $\supdev$. 
\algname\ employs advanced and sharp \emph{data-dependent} bounds to the SD based on Rademacher Averages (defined in Section~\ref{sec:rad_ave}), leading to a much more efficient algorithm for approximate centrality maximization. 
Moreover, \algname\ is oblivious to the (unknown) values of $\centr(\nodeset^*)$, computing an high-quality approximation \emph{progressively} and \emph{adaptively} with progressive sampling. 
Overall, \algname\ is the first progressive and data-dependent approximation algorithm for set centralities.
While in this work we focus on the task of centrality maximization,
\algname's bounds and rigorous guarantees may be useful to scale other exploratory analyses based on the centralities of sets of nodes. 
Furthermore, \algname\ 
directly applies to all set centrality functions supported by the framework introduced in Section~\ref{sec:prelimsetcentr}. 

\subsection{Rademacher Averages}
\label{sec:rad_ave}

Rademacher averages are fundamental tools of statistical learning theory~\cite{KoltchinskiiP00,koltchinskii2001rademacher,BartlettM02} and of the study of empirical processes~\cite{boucheron2013concentration}. 
We introduce the main notions and results relevant to our setting and defer additional details to~\cite{boucheron2013concentration,ShalevSBD14,mitzenmacher2017probability}. 
While Rademacher averages are defined for arbitrary distributions and real valued functions, here we focus on the scenario of interest for set centralities. 

The \emph{Empirical Rademacher Average} (ERA) $\era$ of the family of functions $\F$ computed on $\sample$ is a key quantity to obtain a data-dependent upper bound to the supremum deviation $\supdev$. 
Let $\vsigma=\left< \vsigma_1,\dots, \vsigma_m \right>$ be a vector of $m$ i.i.d. Rademacher random variables, such that each entry $\vsigma_s$ of $\vsigma$ takes value in $\{-1,1\}$ with equal probability. 
The ERA $\era$ is defined as
\begin{equation*}
 \era = \E_\vsigma \left[ \sup_{f_\nodeset \in \F } \frac{1}{m}
 \sum_{s=1}^m \vsigma_s f_\nodeset(h_s) \right].
\end{equation*}
A central result in statistical learning theory implies that the \emph{expected} supremum deviation is sharply controlled by twice the expected ERA, where the expectation is taken over the sample $\sample$; 
this fundamental result is known as the \emph{symmetrization lemma} (see Lemma 11.4 of \cite{boucheron2013concentration} and Lemma~\ref{symlemma}). 
However, the exact computation of $\era$ is usually intractable, since it is not feasible to evaluate all the $2^m$ assignments of $\vsigma$. 
A natural solution to estimate the ERA is based on a Monte-Carlo approach~\cite{BartlettM02}. 

For $t \ge 1$, let $\vsigma \in \{-1,1\}^{t\times m}$ be a $t\times m$ matrix of i.i.d. Rademacher random variables. The \emph{Monte-Carlo Empirical Rademacher average (MCERA)} $\mera$ of $\F$ on $\sample$ using $\vsigma$ is:
\begin{equation*}
\mera = \frac{1}{t} \sum_{j=1}^t \sup_{f_\nodeset \in \F} \Biggl\{ \frac{1}{m} \sum_{s=1}^m\vsigma_{js}f_\nodeset(h_s) \Biggr\}.
\end{equation*}
The MCERA is a powerful tool as it allows to directly estimate the expected supremum deviation from several random partitions of the data, also by taking into account the data-dependent structure of $\F$.
Most importantly, it provides very accurate bounds even with a typically small number $t$ of Monte Carlo trials (as we show in Section~\ref{sec:experiments}). 
For this reason, this Monte Carlo process leads to much more accurate bounds to the SD, compared to \emph{deterministic upper bounds} (e.g., obtained with Massart's Lemma~\cite{Massart00}) or other \emph{distribution-free} notions of complexity, such as the VC-dimension. 

While the MCERA has been applied to centrality approximation for individual nodes~\cite{cousins2021bavarian,pellegrina2021silvan}, its generalization to centralities of sets of nodes is highly non-trivial. 
In fact, for the case of set centralities, computing the MCERA is NP-Hard\footnote{Computing $\max_{\nodeset \subseteq V , |\nodeset| \leq k}\centr_{\sample}(\nodeset)$ reduces to computing $\mera$ by taking $t=1$ and $\vsigma = \{ 1 \}^{1 \times m}$.}.
This is in strong contrast with the restriction to individual nodes ($k=1$), where it can be computed efficiently.  
Moreover, the weighted sum $\sum_{s=1}^m \vsigma_{js} f_\nodeset(h_s)$ is neither monotone or submodular w.r.t. $\nodeset$, due to the negative entries of $\vsigma$. 
For this reason, the greedy algorithm for constrained submodular function maximization described in Section~\ref{sec:prelimsetcentr} does not provide any approximation guarantee for the MCERA. 
It may be tempting to represent the function $\sum_{s=1}^m \vsigma_{js} f_\nodeset(h_s)$ as the difference  
$\sum_{s : \vsigma_{js} = 1 } f_\nodeset(h_s) - \sum_{s : \vsigma_{js} = -1 } f_\nodeset(h_s)$
of two monotone submodular functions; 
unfortunately, there are strong inapproximability results for this maximization problem~\cite{submodulardifference}. 
Moreover, we note that, while recent general approaches based on branch-and-bound optimization can be applied to compute the MCERA for general structured function families~\cite{pellegrina2022mcrapper}, they would still require a very expensive enumeration of the space of node subsets, incurring in high running times. 
These approaches cannot be adapted to centrality maximization as they would defy the purpose of using random sampling to scale the computation. 
These highly non-trivial computational barriers impose a significant challenge to obtaining accurate data-dependent approximations to set centralities. 

In this work we tackle this problem with a new approach to upper bound the MCERA, that we introduce in Section \ref{sec:boundsupdevrade}. Our idea is to relate the ERA of arbitrary sets of nodes (of cardinality $\leq k$) to the ERA of linear combinations of \emph{individual nodes}. 
We make use of the contraction principle~\cite{BartlettM02,boucheron2013concentration}, a fundamental tool in statistical learning theory. 
This approach leads to accurate and efficiently computable probabilistic upper bounds to the ERA, key to our algorithm \algname.  

\section{\algname: Efficient Centrality Maximization with Rademacher Averages}

This Section presents our contributions.
In Section \ref{sec:boundsupdevrade} we introduce new data-dependent bounds to the Supremum Deviation, the key techniques to our algorithm. 
In Section \ref{sec:algorithm} we present our algorithm \algname\ for efficient centrality maximization with progressive sampling. 
In Section \ref{sec:samplecomplexity} we prove new sample complexity bounds for the rigorous approximation of centrality maximization. 

\subsection{Bounding the Supremum Deviation}
\label{sec:boundsupdevrade}

In this Section we introduce our new, efficiently computable, data-dependent bounds to the SD, the key components of \algname. 

For $t,m \geq 1$, recall that $\vsigma \in \{-1 , 1 \}^{t \times m} $ is defined as a $t \times m $ matrix of Rademacher random variables, such that, for all $s \in [1,m], j \in [1,t]$, $\vsigma_{js} \in \{-1 , 1 \}$ independently and with equal probability. 
We define the \emph{Approximate Monte Carlo Empirical Rademacher Average} (AMCERA) $\amera$ as 
\begin{align*}
\amera = \frac{1}{t} \sum_{j=1}^t \sup_{\nodeset \subseteq V , |\nodeset| \leq k} \Biggl\{ \frac{1}{m}
 \sum_{u \in \nodeset} \sum_{s=1}^m \vsigma_{js} f_u(h_s) \Biggr\}  .
\end{align*}
Note that, as anticipated in Section \ref{sec:rad_ave}, the AMCERA $\amera$ takes the form of a linear combination of empirical deviations of \emph{individual nodes}, rather than arbitrary sets of nodes. 
With Theorem~\ref{thm:eraboundhypercube} we show that the AMCERA $\amera$ gives a sharp bound to the ERA $\era$, while Lemma~\ref{amceraefficient} proves that we can compute it efficiently. 
A key quantity governing the accuracy of the probabilistic bound of Theorem~\ref{thm:eraboundhypercube} is the \emph{approximate wimpy variance} $\ewvar_{\F}(\sample)$, that is defined as
\begin{align*}
\ewvar_{\F}(\sample) = \sup_{\nodeset \subseteq V , |\nodeset| \leq k}  \Biggl\{ \frac{b_\sample}{m} \sum_{u \in \nodeset} \sum_{s=1}^{m} \pars{f_u (h_s)}^2 \Biggr\} ,
\end{align*} 
where $b_\sample = \max_{h \in \sample} |h|$. 
We are now ready to state our new, efficiently computable, probabilistic bound to the ERA $\era$. 
Our proof (deferred to the Appendix for space constraints) is based on two key theoretical tools: a contraction inequality for Rademacher Averages~\cite{ledoux1991probability} and a concentration bound for functions defined on the binary hypercube~\cite{boucheron2013concentration}. 
\begin{theorem}
\label{thm:eraboundhypercube}
For $t,m \geq 1$, let $\vsigma \in \{-1 , 1 \}^{t \times m} $ be a $t \times m$ matrix of Rademacher random variables, such that, for all $s \in [1,m], j \in [1,t]$, $\vsigma_{sj} \in \{-1 , 1 \}$ independently and with equal probability. 
Then, 
with probability $\geq 1- \delta$ over $\vsigma$, it holds
\begin{align*}
\era \leq \amera + \sqrt{\frac{4 \ewvar_{\F}(\sample) \ln \bigl( \frac{1}{ \delta } \bigr) }{ t m }} . 
\end{align*}
\end{theorem}

We now investigate additional properties of the AMCERA. 
The following result proves guarantees on the expected approximation ratio of the AMCERA w.r.t. the MCERA. We note that, while there are no guarantees on the relation between the MCERA $\mera$ and the AMCERA $\amera$ 
in the worst-case 
(i.e., that holds \emph{for all} assignments of $\vsigma \in \{-1,1\}^{t \times m}$), their expectations (taken w.r.t. to $\vsigma$) are tightly related. 
In fact, we prove that the \emph{expected} AMCERA $\amera$ is within a factor $k$ of the ERA $\era$, 
therefore it provides a $k$-approximation of the MCERA $\mera$ in expectation. 
To make the dependence of the approximation bound on $k$ explicit, define the set of functions
$\F_{j} = \brpars{ f_\nodeset : \nodeset \subseteq V , |\nodeset| \leq j } $ (note that $\F = \F_k$). 
\begin{lemma}
\label{thm:expectappxratio}
It holds
\begin{align*}
\erade\left(\F_{k}, \sample\right)
\leq 
\E_{\vsigma} \sqpars{ \trade^{t}_{m}(\F_{k}, \sample, \vsigma) }
\leq 
k\erade\left(\F_{1}, \sample\right) 
\leq k\erade\left(\F_{k}, \sample\right) .
\end{align*}
\end{lemma}
In Section~\ref{sec:algorithm} we prove that the AMCERA can be computed efficiently, enabling the joint scalability and accuracy of \algname. 

We now prove that the greedy solution $\nodeset = \text{\texttt{greedyCover}}(k , \sample)$ gives a sharp \emph{upper bound} to the optimum set centrality $\centr(\nodeset^{*})$. 
Interestingly, the accuracy of this tail bound is \emph{independent} of $\F$ and $k$, but leverages the concentration of the empirical estimator $\sup_\nodeset \centr_\sample(\nodeset)$ of the optimal set centrality $\centr(\nodeset^{*})$ directly. 
Our novel proof is based on its \emph{self-bounding} property~\cite{boucheron2013concentration}. 
This result is instrumental for bounding the SD $\supdev$ (Theorem~\ref{thm:bound_dev}) and for the stopping condition of \algname, as we discuss in Section~\ref{sec:algorithm}. 
\begin{theorem}
\label{thm:selfboundingupperbound}
Let a sample $\sample$ of size $m$ and let $\nodeset$ be the output of \texttt{greedyCover($k , \sample$)}. 
With probability $\geq 1 - \delta$, it holds
\begin{align*}
\centr(\nodeset^*) \leq  \frac{\centr_{\sample}(\nodeset)}{1-\frac{1}{e}} + \sqrt{ \pars{ \frac{ \ln(\frac{1}{\delta}) }{ m } } ^2 + \frac{ 2 \centr_\sample(\nodeset) \ln(\frac{1}{\delta}) }{ (1-\frac{1}{e}) m } } + \frac{ \ln(\frac{1}{\delta}) }{ m } .
\end{align*}
\end{theorem}

The following result combines Theorem~\ref{thm:eraboundhypercube} with sharp concentration inequalities~\cite{boucheron2013concentration} to obtain a probabilistic bound to the Supremum Deviation $\sd(\F, \sample)$ from a random sample $\sample$ using the AMCERA. (proof deferred to the Appendix.)
\begin{theorem}
\label{thm:bound_dev}
Let $\sample$ be a sample of $m$ hyper-edges taken i.i.d. uniformly from $\hyperg$. For $k \geq 1$, let $\nodeset$ be the output of \texttt{greedyCover($k,\sample$)}.  
For any $\delta \in (0,1)$, define $\nu, \tilde{\rade}, \rade,$ and $\eta$ as
\begin{align*}
& \nu \doteq \frac{\centr_\sample(\nodeset)}{(1-\frac{1}{e})} + \sqrt{ \pars{ \frac{ \unionbdelta }{ m } } ^2 + \frac{ 2 \centr_\sample(\nodeset) \unionbdelta }{ (1-\frac{1}{e}) m } } + \frac{ \unionbdelta }{ m } , \\ 
& \tilde{\rade} \doteq \amera + \sqrt{ \frac{4  \ewvar_{\F}(\sample) \unionbdelta }{t m} } , \\
& \rade \doteq \tilde{\rade} + \sqrt{ \pars{\frac{\unionbdelta}{m}}^2 + \frac{2\unionbdelta \tilde{\rade} }{ m } } + \frac{\unionbdelta}{m} , \\
& \eta \doteq 2\rade + \sqrt{\frac{2 \unionbdelta
    \left( \nu + 4\rade \right)}{m}}
        + \frac{ \unionbdelta}{3m} \numberthis \label{eq:epsrade} .
\end{align*}
With probability at least $1-\delta$ over the choice of $\sample$ and $\vsigma$, it holds
$\sd(\F, \sample) \le \eta$.
\end{theorem}

In the following Section, we show that these results are essential to the theoretical guarantees of \algname.

\subsection{\algname\ Algorithm}
\label{sec:algorithm}

In this Section we present \algname, our algorithm based on progressive sampling for centrality maximization. Algorithm \ref{algo:main} shows the pseudo-code of \algname, that we now describe in more detail. 

\begin{algorithm}[htb]
\SetNoFillComment%
  \KwIn{Graph $G=(V,E)$; hyper-edge space $\hyperg$; $k , t \geq1$; $\varepsilon \in (0 , 1-\frac{1}{e}); \delta \in (0,1)$.}
  \KwOut{Set $\nodeset \subseteq V $, s.t. $|\nodeset| \leq k , \centr(\nodeset) \geq ( 1 - \frac{1}{e} - \varepsilon ) \centr(\nodeset^*)$ with probability $\ge 1 - \delta$}
  $i \gets 0$; $m_i \gets 0$; $\sample_i \gets \emptyset$\; \label{alg:init}
  \While{true}{ \label{alg:startiters}
  $i \gets i+1$; 
  $\delta_i \gets \delta / 2^i$\;
  $m_i \gets $ \texttt{samplingSchedule()}\; \label{alg:schedule}
  $\sample_i \gets \sample_{i-1} \cup  $ \texttt{sampleHEs($\hyperg , m_i - m_{i-1}$)}\; \label{alg:firstfirstline}
  $\nodeset \gets $ \texttt{greedyCover($k , \sample_i$)}\; \label{alg:greedycover}
  $\vsigma \gets $ \texttt{sampleRadeVars($t \times m_i$)}\; \label{alg:amcerastart}
  \ForAll{$u \in V$}{
  	\lForAll{$j \in [1,t]$}{
  		$r_j^u \gets \frac{1}{m_i} \sum_{s=1}^{m_i} \vsigma_{js} f_u (h_s)$} \label{alg:amceraupdate}
  }
  \lForAll{$j \in [1,t]$}{
  	$r_j \gets \sup_{\nodeset^\prime \subseteq V , |\nodeset^\prime| \leq k} \big\{ \sum_{u \in \nodeset^\prime}  r_j^u \big\}$} \label{alg:amceramaxj}
  $\trade^{t}_{m}(\F, \sample_i, \vsigma) \gets \frac{1}{t} \sum_{j=1}^{t} r_j$\; \label{alg:amceraend}
  $\xi \gets \sqrt{ \biggl( \frac{\ln\bigl( \frac{5}{\delta_i} \bigr) }{m_i} \biggr)^2 + \frac{ 2 \centr_{\sample_i}(\nodeset) \ln\bigl( \frac{5}{\delta_i} \bigr) }{ (1 - \frac{1}{e}) m_i } } + \frac{\ln\bigl( \frac{5}{\delta_i} \bigr)}{m_i}$\; \label{alg:xidef}
  $\eta \gets $ as in Equation \eqref{eq:epsrade} replacing $\delta$ by $\delta_i$ and $\sample$ by $\sample_i$\; \label{alg:eps}
  \lIf{$ (1-\frac{1}{e}) \sqpars{ ( 1 \! - \! \frac{1}{e} \! - \! \varepsilon ) \xi + \eta } \leq \varepsilon \centr_{\sample_i}(\nodeset) $}{\textbf{return} $\nodeset$} \label{alg:stopcond}
  }
  \caption{\algname}\label{algo:main}
\end{algorithm}

\algname\ takes in input a graph $G$, a space of hyper-edges $\hyperg$, and the parameters $k,t,\varepsilon,$ and $\delta$, with the goal of computing an approximation $\nodeset$ of $\nodeset^*$ such that $\centr( \nodeset ) \geq (1-1/e-\varepsilon) \centr(\nodeset^*)$ with probability $\geq 1 - \delta$. 

As a progressive sampling-based algorithm, \algname\ works in iterations.
In each iteration $i$, \algname\ considers a random sample $\sample_i$ of size $m_i$, and tries to establish if it is possible to compute an high-quality approximation of $\nodeset^*$ from $\sample_i$. 
To verify whether this holds, \algname\ checks a suitable stopping condition (defined below), returning the current solution when true. 

In line \ref{alg:init}, the algorithm initializes the variables $i$ to $0$, $\sample_i$ to the empty set, and $m_i = |\sample_i|$ to $0$. 
Then, the iterations of the algorithm start at line \ref{alg:startiters}. 
At the beginning of each iteration, $i$ is increased by $1$, and the confidence parameter $\delta_i$ is set to $\delta/2^i$. 
This correction is needed to prove the guarantees of the algorithm, i.e., to make sure that the probabilistic bounds computed at iteration $i$ holds with probability $\geq 1 - \delta_i , \forall i$, so that \algname\ is correct with probability $\geq 1-\delta$ when stopping at any iteration $i$ (more details in the proof of Proposition~\ref{prop:main}). 
Then, the number of samples $m_i$ to consider is given by a procedure \texttt{samplingSchedule} (line \ref{alg:schedule}). 
This procedure can be implemented in several ways. 
For example, the geometric progression defining $m_i = \alpha \cdot m_{i-1}$ for some $\alpha > 1$ is simple but considered to be optimal~\cite{provost1999efficient}. 
We combine this geometric progression (with $\alpha = 1.2$) with an adaptive schedule that tries to guess the next sample size using an optimistic target for the SD bound $\eta$; 
in fact, note that by setting $\rade=0$ and $\nu = 1$ in Thm. \ref{thm:bound_dev}, solving the equality \eqref{eq:epsrade} for $m$ gives a lower bound to the minimum number $m$ of samples to obtain a sufficiently small $\eta$. 

After getting the number of samples $m_i$, the algorithm generates the new sample $\sample_i$ composing $\sample_{i-1}$ with new $m_i - m_{i-1}$ random hyper-edges, taken from $\hyperg$ by the procedure \texttt{sampleHEs($\hyperg , m_i - m_{i-1}$)}. 
We describe the implementation of \texttt{sampleHEs} that we use for the set \bc\ in Section \ref{sec:experiments}. 
Then, \algname\ calls the procedure \texttt{greedyCover($k , \sample_i$)}, that finds a subset of nodes of size $\leq k$ covering many hyper-edges of $\sample_i$ with the greedy procedure introduced in Section~\ref{sec:prelimsetcentr}, returning the identified set $\nodeset$. 

In lines \ref{alg:amcerastart}-\ref{alg:amceraend}, \algname\ computes the AMCERA (Section~\ref{sec:boundsupdevrade}). 
To do so, it generates a matrix $\vsigma$ of size 
$t \times m_i$ by sampling $t m_i$ Rademacher random variables (Section~\ref{sec:rad_ave}) using the procedure \texttt{sampleRadeVars($t \times m_i$)} (line \ref{alg:amcerastart}). 
Then, for each vertex $u \in V$ and each trial of index $j \in [1,t]$, it populates the values of $r_j^u$ as the weighted average of the function $f_u$ on $\sample_i$ w.r.t. the entries of the matrix $\vsigma$ (line \ref{alg:amceraupdate}).
Using these values, \algname\ computes $r_j$ as the supremum sum of at most $k$ entries of the set $\{ r_j^u : u \in V \}$ (line \ref{alg:amceramaxj}). 

\algname\ obtains the AMCERA in line \ref{alg:amceraend} as the average of the $t$ values of $r_j$. 
We remark that it is not needed to generate the entire matrix $\vsigma$ at each iteration of the algorithm, and to compute $r_j^u$ for all $u \in V$, as done in the simplified procedure described above; 
it is sufficient to extend $\vsigma$ with $m_i - m_{i-1}$ new columns, and to update the values of $r_j^u$ \emph{incrementally} as each sample is generated and added to $\sample_i$. 
All these operations can be done efficiently, as proven by Lemma~\ref{amceraefficient}. 

In line \ref{alg:xidef}, the algorithm defines $\xi$, an accuracy parameter that quantifies how far is the estimate $\centr_{\sample_i}(\nodeset)$ from the optimal centrality $\centr(\nodeset^*)$: it holds 
$\centr(\nodeset^*) \leq  \frac{\centr_{\sample_i}(\nodeset)}{1-1/e} + \xi$ 
with probability $\geq 1 - \delta_i/5$ (see Theorem~\ref{thm:selfboundingupperbound} and the value of $\nu$ in the statement and proof of Theorem~\ref{thm:bound_dev}). 
After obtaining the AMCERA, \algname\ computes (line \ref{alg:eps}) an upper bound $\eta$ to the supremum deviation $\sd(\F, \sample_i)$ on the sample $\sample_i$ with confidence $\delta_i$, using the result of Theorem~\ref{thm:bound_dev}. 

Finally, \algname\ checks a stopping condition (line \ref{alg:stopcond}); if true, it means that $\nodeset$ provides an $(1-1/e-\varepsilon)$-approximation of the optimum $\nodeset^*$ with enough confidence, therefore \algname\ returns $\nodeset$ in output. 
\algname\ uses a new, improved stopping condition to establish its theoretical guarantees. 
As discussed previously, the condition $\supdev \leq \varepsilon \centr(\nodeset^*) / 2 $ proposed by HEDGE~\cite{mahmoody2016scalable} is sufficient to guarantee that 
$\centr(\nodeset) \geq (1-1/e - \varepsilon) \centr(\nodeset^*)$. 
However, this condition is not necessary. 
\algname\ combines the upper bound 
$\centr(\nodeset^*) \leq \nu = \centr_{\sample_{i}}(\nodeset)(1-1/e)^{-1} + \xi$ 
to $\centr(\nodeset^*)$  
with the lower bound 
$\centr_{\sample_{i}}(\nodeset) - \eta \leq \centr(\nodeset)$ 
to $\centr(\nodeset)$, rather than using $\eta$ twice (for both the upper bound to $\centr(\nodeset^*)$ and the lower bound to $\centr(\nodeset)$, see the proof of Proposition~\ref{prop:main}). 
Note that, in general, we expect $ \xi \ll \eta $, as $\xi$ is \emph{independent} of $k$ and any property of $\F$ (e.g., its complexity), but sharply scales with the empirical estimate $\centr_{\sample_{i}}(\nodeset)$ of $\centr(\nodeset^*)$.

We now prove the correctness of \algname, and 
a bound to its time complexity for \emph{incrementally} processing the random samples and for computing the AMCERA. (the proofs are in the Appendix.)
\begin{proposition}
\label{prop:main}
Let $\nodeset$ be the output of \algname. 
With probability $\ge 1 - \delta$, it holds $\centr(\nodeset) \geq ( 1 - \frac{1}{e} - \varepsilon ) \centr(\nodeset^*)$. 
\end{proposition}
\begin{lemma}
\label{amceraefficient} 
Define $b_{\sample}$ as $\sup_{h \in \sample} |h| \leq b_{\sample} \leq n$.
When \algname\ stops at iteration $i \geq 1$ after processing $m_i$ samples, 
it computes $\trade^{t}_{m_j}(\F, \sample_j, \vsigma)$, for all $j \in [1,i]$, in time
$\BOi{ (b_{\sample} m_i + k)t \log(n)}$. 
\end{lemma}
Note that the time bound above considers the total number of operations performed by \algname\ for computing $\trade^{t}_{m_j}(\F, \sample_j, \vsigma)$, summed over \emph{all} iterations $j\leq i$. 
Notably, the time needed by \algname\ to compute the AMCERA is linear in the
final sample size $m_i$ and logarithmic in the number of nodes $n$ of the graph. 
Furthermore, the parameter $b_{\sample}$ is typically very small, as we discuss in Section~\ref{sec:samplecomplexity} (e.g., for the \bc\ is at most the \emph{vertex diameter} $B$ of the graph~\cite{RiondatoK15}).
We obtain the total running time of \algname\ by summing 
the bound from Lemma~\ref{amceraefficient} and $\BOi{ T_\hyperg m_i + n \log(n) }$, 
where $T_\hyperg$ is the time to sample one hyper-edge from $\hyperg$, 
and the second term is 
for the greedy algorithm \texttt{greedyCover} (using the implementation from~\cite{borgs2014maximizing}, and assuming the final iteration index $i \in \BOi{1}$, that is always the case in practice). 
In our experimental evaluation (Section~\ref{sec:experiments}) we show that 
\algname\ is a very practical and scalable approach, 
since the time required to compute the AMCERA is negligible w.r.t. to the time to generate the random samples and \texttt{greedyCover}.  
 
Finally, we remark that \algname\ applies directly to approximate set centrality functions that are submodular and can be defined as averages over some hyper-edge space $\hyperg$ (see Section~\ref{sec:prelimsetcentr}), such as $\kappa$-path, coverage, and triangle centralities, as \algname\ is oblivious to the procedure \texttt{sampleHEs} to sample hyper-edges.

\subsection{Sample Complexity bounds}
\label{sec:samplecomplexity}
In this Section we present a new bound (Theorem~\ref{thm:vcsamplecomplexity}) to the sufficient number of random samples to obtain an $( 1 - \frac{1}{e} - \varepsilon )$-approximation of $\nodeset^*$. 
(due to space constraints, we defer all proofs of this Section to the Appendix.)

We first define the range space associated to set centralities and its VC-dimension, and remand to~\cite{ShalevSBD14,mitzenmacher2017probability,mohri2018foundations} for a more complete introduction to the topic. 
Define the range space $Q_k = (\hyperg , R_k)$, where $\hyperg$ is a family of hyper-edges, 
and 
\begin{align*}
R_k = \brpars{ \brpars{ h : h \in \hyperg , f_\nodeset(h) = 1 } : \nodeset \subseteq V , |\nodeset| \leq k }
\end{align*}
is a family of subsets of $\hyperg$, indexed by the indicator functions $f_\nodeset \in \F$ associated to subsets of $V$ with size at most $k$. 
For any $\sample \subseteq \hyperg$ of size $m$, the projection $P(\sample , R_k)$ of $R_k$ on $\sample$ is 
\[
P(\sample , R_k) = \brpars{ \sample \cap r : r \in R_k } .
\] 
We say that $\sample$ is \emph{shattered} by $R_k$ if 
$|P(\sample , R_k)| = 2^m $. 
The \emph{VC-dimension}~\cite{Vapnik:1971aa} $VC(Q_k)$ of the range space $Q_k$ is the maximum cardinality $m$ of a set $\sample$ that is shattered by $R_k$. 
Under this framework, we are ready to state the following sample complexity bound. 
\begin{theorem}
\label{thm:vcsamplecomplexity}
Let $\sample$ be a sample of $m$ hyper-edges taken i.i.d. uniformly from $\hyperg$. For $k \geq 1$, let $\nodeset$ be the output of \texttt{greedyCover($k,\sample$)}. 
Let $d_k$ such that $VC(Q_k)\leq d_k$. 
For $\delta \in (0,1)$ and $\varepsilon \in (0 , 1-\frac{1}{e})$, 
if 
\begin{align}
m \in \BO{ \frac{ d_k \log \bigl(  \frac{1}{\centr(\nodeset^*)} \bigr) + \log \bigl( \frac{1}{\delta} \bigr) }{\varepsilon^2 \centr(\nodeset^*)} } ,  \label{eq:samplecompl}
\end{align} 
then, it holds $\centr(\nodeset) \geq ( 1 - \frac{1}{e} - \varepsilon ) \centr(\nodeset^*)$ with probability $\geq 1 - \delta$. 
\end{theorem}
We may observe that the main consequence of Theorem \ref{thm:vcsamplecomplexity} is that the number of samples required to obtain a $( 1 - \frac{1}{e} - \varepsilon )$-approximation of $\nodeset^*$ with probability $\geq 1-\delta$
does not necessarily depend on the size of the graph $G$ (e.g., the number of nodes $n = |V|$), but rather scales with the VC-dimension $VC(Q_k)$. 
Moreover, note that it always holds $VC(Q_k) \leq \log_2 |R_k| \leq \lfloor k \log_2 (n) \rfloor$ (in other words, $VC(Q_k)$ cannot be larger than a na\"ive application of the union bound), but the gap between $VC(Q_k)$ and $\log_2 |R_k|$ (and $\lfloor k \log_2 (n) \rfloor$) can be arbitrary large~\cite{ShalevSBD14}. 
In order to characterize $VC(Q_k)$, we give an upper bound to it as a function of $VC(Q_1)$.
\begin{lemma}
\label{thm:lemmavckbound}
The VC-dimension $VC(Q_k)$ of the range space $Q_k$ is $VC(Q_k) \leq 2 VC(Q_1) k \log_2(3k) $.
\end{lemma}
Lemma~\ref{thm:lemmavckbound} conveniently involves $VC(Q_1)$, a quantity that is often much easier to bound than $VC(Q_k)$. In fact, for the case of the \bc, it holds $VC(Q_1) \leq \lfloor \log_2(2B) \rfloor$, where $B$ is (an upper bound to) the \emph{vertex diameter} of the graph $G$~\cite{RiondatoK15}, i.e., the maximum number of distinct nodes that are internal to a shortest path of $G$. 
Our next result generalizes such bound to general set centralities, in which we assume that the maximum number of distinct nodes in an hyper-edge $h$ that we can sample from $\hyperg$ is upper bounded by a constant $b$. 
\begin{lemma}
\label{thm:vcboundgeneral}
If $\sup_{h \in \hyperg}{ |h| } \leq b$, then 
$VC(Q_1) \leq d_1 = \lfloor \log_2(2b) \rfloor $.
\end{lemma}
For instance, when sampling simple paths of length at most $\kappa \geq 1$ for estimating the $\kappa$-path centrality it holds $b \leq \kappa$. 
For the triangle centrality, $b=3$. 
Furthermore, note that $b \leq n$ as $h \subseteq V, \forall h \in \hyperg$. 
In light of Theorem~\ref{thm:vcsamplecomplexity} and Lemma~\ref{thm:lemmavckbound},
and under the assumption that $\centr(\nodeset^*) \in \BTi{1}$, as common in real-world networks~\cite{yoshida2014almost,mahmoody2016scalable}, it is immediate to observe that the sufficient sample size $m$ to achieve an $( 1 - \frac{1}{e} - \varepsilon )$-approximation is 
\begin{align}
m \in \BO{ \frac{d_k  + \log(1/\delta) }{ \varepsilon^2  } } ,  m \in \BO{ \frac{d_1 k \log(k)  + \log(1/\delta) }{ \varepsilon^2  } } ,  \label{eq:samplecomplexpl}
\end{align}
where $d_1$ can be easily bounded with Lemma~\ref{thm:vcboundgeneral}. 

Remarkably, 
the bounds \eqref{eq:samplecomplexpl} are never worse than \eqref{eq:unionboundsamples} 
(from the union bound, as $d_k \leq \lfloor k \log_2(n) \rfloor$), 
while they are significantly refined in many realistic scenarios. 
An interesting example is given by \emph{small-world} networks~\cite{watts1998collective}, that have vertex diameter $B \in \BO{\log n}$; 
in such cases, both \eqref{eq:samplecompl} and \eqref{eq:samplecomplexpl} provide an \emph{exponential improvement} on the dependence on $n$ for set betweenness centrality approximations (as $b \leq B \in \BO{\log n}$, and $d_1 \in \BO{\log \log n}$). 
Furthermore, Lemma \ref{thm:lemmavckbound} and \ref{thm:vcboundgeneral} imply that, when $d_1$ or $d_k$ are constants (not depending on $n$), the sample complexity bounds of \eqref{eq:samplecomplexpl} are completely \emph{independent} of the graph size $n$. 
This is the case, for example, when estimating the $\kappa$-path centrality and $\kappa$ is a small constant (as it is typically in applications), where $d_1 = \lfloor \log_2(\kappa) \rfloor$, 
or for undirected graphs with unique shortest paths between all pairs of vertices\footnote{In road networks, this feature is often enforced~\cite{geisberger2008better}.}, where $d_1 = 3$ (see Lemma 2 of~\cite{RiondatoK15}). 
These observations confirm that in many situations standard techniques do not capture the correct underlying complexity of the set centrality approximation task, yielding loose guarantees. 

We remark that, while these results improve state-of-the-art bounds from a theoretical perspective, they also have a practical impact as they can be embedded in the progressive sampling strategy of \algname\ by setting an upper limit to the number of samples to consider.

\section{Experimental Evaluation}
\label{sec:experiments}

\begin{figure*}[ht]
\centering
\begin{subfigure}{.75\textwidth}
  \centering
  \includegraphics[width=\textwidth]{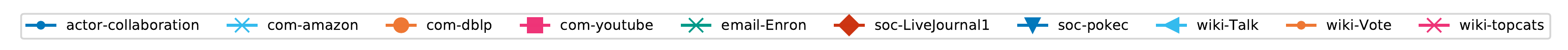}
\end{subfigure} \\
\begin{subfigure}{.24\textwidth}
  \centering
  \includegraphics[width=\textwidth]{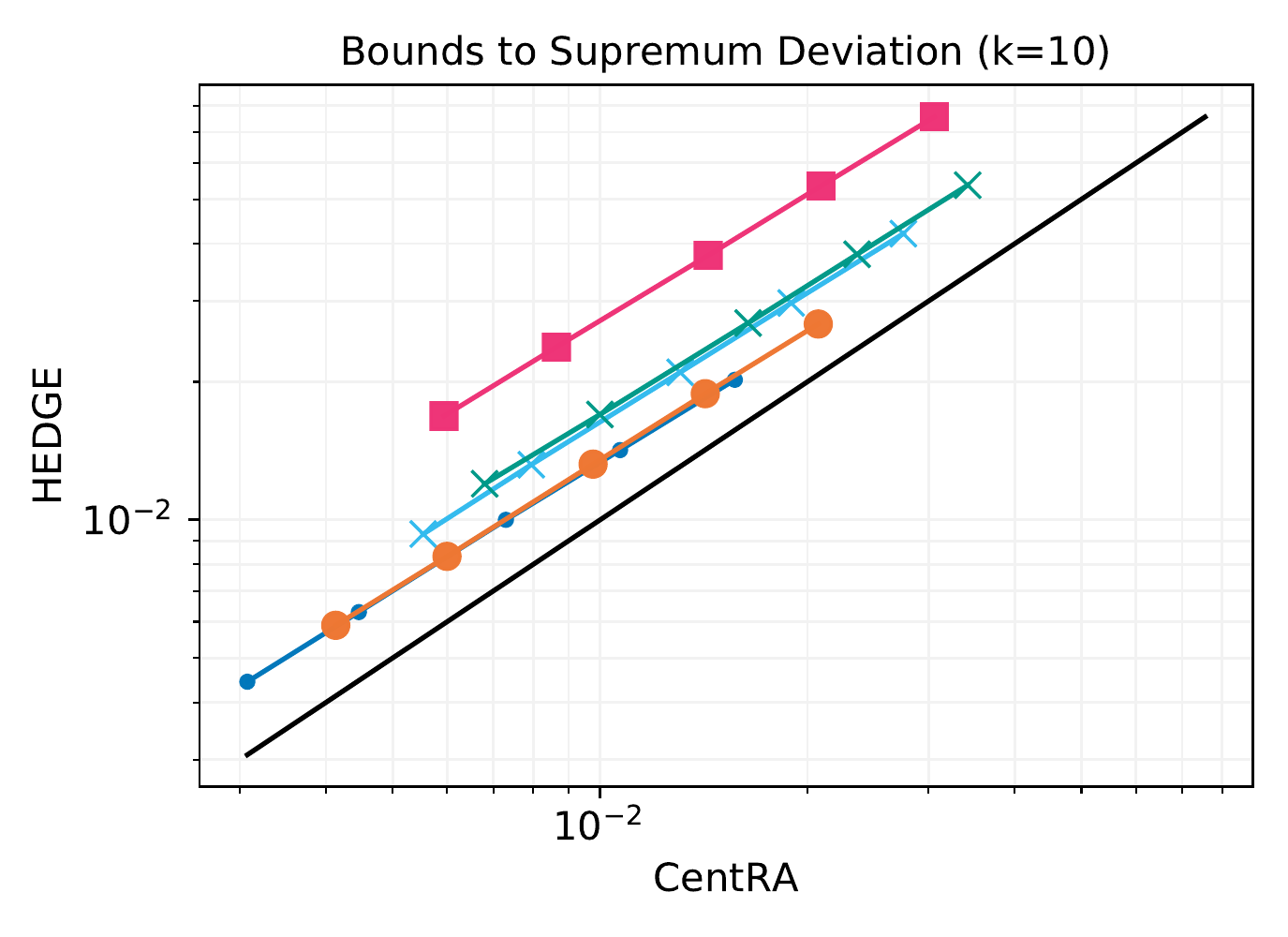}
  \caption{}
\end{subfigure}
\begin{subfigure}{.24\textwidth}
  \centering
  \includegraphics[width=\textwidth]{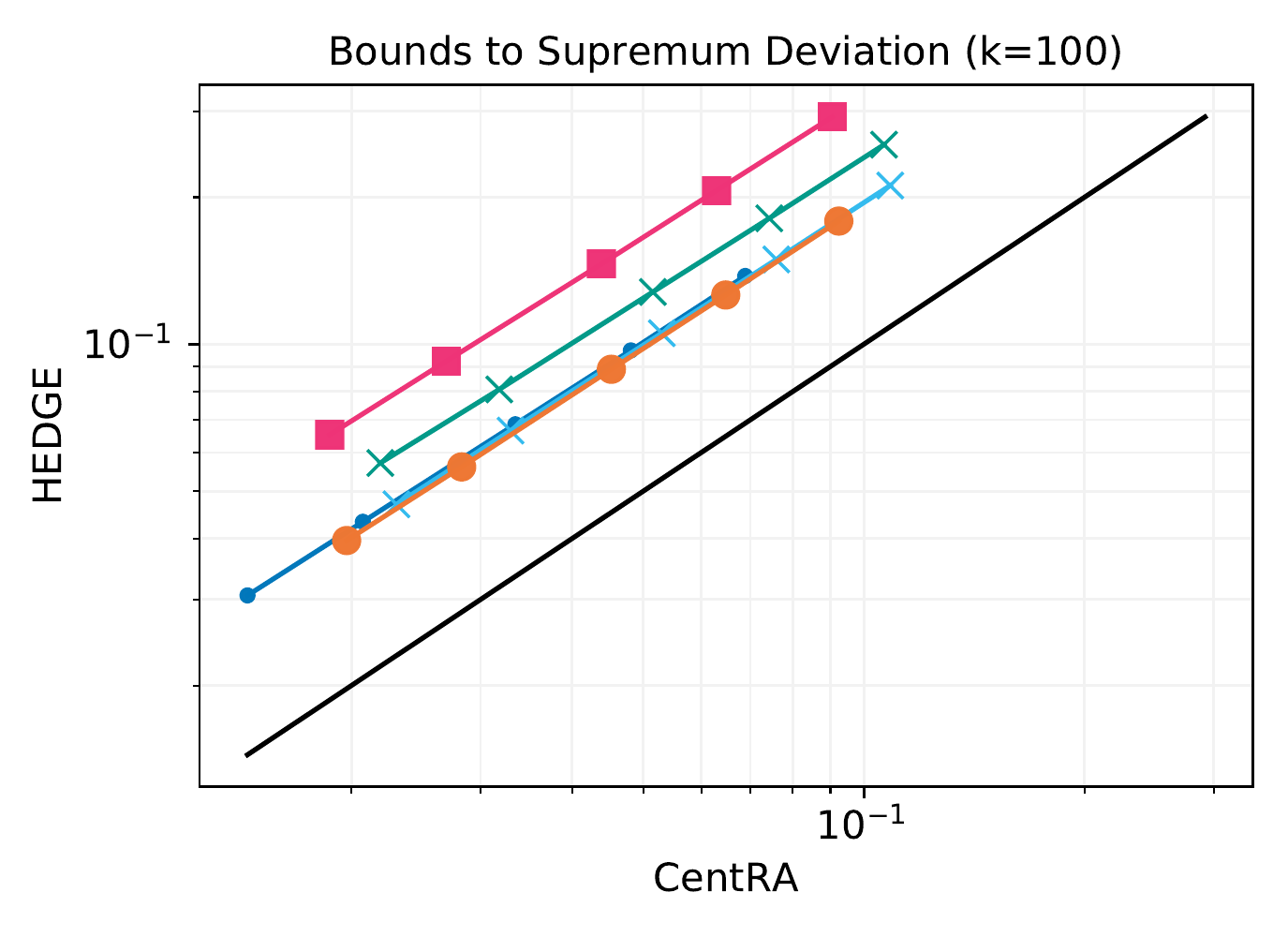}
  \caption{}
\end{subfigure}
\begin{subfigure}{.24\textwidth}
  \centering
  \includegraphics[width=\textwidth]{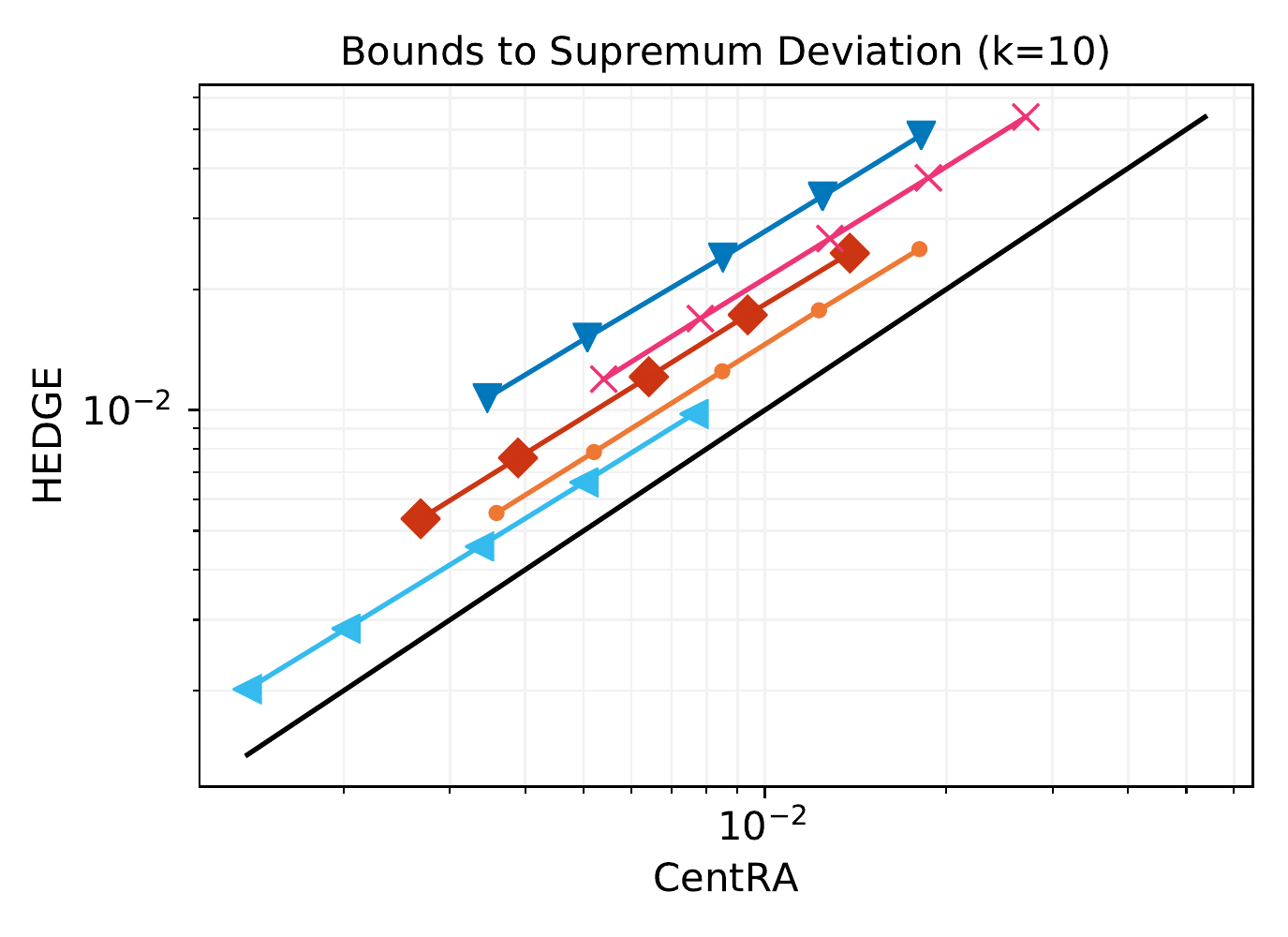}
  \caption{}
\end{subfigure}
\begin{subfigure}{.24\textwidth}
  \centering
  \includegraphics[width=\textwidth]{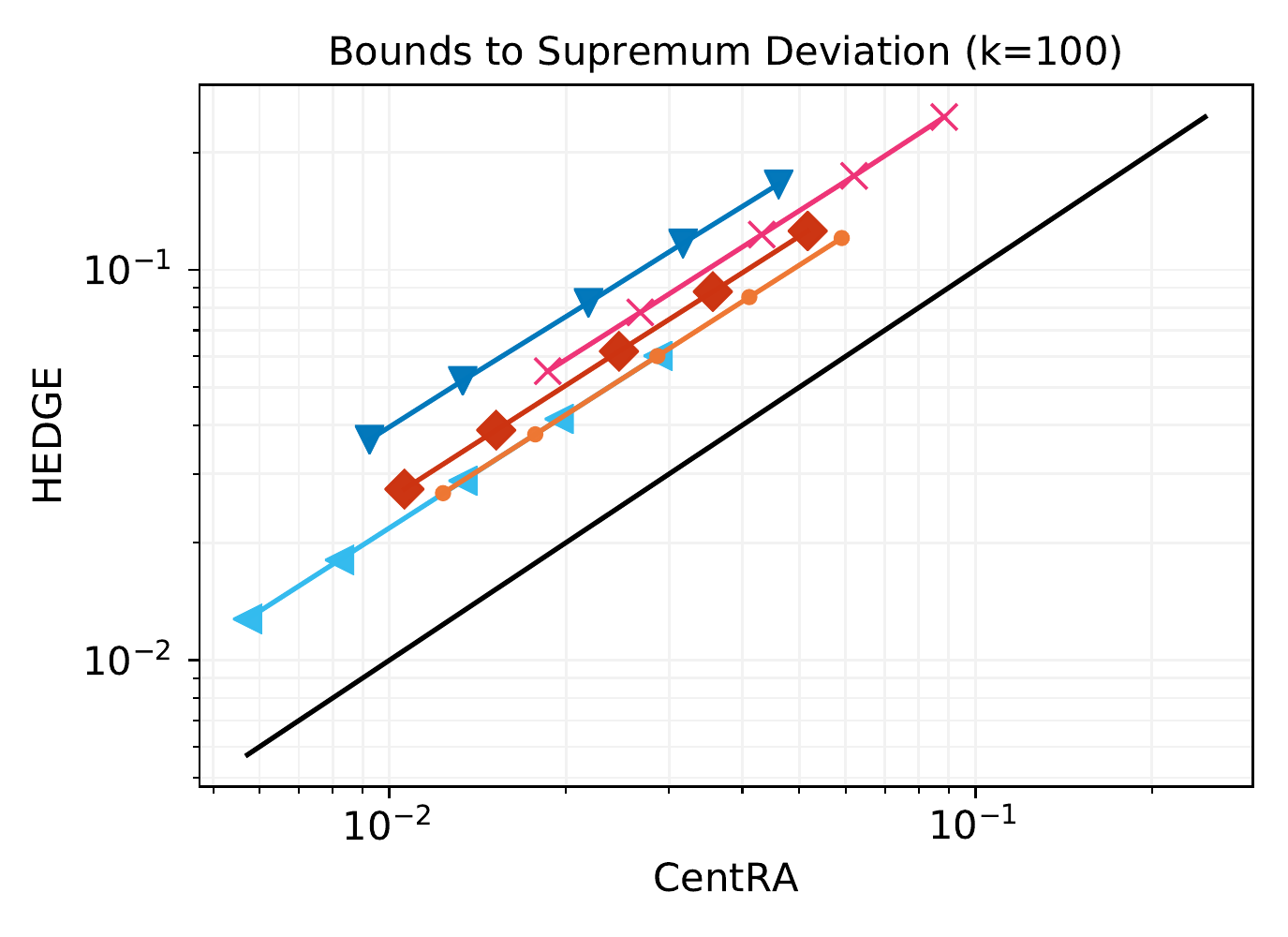}
  \caption{}
\end{subfigure}
\caption{
Comparison between the 
bounds to the Supremum Deviation $\supdev$ obtained by HEDGE ($y$ axes, based on the union bound) and \algname\ ($x$ axes, Section~\ref{sec:algorithm}) on samples of size $m \in \{ 5 \cdot 10^4 , 10^5 , 2 \cdot 10^5 , 5 \cdot 10^5 , 10^6 \}$, for $k=10$ and $k=100$ (other values in Figure \ref{fig:fixedmappx}). 
Figures (a)-(b): undirected graphs.
Figures (c)-(d): directed graphs. 
Each point corresponds to a different value of $m$ (the same for both algorithms).
The black diagonal line corresponds to $y = x$.
}
\Description{This Figure compares the bounds to the supremum deviation obtained by CentRA and HEDGE on random samples of different size, for different values of k. 
The plots show that CentRA obtains smaller bounds to the supremum deviation in all cases, confirming that data-dependent bounds are more accurate than standard techniques.}
\label{fig:fixedm}
\end{figure*}

\begin{figure*}[ht]
\centering
\begin{subfigure}{.75\textwidth}
  \centering
  \includegraphics[width=\textwidth]{./figures/bounds-fixed-m-legend.pdf}
\end{subfigure} \\
\begin{subfigure}{.24\textwidth}
  \centering
  \includegraphics[width=\textwidth]{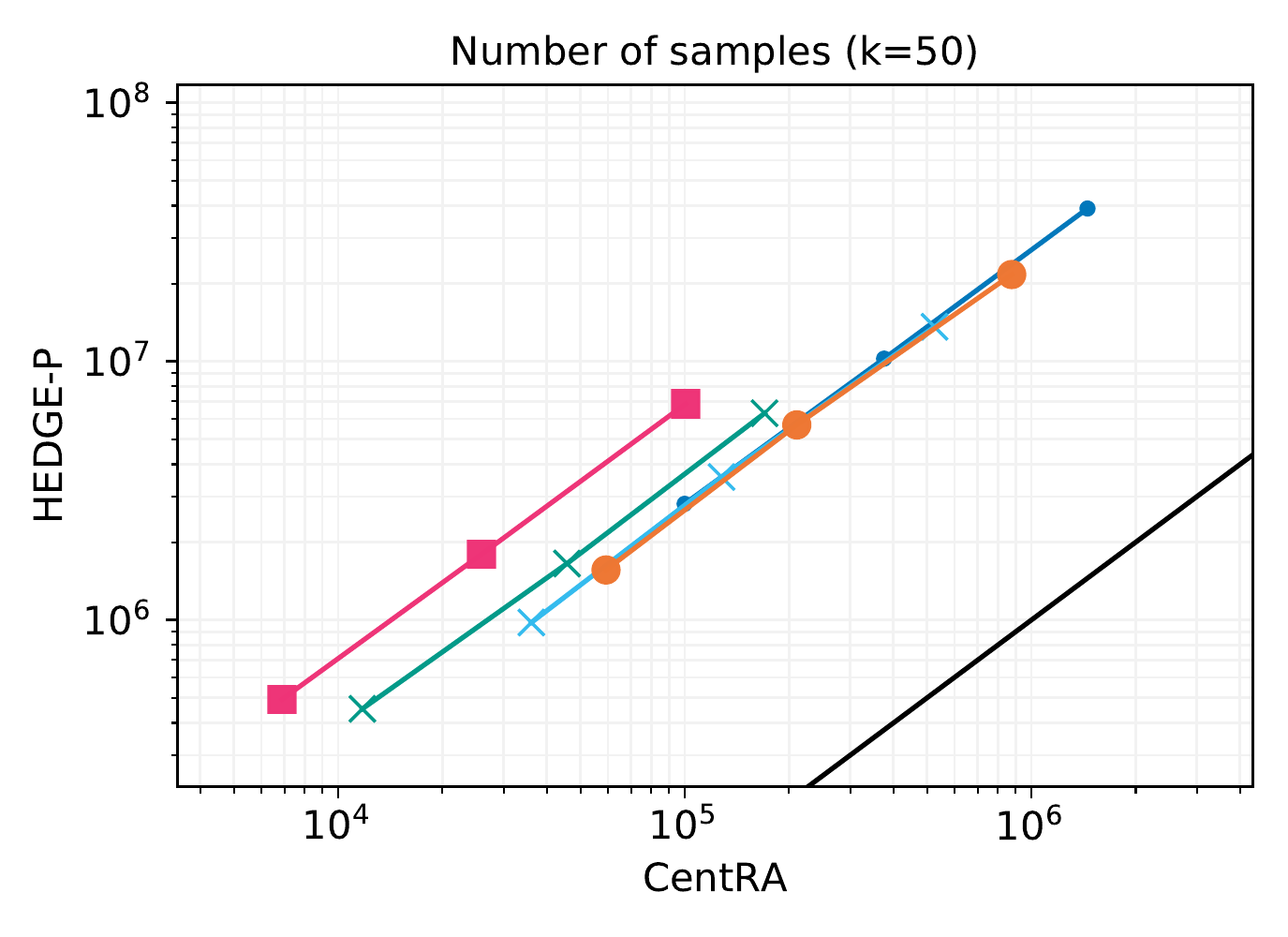}
  \caption{}
\end{subfigure}
\begin{subfigure}{.24\textwidth}
  \centering
  \includegraphics[width=\textwidth]{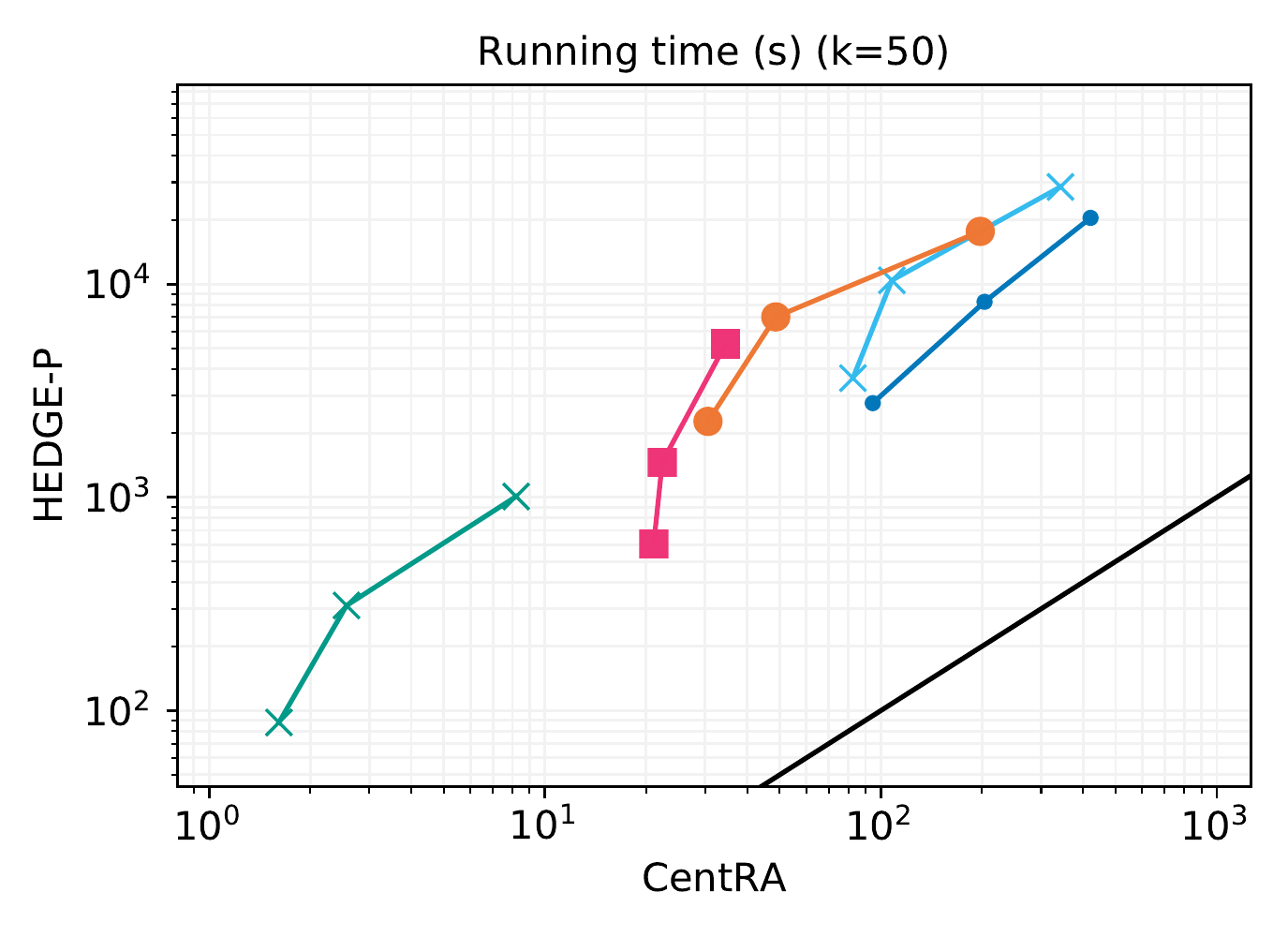}
  \caption{}
\end{subfigure}
\begin{subfigure}{.24\textwidth}
  \centering
  \includegraphics[width=\textwidth]{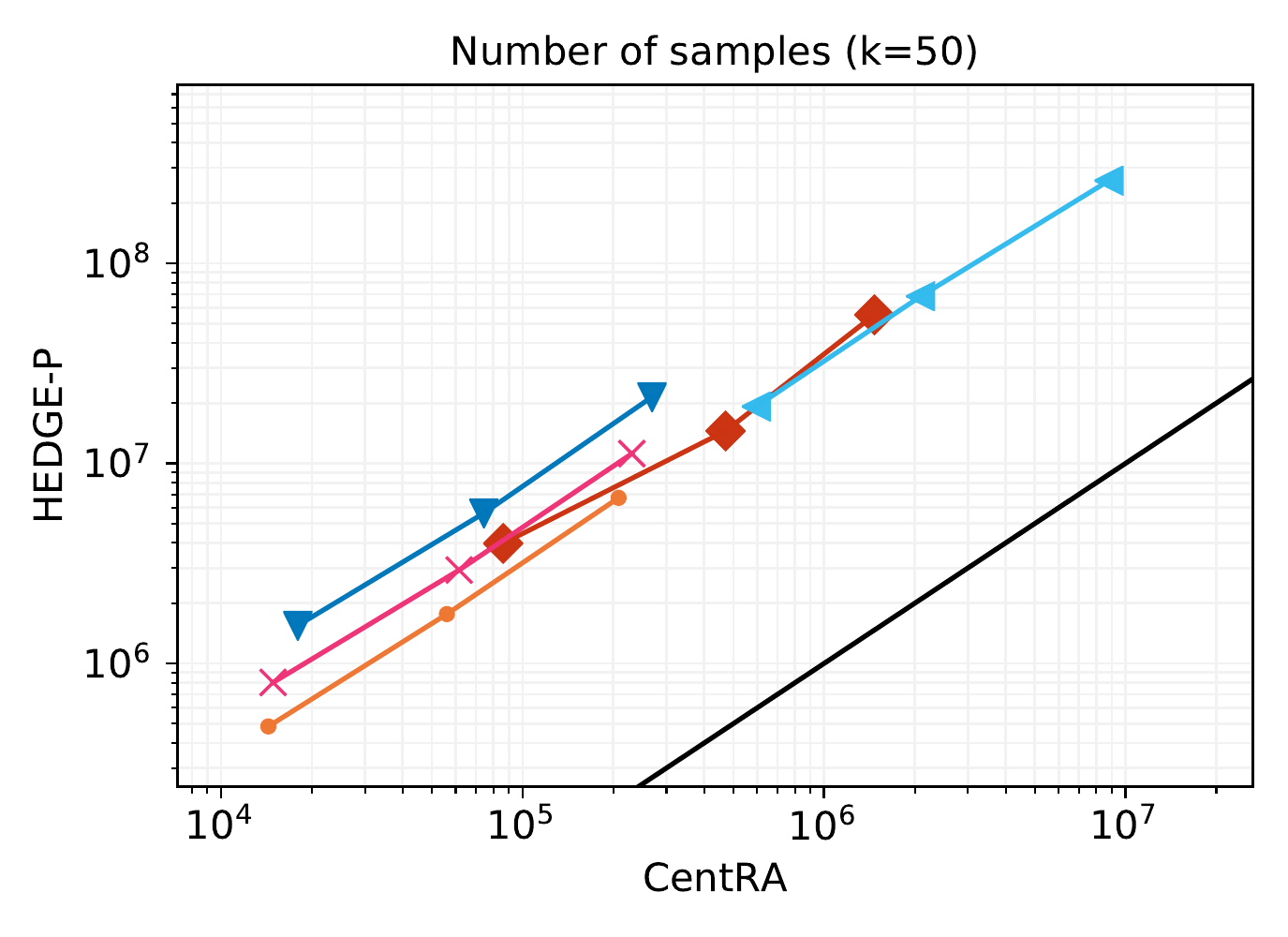}
  \caption{}
\end{subfigure}
\begin{subfigure}{.24\textwidth}
  \centering
  \includegraphics[width=\textwidth]{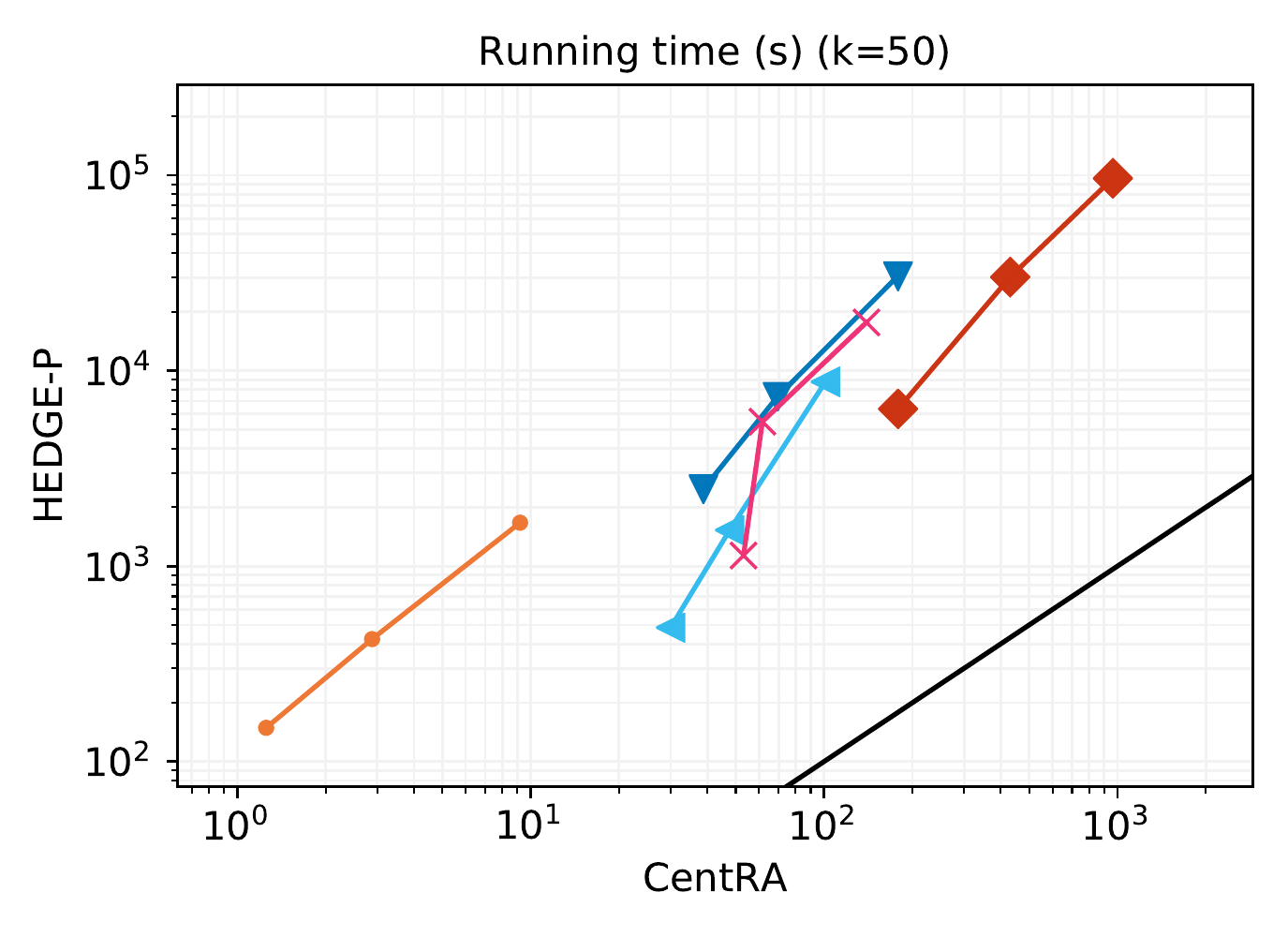}
  \caption{}
\end{subfigure}
\caption{
Comparison between the 
number of samples and running times (in seconds) to obtain an $(1-1/e - \varepsilon)$-approximation of the optimal set centrality $\nodeset^*$ using \algname\ ($x$ axes) and HEDGE-P ($y$ axes) for $\varepsilon \in \{ 0.2 , 0.1 , 0.05 \}$ and $k=50$ (other values of $k$ shown in Appendix). 
Figures (a)-(b): samples and times for undirected graphs.
Figures (c)-(d): directed graphs.
Each point corresponds to a different value of $\varepsilon$ (the same for both algorithms).
The black diagonal line is at $y = x$. 
}
\label{fig:progsampling}
\Description{This Figure compares the running time and number of samples of CentRA and HEDGE-P to compute an high quality approximation of the most central set of nodes. 
The plots clearly show that, in all cases, CentRA converges after considering a much smaller number of random samples, and at a fraction of the time required by HEDGE-P, obtaining an improvement of up to two orders of magnitude.}
\end{figure*}
In this Section we present our experiments. 
The goals of our experimental evaluation are: 
1) to test the practical impact of the number $t$ of Monte Carlo trials for estimating the AMCERA, both in terms of running time and bound on the SD $\supdev$;
2) to compare the data-dependent bounds computed by \algname\ with state-of-the-art techniques; 
3) to show the impact, both in terms of running time and number of random samples, that a finer characterization of the SD has to the task of centrality maximization. 

\emph{Experimental Setup.} 
We implemented \algname\ in \texttt{C++}. The code and the scripts to reproduce all experiments are available 
at
\url{https://github.com/leonardopellegrina/CentRA}. 
All the code was compiled with \texttt{GCC} 8 and run on a machine with 2.30 GHz Intel Xeon CPU, 512 GB of RAM, on Ubuntu 20.04. 

In our experiments we focus on the set betweenness centrality, and remand testing other centralities to the full version of this work. 
To sample random shortest paths from $\hyperg$, the procedure \texttt{sampleHEs} chooses a starting node $u$ and final node $v$ uniformly at random, and then picks a random shortest paths from $u$ to $v$, using the balanced bidirectional BFS~\cite{borassi2019kadabra} (in which two BFSs are expanded from both $u$ and $v$). 
If $v$ is not reachable from $u$, the procedure returns the empty set.  
We compare \algname\ with HEDGE~\cite{mahmoody2016scalable}, the state-of-the-art method for approximate centrality maximization. 
As the implementation of HEDGE is not publicly available, we implemented it is as a modification of \algname.  
Note that we do not compare with the exact approach (that computes centralities exactly), as HEDGE~\cite{mahmoody2016scalable} already outperforms it on small graphs, and because the exact approach does not conclude in reasonable time on the large graphs we considered\footnote{Computing the centralities of \emph{individual nodes} on large graphs already requires significant computation (e.g.,~\cite{borassi2019kadabra} reports that $1$ week is required on a $40$ cores cluster for graphs of size comparable to the ones we tested). In addition, it is infeasible to update these centralities at each of the $k$ iterations of the greedy algorithm.}, and to~\cite{yoshida2014almost} (due to some flaws in the analysis pointed out by~\cite{mahmoody2016scalable}). 
Note that \algname\ does not use the upper bounds proved in Section~\ref{sec:samplecomplexity} to limit its number of samples, so to directly test its data-dependent bounds. 
However, due to space constraints, we defer to the Appendix (Section \ref{sec:additionalexperiments}) analogous comparisons with a variant of \algname, that we call \algname-VC, that uses VC-dimension bounds (Section \ref{sec:samplecomplexity}) instead of the MCERA (Section \ref{sec:boundsupdevrade}). 
We briefly comment these results in the following paragraphs.
We repeat all experiments $10$ times and report averages (standard deviations not shown since the experiments were extremely stable). 

~\emph{Graphs.}
We tested \algname\ on $5$ undirected and $5$ directed real-world graphs from SNAP\footnote{\url{http://snap.stanford.edu/data/index.html}} and KONECT\footnote{\url{http://konect.cc/networks/}}.
The characteristics of the graphs are described in detail in Table~\ref{tab:graphs} (in the Appendix).

\emph{Impact of \algname\ parameters.}
In our first experiments, we evaluate the practical impact of the number $t$ of Monte Carlo trials used by \algname\ to derive data-dependent bounds to the SD via the AMCERA (Section~\ref{sec:boundsupdevrade}). 
We run \algname, letting it process a fixed number of samples $m = 5 \cdot 10^5$ using $t \in \{ 1 , 10 , 50 , 10^2 , 2.5 \cdot 10^2 , 5 \cdot 10^2 \}$, measuring the upper bound $\eta$ to the SD (Theorem~\ref{thm:bound_dev}) and its running time. 
We focus on the two largest and two smallest graphs of the ones we considered, two undirected and two directed. 
In all cases, we fix $\delta = 0.05$ (note that we do not vary $\delta$ as its impact is negligible, due to the use of exponential tails bounds, see Thm.~\ref{thm:bound_dev}).  
In Figure \ref{fig:paramscentra} (a) we show the bounds to the SD, and in Figure \ref{fig:paramscentra} (b) the running time (Figures in the Appendix). 
We conclude that, when using more than $10^2$ trials, the improvements in terms of deviation bound is negligible. 
This confirms that a small number $t$ of Monte Carlo trials is sufficient to estimate the ERA accurately. 
More importantly, the impact to the running time is negligible when increasing $t$ on large graphs, and not predominant for smaller graphs (Figure \ref{fig:paramscentra} (b)). 
This means that the most expensive operations performed by \algname\ are sampling shortest paths and processing the sample with the greedy algorithm, while the AMCERA is computed efficiently. 
Based on these observations, we fix $t  = 10^2$ for all experiments.

\emph{Bounds to the Supremum Deviation.}
We now compare the bounds to the Supremum Deviation $\supdev$ obtained by HEDGE (using the union bound) with our analysis of Section~\ref{sec:boundsupdevrade} based on the AMCERA, at the core of \algname. 
In this way, we directly assess the trade-off in terms of accuracy and sample size $m$ of both algorithms.  
To do so, we sample for both algorithms the same number $m$ of hyper-edges and compute $\eta \geq \supdev$ for \algname\ (using Thm.~\ref{thm:bound_dev}) and combine the Chernoff bound with the union bound (following the analysis of Lemma~2 of~\cite{mahmoody2016scalable}) for HEDGE, obtaining $\eta^{UB} \geq \supdev$. 
We considered $m \in \{ 5 \cdot 10^4 , 10^5 , 2 \cdot 10^5 , 5 \cdot 10^5 , 10^6 \}$, $k \in \{ 10 , 25 , 50 , 75 , 100 \}$, and $\delta = 0.05$ for both methods. 

Figure~\ref{fig:fixedm} shows the results for these experiments. 
Figures~\ref{fig:fixedm}~(a)-(b) show the bounds to the SD as functions of $m$ for $k=10$ and $k=100$ for HEDGE ($y$ axes) and \algname\ ($x$ axes) for undirected graphs (to ease readability, directed are in Figures~\ref{fig:fixedm}~(c)-(d)). 
Note that each point in the plots corresponds to a value of $m$, while the diagonal black line is at $y=x$. 
Figures for other values of $k$ are very similar and showed in the Appendix (Figure \ref{fig:fixedmappx}) for space constraints. 
In Figure \ref{fig:estcentr} we show the estimated centrality of the set $\nodeset$ returned by \algname\ as function of $k$, for $m=10^6$.

From these results, we can clearly see that \algname\ obtains a much smaller bound to the SD than HEDGE (as the points are above the black diagonal), uniformly for all graphs. 
This is a consequence of the sharp \emph{data-dependent} bounds derived in Section~\ref{sec:boundsupdevrade}, which significantly improve the loose guarantees given by the union bound. 
We observe that, for $k=10$, the upper bound $\eta^{UB}$ to the SD from the union bound is at least $1.25$ times higher than $\eta$ for all graphs, more than $1.5$ higher for $7$ graphs, and at least $2$ times higher for $3$ graphs. 
At $k=50$ (Figure \ref{fig:fixedmappx}), it is $1.5$ higher for all graphs, at lest $2$ times higher for $6$ graphs, and at least $3$ times higher for $2$ graphs. 
This gap is even larger for $k=100$. 
Therefore, $\eta$ is often \emph{a fraction} of the upper bound $\eta^{UB}$ from HEDGE. 
Interestingly, we observe this difference on both smaller and larger graphs. 
Furthermore, this gap grows with $k$, confirming that the data-dependent bounds of \algname\ scale much more nicely than the union bound w.r.t. $k$. 
To appreciate the magnitude of the improvement, note that both $\eta$ and $\eta^{UB}$ decrease proportionally with $\BTi{\sqrt{1/m}}$ (see the rates in Figures~\ref{fig:supdevvsmindividual}); 
this means that to make the upper bound to the SD $\alpha$ times smaller, we need a sample that is approximately $\alpha^2$ times larger. 
Therefore, we expect that obtaining a sharper bound to the SD should have a significant impact on the the sample size and, consequently, on the cost to analyze it. 
The following experiments confirm this intuition. 
Finally,
the results deferred to Section~\ref{sec:additionalexperiments} (comparing \algname, \algname-VC, and HEDGE) confirm our theoretical insights: 
the novel VC-dimension bounds (used by \algname-VC) are always equally or more accurate than standard bounds, while \algname\ is the overall superior approach, as it leverages sharp data-dependent techniques.

\emph{Approximate Centrality Maximization with Progressive Sampling.}
In the last set of experiments, our goal is to evaluate the performance of \algname\ to compute an $(1-1/e - \varepsilon)$-approximation of the optimum $\nodeset^*$ with progressive sampling. 
Note that HEDGE cannot be applied to this problem as is, since it requires knowledge of $\centr(\nodeset^*)$ to fix its sample size. 
Since \algname\ is the \emph{first} algorithm of this kind, we compare it to a \emph{new} variant of HEDGE, that we call HEDGE-P, that we implemented as a modification of \algname. 
HEDGE-P follows the same iterative scheme of Algorithm~\ref{algo:main}, but with crucial differences: 
instead of using the AMCERA to upper bound the SD, it uses $\eta^{UB}$ (described previously) from a union bound (Lemma 2 of~\cite{mahmoody2016scalable}). 
As the stopping condition, we follow the analysis of~\cite{mahmoody2016scalable}: 
HEDGE-P stops when it guarantees that $\supdev \leq \varepsilon \centr(\nodeset^*) / 2 = \eta^{UB} $, as this is a sufficient condition to obtain an $(1-1/e - \varepsilon)$-approximation (see Thm. 1 of~\cite{mahmoody2016scalable}). 
Instead, \algname\ leverages a refined stopping condition (line 14 of Alg. 1, Section \ref{sec:algorithm}).
We compare \algname\ to HEDGE-P in terms of the number of samples and running times required to converge, using $\varepsilon \in \{0.2 , 0.1, 0.05\}$ and $k \in \{10 , 50 , 100 \}$ (that are well representative of other values of $k$, as shown previously). 

Figure \ref{fig:progsampling} shows the results of these experiments for $k=50$ (other values in Figure \ref{fig:progsamplingappx}). 
\algname\ obtains an $(1-1/e - \varepsilon)$-approximation of the optimum $\nodeset^*$ using a \emph{fraction} of the samples and at a fraction of the time of HEDGE-P. 
For $k=10$, \algname\ needs a random sample that is one order of magnitude smaller than HEDGE-P. 
For $k=50$, the random samples required by \algname\ are at least $20$ times smaller than the samples needed by HEDGE-P. 
When $k=100$, the sample size reduction is close to $2$ orders of magnitude (Figure \ref{fig:progsamplingappx}). 
From these results it is clear that a better characterization of the SD 
and a refined stopping condition allow \algname\ to be a much more efficient algorithm for approximate centrality maximization in terms of sample sizes. 
As expected, we observed that reducing the sample sizes has a significant impact on the running times. 

In fact, we observed the running time to be essentially linear w.r.t. the sample size. 
For $k=10$, for $3$ graphs and $\varepsilon = 0.2$, \algname\ is at least $5$ times faster than HEDGE-P; for all other cases, \algname\ is one order of magnitude faster. 
For $k=50$, \algname\ finished in less than $1/20$ of the time needed by HEDGE-P, improving up to two orders of magnitude when $k = 100$ (Figure \ref{fig:progsamplingappx}). 
Both \algname\ and HEDGE-P conclude after a small number of iterations (always at most $3$), confirming that the adaptive progressive sampling schedule is very accurate. 
These observations confirm the efficiency of \algname: our new techniques provide efficiently computable, sharp accuracy bounds that enable much more scalable approximations for the centrality maximization task.

\section{Conclusions}
In this work we presented \algname, a new algorithm for approximate centrality maximization of sets of nodes. 
First, 
we developed a new approach based on efficiently computable bounds to Monte Carlo Rademacher Averages, a fundamental tool from statistical learning theory to obtain tight data-dependent bounds to the Supremum Deviation. 
Then, we derived new sample complexity bounds, proving that standard techniques typically provide overly conservative guarantees.
We tested \algname\ on large real-world networks, showing that it significantly outperforms the state-of-the-art. 

For future work,  
\algname\ can be extended to analyse dynamic~\cite{bergamini2014approximating,bergamini2015fully,hayashi2015fully}, uncertain~\cite{saha2021shortest}, and temporal networks~\cite{santoro2022onbra}, all settings in which our contributions may be useful to design efficient approximation algorithms for set centralities. 
Another direction is to consider the approximation of different group centralities that are not directly captured by the framework of Section~\ref{sec:prelimsetcentr} (e.g., \cite{bergamini2018scaling,angriman2021group}). 
More generally, it would be interesting to adapt the techniques introduced in this work to other data mining problems, such as mining interesting~\cite{riondato2020misosoup,fischer2020discovering}, significant~\cite{hamalainen2019tutorial,pellegrina2019hypothesis,pellegrina2019spumante,pellegrina2022mcrapper,dalleiger2022discovering}, and causal patterns~\cite{simionatobounding}.

\section{Acknowledgments}
This work was supported by the \qt{National Center for HPC, Big
Data, and Quantum Computing}, project CN00000013, funded by the
Italian Ministry of University and Research (MUR). 

\newpage
\bibliographystyle{ACM-Reference-Format}
\bibliography{bibliography}
\input{appendix-centrmax}

%% file: appendix-centrmax.tex

\appendix 

\section{Appendix}

In this Appendix we provide the proofs for our main results and additional experimental results that were deferred due to space constraints.

\subsection{Proofs of Section \ref{sec:boundsupdevrade}}

First, we state some key technical tools.
The first is a contraction inequality for Rademacher averages. 
\begin{lemma}[Lemma 26.9 \cite{ShalevSBD14}]
\label{thm:contrrade}
Let $\F$ be a family of functions from a domain $\X$ to $\R$.
Let $\phi : \R \rightarrow \R$ be a $L$-Lipschitz function, such that, for all $a,b \in \R$ it holds $|\phi(a) - \phi(b)| \leq L|a-b|$. 
Define $\vsigma$ as a vector of $m$ i.i.d. Rademacher random variables $\vsigma = < \vsigma_1 , \dots  , \vsigma_m >$. 
For any $X \in \X^m$ with $X = \{ x_1 , \dots  , x_m \}$, it holds
\begin{align*}
\E_\vsigma \Biggr[ \sup_{f \in \F} \sum_{i=1}^m \vsigma_i \phi( f(x_i) ) \Biggr] 
\leq L \E_{\vsigma} \Biggr[ \sup_{f \in \F} \biggl\{ \sum_{i=1}^m \vsigma_{i} f(x_i) \biggl\} \Biggr] . 
\end{align*}
\end{lemma}

We now state a concentration inequality for functions uniformly distributed on the binary hypercube.
\begin{theorem}[Theorem 5.3 \cite{boucheron2013concentration}]
\label{thm:concentrationhyperc}
For $c > 0$, let a function $g : \{ -1 , 1 \}^c \rightarrow \mathbb{R}$ and assume that $X$ is uniformly distributed on $\{ -1 , 1 \}^c$. 
Let $v > 0$ be such that 
\begin{align*}
\sum_{i=1}^{c} \pars{ g\pars{x} - g ( \overline{x}^i )  }^2_{+} \leq v 
\end{align*}
for all $x = ( x_1 , \dots , x_c ) \in \{-1 , 1\}^c$, where 
\[
\overline{x}^i = ( x_1 , \dots , x_{i-1} , -x_i , x_{i+1} , \dots , x_c )
\] 
is a copy of $x$ with the $i$-th component multiplied by $-1$, and $( b )_+ = \max \{ b , 0 \}$ is the positive part of $b \in \mathbb{R}$.
Then, the random variable $Z \doteq g(X)$ satisfies, for all $q > 0$,
\begin{align*}
\Pr \pars{ Z > \E \sqpars{ Z } + q } , \Pr \pars{ Z < \E \sqpars{ Z } - q } \leq \exp ( -q^2 /v ).
\end{align*}
\end{theorem}

\begin{figure}
\centering
\begin{subfigure}{.23\textwidth}
  \centering
  \includegraphics[width=\textwidth]{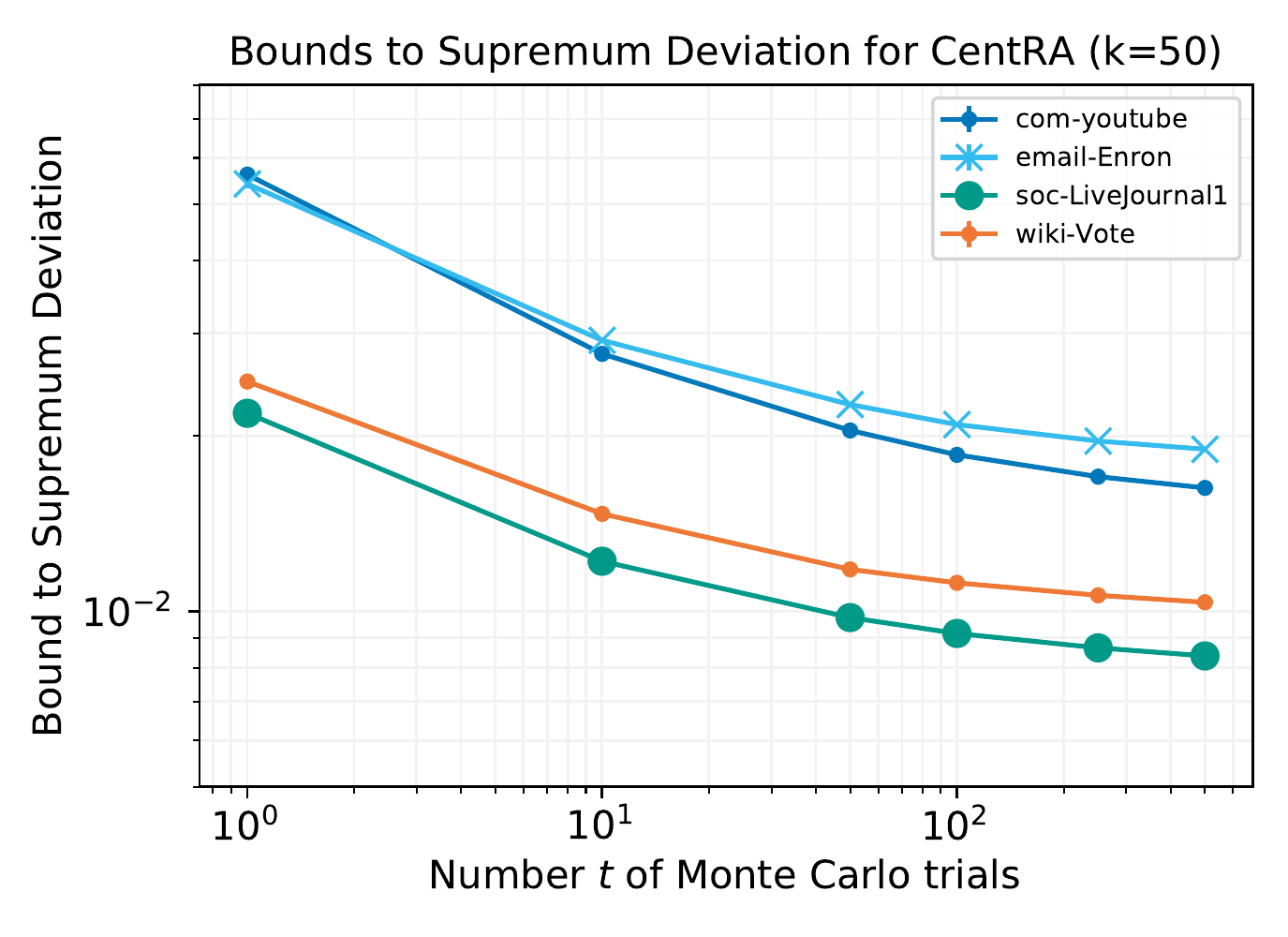}
  \caption{}
\end{subfigure}
\begin{subfigure}{.23\textwidth}
  \centering
  \includegraphics[width=\textwidth]{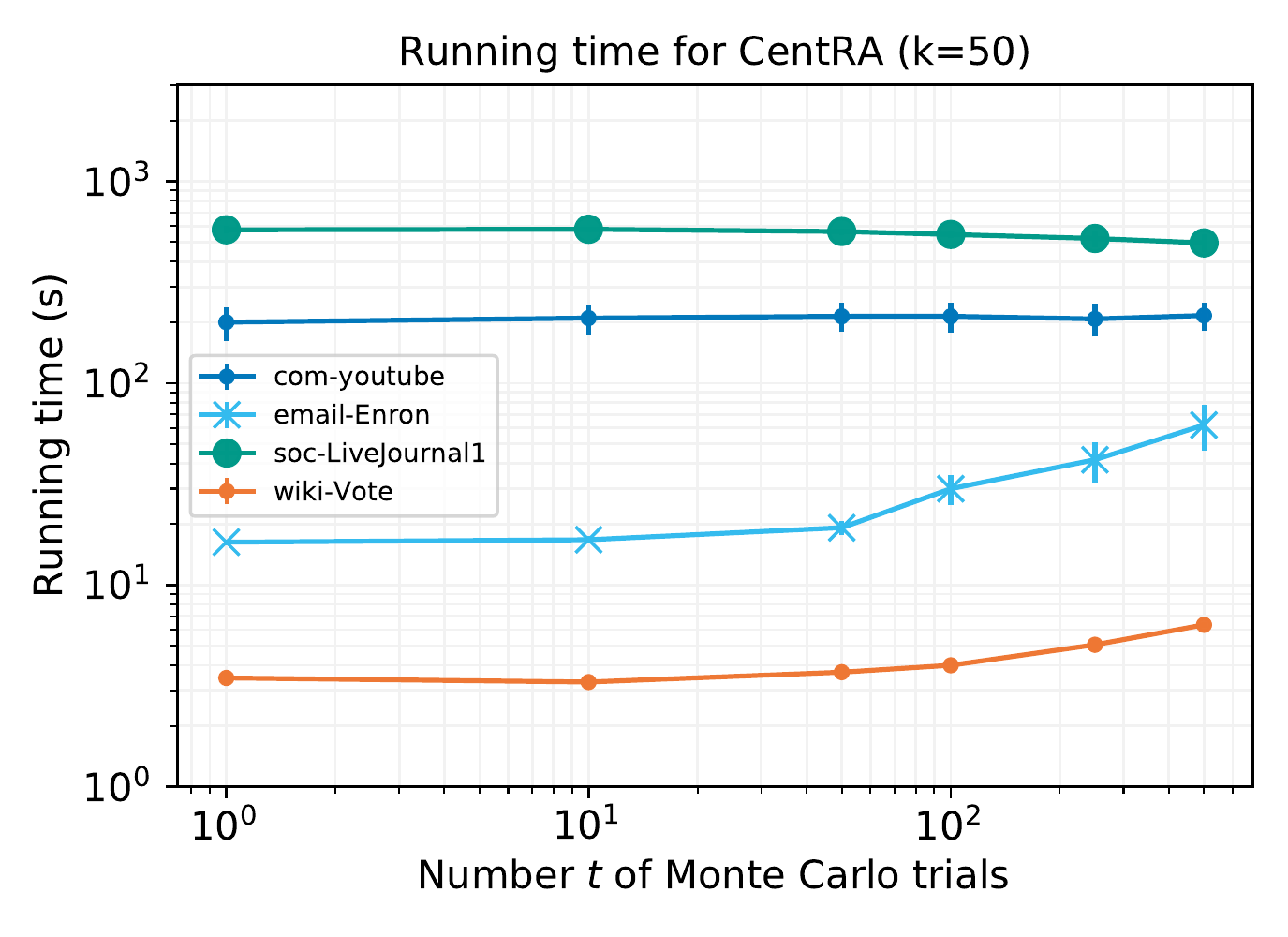}
  \caption{}
\end{subfigure}
\caption{
Impact of \algname\ parameters on upper bounds to the SD and running times on a sample of size $m = 5 \cdot 10^5$ for $k = 50$.  
(a): bounds to the SD as function of the number of Monte Carlo trials $t$. 
(b): running time of \algname.  
}
\label{fig:paramscentra}
\Description{This Figure shows the impact to the bounds to the supremum deviation and running time of CentRA of the parameter t, that is the number of Monte Carlo trials of the AMCERA. The plots show that the supremum deviation is quite stable for t larger than 100, while the running time is not significantly impacted by t, in particular for the case of larger graphs.}
\end{figure}

\begin{figure}
\centering
\begin{subfigure}{.32\textwidth}
  \centering
  \includegraphics[width=\textwidth]{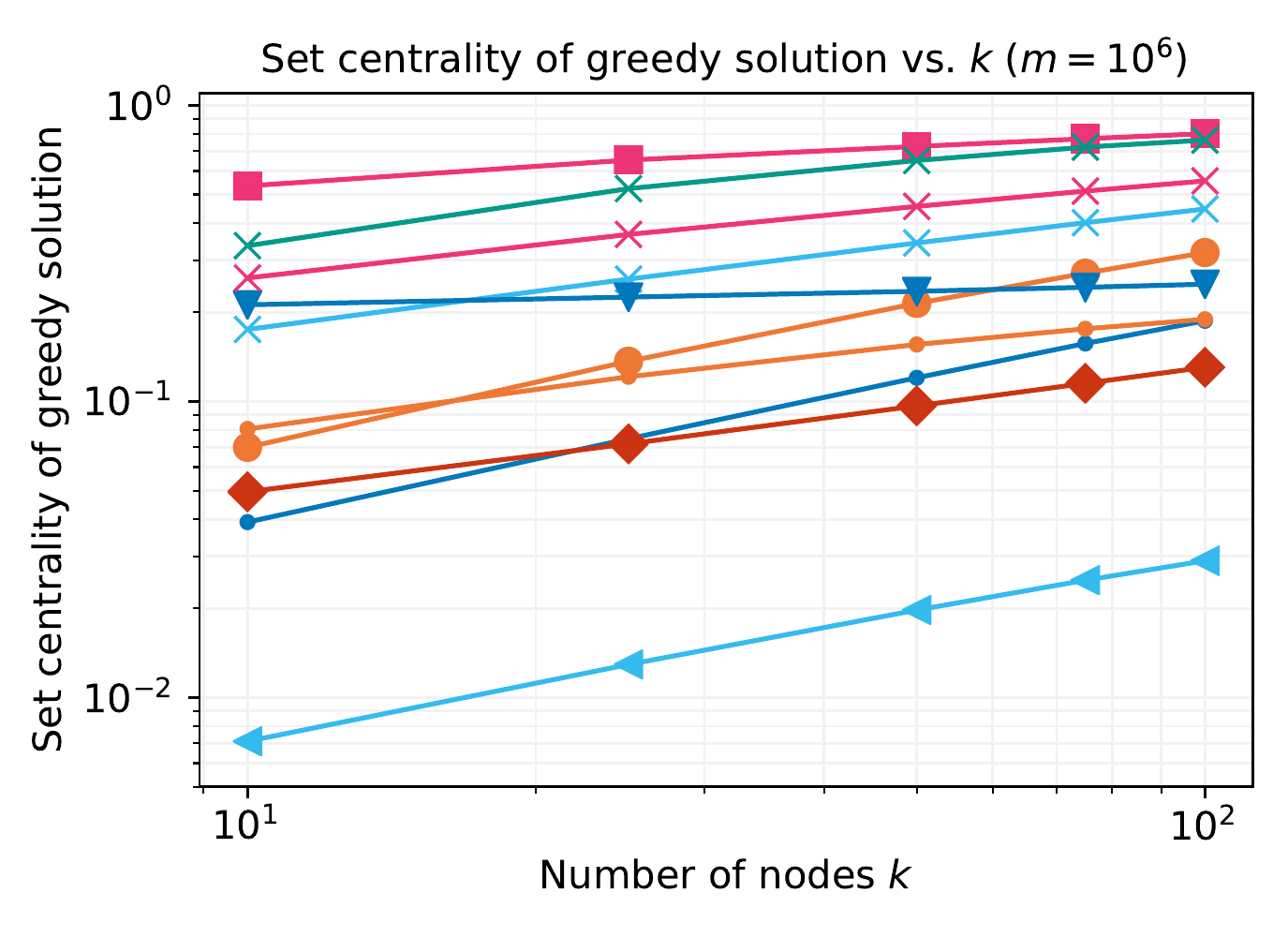}
\end{subfigure}
\caption{
Estimated centrality $\centr_\sample(\nodeset)$ of the node set $\nodeset$ returned by \algname\ for different $k$ and $m=10^6$. (the legend is the same of Figure \ref{fig:fixedm}.)
}
\label{fig:estcentr}
\Description{This Figure shows the estimated maximum set centrality obtained by CentRA for different values of k. The plot shows that for k sufficiently large, the maximum centrality approaches a constant independent of the size of the graph n.}
\end{figure}

\begin{table}
  \caption{Statistics of undirected (top section) and directed (bottom section) graphs. $B$ is the vertex diameter. }
\label{tab:graphs}
  \begin{tabular}{lrrrrr}
    \toprule
    $G$              & $|V|$      & $|E|$ & $B$       \\
    \midrule
    actor-collaboration     & 3.82e5  & 3.31e7 & 13  \\
	com-amazon     & 3.34e5 & 9.25e5 & 44  \\
	com-dblp     & 3.17e5 & 1.04e6  & 21 \\
	com-youtube     & 1.13e6 & 2.98e6 & 20  \\
	email-Enron & 3.66e4 & 1.83e5  & 11 \\
	\midrule
        soc-LiveJournal1 & 4.84e6 & 6.90e7 & 16  \\
        soc-pokec & 1.63e6 & 3.06e7  & 16 \\
        wiki-Talk & 2.39e6 & 5.02e6  & 9 \\	
        wiki-topcats & 1.79e6 & 2.85e7 & 9  \\	
	wiki-Vote & 7.11e3 & 1.03e5 & 7 \\	
  \bottomrule
\end{tabular}
\end{table}

We are now ready to prove Theorem~\ref{thm:eraboundhypercube}. 
To do so, we first prove that the expected AMCERA upper bounds the ERA with the contraction inequality of Lemma~\ref{thm:contrrade},
and then we show that the AMCERA is sharply concentrated around its expectation, leveraging Theorem~\ref{thm:concentrationhyperc}. 

\begin{proof}[Proof of Theorem~\ref{thm:eraboundhypercube}]
We first prove that 
\begin{align*}
\era \leq \E_\vsigma \sqpars{ \amera }.
\end{align*}
We observe that, for all $h \in \hyperg$ and all $\nodeset \subseteq V$, we can equivalently write the function $f_\nodeset$ as 
\begin{align*}
f_\nodeset (h) 
&= \max \big\{ f_u (h) : u \in \nodeset \big\} 
= \bigvee_{u \in \nodeset} f_u (h) \\
&= \min \brpars{1 , \sum_{u \in \nodeset} f_u(h)} .
\end{align*}
We define the function $\phi : \R \rightarrow \R$ as 
\begin{align*}
\phi(x) = \min\{1 , x\} , \text{ for } x \in [0 , k],
\end{align*}
and observe that it is $1$-Lipschitz. 
Replacing the above in the definition of the ERA $\era$, we have 
\begin{align*}
\era &= 
\E_\vsigma \Biggl[ \sup_{f_\nodeset \in \F } \Biggl\{ \frac{1}{m}
 \sum_{s=1}^m \vsigma_s f_\nodeset(h_s) \Biggr\} \Biggr] \\
 &=  \E_\vsigma \left[ \sup_{\nodeset \subseteq V , |\nodeset| \leq k} \Biggl\{ \frac{1}{m}
 \sum_{s=1}^m \vsigma_s \phi \biggl( \sum_{u \in \nodeset} f_u(h_s) \biggr) \Biggr\} \right] . \numberthis \label{eq:eraequiv}
\end{align*}
We apply Lemma~\ref{thm:contrrade} to upper bound \eqref{eq:eraequiv} obtaining that
\begin{align*}
\text{\eqref{eq:eraequiv}} & \leq 
 \E_\vsigma \left[ \sup_{\nodeset \subseteq V , |\nodeset| \leq k} \Biggl\{ \frac{1}{m}
 \sum_{s=1}^m \vsigma_s \sum_{u \in \nodeset} f_u(h_s) \Biggr\} \right] \\
 & = \E_{\vsigma} \left[ \amera \right] . 
\end{align*}
We now derive a new concentration bound for $\amera$ w.r.t. its expectation $\E_{\vsigma} \left[ \amera \right]$. 
Let the matrix $x \in \{ - 1 , 1\}^{t \times m}$.
For $j \in [1,t]$ and $s \in [1,m]$, define the matrix $\overline{x}^{js}$ as a copy of $x$ such that its $(j,s)$-th component $\overline{x}^{js}_{js}$ is equal to the $(j,s)$-th component $x_{js}$ of $x$ multiplied by $-1$, i.e. $\overline{x}^{js}_{js} = - x_{js}$. 
We define the function $g$ as  
$g(x) = tm \trade^{t}_{m}(\F, \sample, x)$. 
We aim at upper bounding 
\begin{align*}
\sum_{j = 1}^t  \sum_{s = 1}^m  \pars{ g\pars{x} - g ( \overline{x}^{js} )  }_+^2 
\end{align*}
below some constant $v$. 
We first observe that the equivalence 
\begin{align}
\pars{ g\pars{x} - g ( \overline{x}^{js} )  }_+ = \pars{ g\pars{x} - g ( \overline{x}^{js} ) } \ind{ g\pars{x} > g ( \overline{x}^{js} ) }  
\label{eq:proofequivalence}
\end{align}
implies that we can focus on the case $g\pars{x} > g ( \overline{x}^{js} )$ to upper bound \eqref{eq:proofequivalence}, as otherwise \eqref{eq:proofequivalence} is equal to $0$. 
Let $S^*_j$ be the subset of $V$ such that 
\begin{align*}
S^*_j = \argmax_{ \nodeset \subseteq V , |\nodeset| \leq k } \brpars{ \sum_{u \in \nodeset} \sum_{s=1}^{m} f_u (h_s) x_{js}  } .
\end{align*}
We note that, from the definition of supremum and after simple calculations, for any $j \in [1,t]$ and $s \in [1,m]$ it holds 
\begin{align*}
g(x) - g(\overline{x}^{js}) & \leq \sum_{u \in S^*_j} f_u (h_s) \sqpars{ x_{js} - \overline{x}^{js}_{js}  }  \\ & 
= 2 \sum_{u \in S^*_j} f_{u} (h_{s}) x_{js} ,
\end{align*}
since all entries of $x$ and $\overline{x}^{js}$ are equal, apart from the pair $(j,s)$ for which $x_{js} = - \overline{x}^{js}_{js}$. 
Summing over all $j,s$ we obtain 
\begin{align*}
& \sum_{j = 1}^t  \sum_{s = 1}^m  \pars{ g\pars{x} - g ( \overline{x}^{js} )  }_+^2 
\leq \sum_{j = 1}^t \sum_{s = 1}^m \biggl( 2 \sum_{u \in S^*_j} f_{u} (h_{s}) x_{js} \biggr)^2  \\
&= 4 \sum_{j = 1}^t \sum_{s = 1}^m \biggl( \sum_{u \in S^*_j \cap h_{s}} f_{u} (h_{s}) \biggr)^2  
\leq 4 \sum_{j = 1}^t \sum_{s = 1}^m b_\sample \sum_{u \in S^*_j} \pars{ f_{u} (h_{s}) }^2  \\
&\leq 4 b_\sample t \sup_{\nodeset \subseteq V , |\nodeset| \leq k} \biggr\{ \sum_{u \in \nodeset} \sum_{s=1}^{m} \pars{f_u (h_s)}^2 \biggl\}  \\
&= 4 t m  \ewvar_{\F}(\sample) , 
\end{align*}
where the second-last inequality holds by Cauchy-Schwarz inequality. 
We apply Theorem~\ref{thm:concentrationhyperc} with $v = 4 t m  \ewvar_{\F}(\sample)$, obtaining
\begin{align*}
\E_\vsigma\sqpars{ tm \amera } \leq tm \amera + q 
\end{align*}
with probability $ \geq 1 - \exp(-q^2/v)$. 
Setting $\exp(-q^2/v) \leq \delta $ and solving for $q$ proves the statement. 
\end{proof}

\begin{proof}[Proof of Lemma~\ref{thm:expectappxratio}]
First, note that we have already shown
\begin{align*}
\erade\left(\F_{k}, \sample\right)
\leq 
\E_{\vsigma} \sqpars{ \trade^{t}_{m}(\F_{k}, \sample, \vsigma) }
\end{align*}
in the proof of Theorem~\ref{thm:eraboundhypercube}.
We now prove 
\begin{align}
\label{eq:toprovekappx}
\E_{\vsigma} \sqpars{ \trade^{t}_{m}(\F_{k}, \sample, \vsigma) }
\leq 
k\erade\left(\F_{1}, \sample\right)  .
\end{align}
We have

\begin{align*}
& \trade^{t}_{m}(\F_{k}, \sample, \vsigma) 
= \frac{1}{t} \sum_{j=1}^{t} \sup_{\nodeset \subseteq V , |\nodeset| \leq k} \Biggl\{ 
 \sum_{u \in \nodeset} \frac{1}{m} \sum_{s=1}^m \vsigma_{js} f_u(h_s) \Biggr\}  \\
& \leq \frac{1}{t} \sum_{j=1}^{t} \sup_{\nodeset \subseteq V , |\nodeset| \leq k} \Biggl\{ 
 \sum_{u \in \nodeset} \sup_{u \in \nodeset} \brpars{ \frac{1}{m} \sum_{s=1}^m \vsigma_{js} f_u(h_s) \Biggr\} }  \\
&  \leq \frac{k}{t} \sum_{j=1}^{t}  \sup_{\nodeset \subseteq V , |\nodeset| \leq k} \Biggl\{  \sup_{u \in \nodeset} \brpars{ \frac{1}{m} \sum_{s=1}^m \vsigma_{js} f_u(h_s) \Biggr\} } \\
&  = \frac{k}{t} \sum_{j=1}^{t}  \sup_{\nodeset \subseteq V , |\nodeset| \leq 1} \Biggl\{ \sum_{u \in \nodeset}  \frac{1}{m} \sum_{s=1}^m \vsigma_{js} f_u(h_s) \Biggr\}  \\
&= k \trade^{t}_{m}(\F_{1}, \sample, \vsigma) .
\end{align*}
Taking the expectation w.r.t. $\vsigma$ on both sides yields \eqref{eq:toprovekappx}.
The rightmost inequality of the statement follows from simple properties of the supremum. 
\end{proof}

We now show Theorem~\ref{thm:selfboundingupperbound}. 
Our proof leverages the concept of \emph{self-bounding} functions \cite{boucheron2013concentration}. 
Define the functions $g : \X^m \rightarrow \R$ and $g^i : \X^{m-1} \rightarrow \R$, let the vector $x = ( x_1 , \dots , x_m ) \in \X^m$ and define $x^i = ( x_1 , \dots , x_{i-1} , x_{i+1} , \dots , x_m )$ 
as a copy of $x$ removing the $i$-th entry $x_i$. 
We say that $g$ is a self-bounding function if there exists a $g^i$ such that, for all $x \in \X^m$,  it holds 
\begin{align*}
0 \leq g(x) - g^i(x^i) \leq 1 ,
\end{align*}
and that 
\begin{align*}
\sum_{i=1}^m (g(x) - g^i(x^i)) \leq g(x). 
\end{align*}
If all the $x_i \in \X$ are independent random variables, a self-bounding function satisfies (Theorem 6.12 \cite{boucheron2013concentration})  
\begin{align*}
\Pr \pars{ g(x) \leq \E_x[g(x)] - q  } \leq \exp \pars{ - \frac{q^2}{2 \E_x[g(x)] } } .
\end{align*}

\begin{proof}[Proof of Theorem~\ref{thm:selfboundingupperbound}]
For any sample $\sample$ of size $m$, define $\nodeset^\prime = \argmax_{\nodeset \subseteq V , |\nodeset| \leq k} \centr_\sample \pars{\nodeset}$. 
From Jensen's inequality it holds 
\begin{align*}
\centr \pars{ \nodeset^* } & = \sup_{\nodeset \subseteq V , |\nodeset| \leq k} \E_\sample\sqpars{ \centr_\sample \pars{\nodeset} }  \\
& \leq 
\E_\sample \biggl[ \sup_{\nodeset \subseteq V , |\nodeset| \leq k}  \centr_\sample \pars{\nodeset} \biggr] 
= \E_\sample\sqpars{ \centr_\sample( \nodeset^\prime  ) }. 
\end{align*}
We now show that the function $g : \hyperg^m \rightarrow \R$, with
\begin{align*}
g(\sample) = m \centr_\sample(\nodeset^\prime) =  \sup_{\nodeset \subseteq V , |\nodeset| \leq k}  \sum_{s=1}^m f_\nodeset (h_s) 
\end{align*}  
is self-bounding. 
Define $g^i(\sample)$ as 
\begin{align*}
g^i(\sample) = \sup_{\nodeset \subseteq V , |\nodeset| \leq k} \sum_{s=1 , s \neq i}^m f_\nodeset (h_s) .
\end{align*}
First, denoting $x = \sample$, note that $0 \leq g(x) - g^i(x^i) \leq f_{\nodeset^\prime} (h_i) \leq 1$ are immediate from properties of the supremum and the definition of $f$.
Then, we have 
\begin{align*}
\sum_{i=1}^m (g(x) - g^i(x^i)) \leq 
\sum_{i=1}^m f_{\nodeset^\prime} (h_i) = g(x). 
\end{align*}
Therefore, $g$ is self-bounding, and it holds 
\begin{align*}
\Pr \pars{ m \centr_\sample(\nodeset^\prime) \leq \E[ m \centr_\sample(\nodeset^\prime)] - q } \leq \exp \pars{ \frac{ - q^2}{2 \E_\sample[m
\centr_\sample(\nodeset^\prime)] } } .
\end{align*}
Setting the r.h.s. $\leq \delta$, and solving for $q$, we obtain w.p. $\geq 1 - \delta$
\begin{align*}
\E[ \centr_\sample(\nodeset^\prime)] \leq \centr_\sample(\nodeset^\prime) + \sqrt{ \frac{2 \E[ \centr_\sample(\nodeset^\prime)] \ln(\frac{1}{\delta})}{m}} .
\end{align*}
Finding the fixed point of the above inequality gives
\begin{align}
\E[ \centr_\sample(\nodeset^\prime)] \leq \centr_\sample(\nodeset^\prime) + \sqrt{ \pars{\frac{\ln(\frac{1}{\delta})}{m}}^2 + \frac{2 \centr_\sample(\nodeset^\prime) \ln(\frac{1}{\delta})}{m}} + \frac{\ln(\frac{1}{\delta})}{m} . \label{eq:empupperboundestsup}
\end{align}
Let $\nodeset$ be the output of \texttt{greedyCover($k , \sample$)} for a fixed sample $\sample$. 
For the properties of \texttt{greedyCover}, it holds 
\begin{align*}
\centr_\sample \pars{\nodeset^\prime} \leq \centr_\sample \pars{\nodeset} \pars{ 1 - 1/e}^{-1} .
\end{align*}
By replacing this upper bound to $\centr_\sample \pars{\nodeset^\prime}$ in \eqref{eq:empupperboundestsup} and recalling that $\centr(\nodeset^*) \leq \E[ \centr_\sample(\nodeset^\prime)]$ we obtain the statement. 
\end{proof}

Before proving Theorem~\ref{thm:bound_dev}, we state some intermediate results. 
These results provide refined concentration inequalities relating Rademacher averages to the supremum deviation~\cite{boucheron2013concentration}. 
Let the \emph{Rademacher complexity} $\rc$ of the set of functions $\F$ (defined in Section~\ref{sec:prelims}) be defined as $\rc = \E_{\sample} \sqpars{ \era }$. 
The following central result relates $\rc$ to the \emph{expected} supremum deviation.

\begin{lemma}{Symmetrization Lemma \citep{ShalevSBD14,mitzenmacher2017probability}}
\label{symlemma}
Let $Z$ be either 
\begin{align*}
\sup_{f\in \F} \brpars{ \centr_{\sample}(\nodeset) -  \centr(\nodeset) }  
\text{ or }  
\sup_{f\in \F} \brpars{ \centr(\nodeset) -  \centr_{\sample}(\nodeset) } .
\end{align*}
It holds
$ \E_{\sample} \sqpars{ Z } \leq 2 \rc $.
\end{lemma}

The following theorem states Bousquet's inequality~\cite{bousquet2002bennett}, that   gives a variance-dependent bound to the supremum deviation above its expectation.
\begin{theorem}{(Thm. 2.3 \cite{bousquet2002bennett})}\label{thm:sdbousquetbound}
  Let $\hat{\nu}_{\F} \geq \sup_{f\in \F} \left\lbrace Var(f) \right\rbrace$,
  and $Z$ be either 
$\sup_{f\in \F} \brpars{ \centr_{\sample}(\nodeset) -  \centr(\nodeset) } $ 
or 
$\sup_{f\in \F} \brpars{ \centr(\nodeset) -  \centr_{\sample}(\nodeset) } $. 
  Then, with probability at
  least $1 - \lambda$ over $\sample$, it holds 
  \begin{align*}
    Z \le \E_\sample\left[ Z \right] + \sqrt{\frac{2 \ln \bigl( \frac{1}{\lambda} \bigr)
    \left( \hat{\nu}_{\F} + 2 \E[ Z ] \right)}{m}}
        + \frac{ \ln \bigl( \frac{1}{\lambda} \bigr) }{3m} .
  \end{align*}
\end{theorem}

The next result bounds $\rc$ above its estimate $\era$.
\begin{theorem}{\cite{boucheron2013concentration}}
\label{thm:rcboundselfbounding}
With probability $\geq 1-\lambda$ over $\sample$, it holds
\begin{align*}
\rc \leq \era  + \sqrt{ \pars{\frac{\ln \bigl( \frac{1}{\lambda} \bigr)}{m}}^2 + \frac{2\ln\pars{\frac{1}{\lambda}} \era }{ m } }   + \frac{\ln \bigl( \frac{1}{\lambda} \bigr)}{m} .
\end{align*}
\end{theorem}


We now prove Theorem~\ref{thm:bound_dev}. 

\begin{proof}[Proof of Theorem~\ref{thm:bound_dev}]
We define $5$ events $A_1 , \dots , A_5$ as 
\begin{align*}
A_1 &= \qtm{ \sup_{f\in \F} \brpars{ \centr_{\sample}(\nodeset) -  \centr(\nodeset) }  > \eta } , \\
A_2 &= \qtm{ \sup_{f\in \F} \brpars{ \centr(\nodeset) - \centr_{\sample}(\nodeset) } > \eta  } , \\
A_3 &= \qtm{ \era > \tilde{\rade}  } , \\ 
A_4 &= \qtm{ \rc > \rade } , \\
A_5 &= \qtm{ \sup_{f\in \F} \brpars{ \centr(\nodeset) } > \nu  } .
\end{align*}
The statement holds if, with probability $\geq 1 - \delta$, $A_1$ and $A_2$ are both false. 
If we assume that $\Pr(A_i) \leq \delta/5 , \forall i$, then it holds 
\begin{align*}
\Pr\Bigl( \bigcup_i A_i \Bigr) \leq \sum_{i} \Pr \pars{A_i} \leq \delta .
\end{align*}
This implies that all events $A_i$ are false simultaneously with probability $\geq 1 - \delta$, obtaining the statement. 
We now show that $\Pr(A_i) \leq \delta/5 , \forall i$. 
We observe that $Pr(A_3) \leq \delta/5$ is a consequence of Theorem~\ref{thm:eraboundhypercube} (replacing $\delta$ by $\delta/5$). 
Then, $Pr(A_4) \leq \delta/5$ follows from Theorem~\ref{thm:rcboundselfbounding} (using $\lambda = \delta/5$ and $\era \leq \tilde{\rade}$). 
$Pr(A_5) \leq \delta/5$ is a consequence of Theorem~\ref{thm:selfboundingupperbound} (replacing $\delta$ by $\delta/5$). 
$Pr(A_1) \leq \delta/5$ and $Pr(A_2) \leq \delta/5$ both hold applying Bousquet's inequality (Theorem~\ref{thm:sdbousquetbound}) twice (with $\lambda = \delta/5$), observing that $\rc \leq \rade$, using the Symmetrization Lemma (Lemma~\ref{symlemma}), and observing that 
\begin{align*}
\sup_{f\in \F} \left\lbrace Var(f) \right\rbrace \leq \sup_{f\in \F} \left\lbrace \centr(\nodeset) \right\rbrace \leq \nu .
\end{align*}
The fact $\Pr(A_i) \leq \delta/5 , \forall i$,  
concludes our proof. 
\end{proof}

\subsection{Proofs of Section \ref{sec:algorithm}}

\begin{proof}[Proof of Proposition~\ref{prop:main}]
We prove that the statement is a consequence of 
the following facts, that we assume hold simultaneously with probability 
$\geq 1 - \delta/2^{i}$ at a fixed iteration $i \geq 0$: 
\begin{enumerate}
\item \algname\ computes $\amera$ correctly at line~\ref{alg:amceraend}; \label{item:mcera}
\item $\eta$, at the end of iteration $i$, is computed using Theorem~\ref{thm:bound_dev} such that $\sd(\F, \sample_i) \leq \eta$; \label{item:supdevs}
\item $\xi$ (line \ref{alg:xidef}) is such that $\centr(\nodeset^*) \leq \centr_{\sample_i}(\nodeset) / (1-1/e) + \xi$; \label{item:xisupdevs}
\item $\nodeset$ is computed by the greedy algorithm \texttt{greedyCover($k , \sample_i$)} s.t. $\centr_{\sample_i}(\nodeset) \geq (1-1/e) \sup_{\nodeset^\prime \subseteq V , |\nodeset^\prime| \leq k} \centr_{\sample_i}(\nodeset^\prime)$; \label{item:greedy}
\item \algname\ stops at the current iteration $i$ when the stopping condition of line \ref{alg:stopcond} is true. \label{item:stopping}
\end{enumerate}

Notice that if all the following inequalities hold,
\begin{align*}
\centr(\nodeset^*) \leq \frac{\centr_{\sample_i}(\nodeset)}{1 - \frac{1}{e}} + \xi \leq \frac{ \centr_{\sample_i}(\nodeset) - \eta }{ 1 - \frac{1}{e} - \varepsilon } \leq \frac{ \centr(\nodeset) }{1 - \frac{1}{e} - \varepsilon} ,
\end{align*}
then $\centr(\nodeset^*) \leq \centr(\nodeset)/(1-1/e-\varepsilon)$ and \algname\ is correct. 
The rightmost inequality holds due to facts \eqref{item:mcera} and \eqref{item:supdevs}, while the leftmost holds from facts \eqref{item:xisupdevs} and \eqref{item:greedy}. 
The central inequality holds as it is equivalent to the stopping condition (line \ref{alg:stopcond}) and from fact \eqref{item:stopping}.

We need to prove that the algorithm is correct with probability $\geq 1 - \delta$, taking into account the validity of the probabilistic guarantees for of all possible iterations $i$. 
The probability that there is at least one iteration $i$ in which any of the above facts is false is at most 
$\textstyle\sum_{i} \delta_i = \sum_i \delta 2^{-i} \leq \delta$,  
from a union bound, proving the statement. 
\end{proof}

\begin{proof}[Proof of Lemma~\ref{amceraefficient}]
When a sample $h \in \sample$ is generated, $\BOi{ |h| t }$ operations are sufficient to generate $t$ Rademacher random variables and update all values of $r^u_j$, for all $u \in h$ and $j \in [1,t]$. 
With $t$ heaps, each of size at most $n$, the $t$ sets of values $\{r^u_j : u \in V\}$ can be kept sorted with a total of $\BOi{ |h| t \log(n) }$ time for each sample $h$. 
Summing such costs for all $h \in \sample$ (noting that each sample is considered only once) and observing that $|h| \leq b_\sample$, we obtain the first term of the bound.
The second term is the time to retrieve the $k$ largest positive values from all heaps, after processing all $m_i$ samples, that costs $\BOi{k t \log(n)}$ operations. 
\end{proof}

\subsection{Proofs of Section \ref{sec:samplecomplexity}}

Let $\F$ be a function family from a domain $\X$ to $\{ 0,1 \}$, and define the range space $Q = (\X , R)$ such that $R$ is the family of subsets of $\X$ generated by $\F$:
\begin{align*}
R = \brpars{ \brpars{ x \in \X : f(x) = 1 } : f \in \F } . 
\end{align*}
A sample $\sample = \{ h_1 , \dots , h_m \} \in \X^m$ of size $m$ taken i.i.d. from a distribution $\probdist$ gives an \emph{$(r , \theta)$-relative approximation} for $\F$ if, for $r >0$ and $\theta \in (0,1]$, it holds 
\begin{align*}
\biggl\lvert  \sum_{i=1}^m f(h_i) - \E_{h \sim \probdist}\sqpars{ f(h) } \biggr \rvert \leq r \max\{ \theta , \E_{h \sim \probdist}\sqpars{ f(h) } \}  , \forall f \in \F. 
\end{align*}
We state a result due to \cite{li2001improved} (in a version presented by \cite{har2011relative}).
\begin{theorem}[Thm. 2.11 \cite{har2011relative}]
\label{thm:relappxsamples}
Let $d \geq VC(Q)$. 
A random sample $\sample$ of size $m$ taken i.i.d. from a distribution $\probdist$ gives an $(r , \theta)$-relative approximation for $\F$ with probability $\geq 1 - \delta$ if  
\begin{align*}
m \geq c \frac{d \ln\pars{ \frac{1}{\theta} } + \ln\pars{ \frac{1}{\delta} } }{ r^2 \theta } ,
\end{align*}
where $c$ is an absolute constant. 
\end{theorem}
Note that the sample complexity bound of Theorem~\ref{thm:relappxsamples} is optimal up to constant factors \cite{li2001improved}. 
We now prove Theorem~\ref{thm:vcsamplecomplexity}. 
\begin{proof}[Proof of Theorem~\ref{thm:vcsamplecomplexity}]
In accordance with the statement, assume $m$ to be 
\begin{align*}
m = 4c \frac{d_k \ln \pars{ \frac{1}{ \centr(\nodeset^*) } } + \ln \pars{ \frac{1}{\delta} } }{ \varepsilon^2 \centr(\nodeset^*) } ,
\end{align*}
where $c$ is the absolute constant of Theorem~\ref{thm:relappxsamples}. 
For $r = \varepsilon /2$, $\theta = \centr(\nodeset^*)$, and $Q = Q_k$,  
the guarantees of Theorem~\ref{thm:relappxsamples} imply that the sample $\sample$ provides an $(\varepsilon/2 ,  \centr(\nodeset^*))$-relative approximation for $\F$ (where $\F$ is defined in Section~\ref{sec:prelimsetcentr}) with probability $\geq 1 - \delta$. From the definition of $(\varepsilon/2 ,  \centr(\nodeset^*))$-relative approximation, it holds 
\begin{align*}
\supdev \leq \varepsilon \centr(\nodeset^*) / 2 .
\end{align*}
Following analogous derivations of the proof of Theorem 1 in \cite{mahmoody2016scalable}, this constraint to the SD is a sufficient condition to guarantee that $\nodeset$ provides an $(1- 1/e - \varepsilon)$-approximation of $\nodeset^*$ with probability $\geq 1 - \delta$, proving the statement. 
\end{proof}

We now prove that the VC-dimension $VC(Q_k)$ of the range space $Q_k$ can be upper bounded in terms of the VC-dimension $VC(Q_1)$ of the range space $Q_1$ (Lemma~\ref{thm:lemmavckbound}). 
We follow steps that are similar to the proof of bounds to the VC-dimension of a concept class composed by the intersection of up to $k$ concept classes (Exercise 3.23 of \cite{mohri2018foundations}). 
In our case, $R_k$ can be seen as the \emph{disjunction} of up to $k$ ranges from $R_1$. 
\begin{proof}[Proof of Lemma~\ref{thm:lemmavckbound}]
We first prove that, for all $k \geq 1$,  
\begin{align*}
|P(\sample , R_k)| \leq |P(\sample , R_1)| |P(\sample , R_{k-1})| ,
\end{align*}
noting that when $k=1$ it holds $|P(\sample , R_{k-1})| = 1 $.  
Define the set of ranges $R_k^v$ built from the set of nodes $V\setminus\{v\}$ as
\begin{align*}
R_k^v = \brpars{ \brpars{ h : f_\nodeset(h) = 1 } : \nodeset \subseteq V \setminus \{ v \} , |\nodeset| \leq k} . 
\end{align*}
Clearly, it holds $R_k^v \subseteq R_k, \forall v \in V$. We can write $R_k$ as the union of the ranges of individual nodes $v$ with the ranges of $R_{k-1}^v$: 
\begin{align*}
& R_k  = \brpars{ \brpars{ h : f_\nodeset(h) = 1 } : \nodeset \subseteq V , |\nodeset| \leq k  } \\
& = \bigcup_{v \in V} \brpars{ \brpars{ h : f_v(h) = 1 } \cup \brpars{ h : f_\nodeset(h) = 1 } : \nodeset \subseteq V \setminus \{ v \} , |\nodeset| \leq {k-1}  } \\
& = \bigcup_{v \in V} \brpars{ \brpars{ h : f_v(h) = 1 } \cup r : r \in R_{k-1}^v  } . 
\end{align*}
Therefore, since $R_{k-1}^v \subseteq R_{k-1}$ for all $v$, the set of projection $P(\sample , R_k)$ of $R_k$ on $\sample$ is contained in the union of all the projections of $P(\sample , R_{k-1})$ in conjunction with all projections $x \in P(\sample , R_{1})$:
\begin{align*}
P(\sample , R_k) \subseteq \bigcup_{x \in P(\sample , R_{1})} \brpars{ (\sample \cap r) \cup x : r \in R_{k-1} } .
\end{align*}
This implies that the size $|P(\sample , R_k)|$ of $P(\sample , R_k)$ can be bounded with an union bound by
\begin{align*}
|P(\sample , R_k)| 
& \leq \sum_{x \in P(\sample , R_{1})} \abs{ \brpars{ (\sample \cap r) \cup x : r \in R_{k-1} } } \\
& \leq \sum_{x \in P(\sample , R_{1})} \abs{ \brpars{ \sample \cap r : r \in R_{k-1} } } \\
& = \sum_{x \in P(\sample , R_{1})} | P(\sample , R_{k-1}) | \\
& = |P(\sample , R_{1})| | P(\sample , R_{k-1}) | ,
\end{align*}
where the second inequality holds since the number of intersections of the ranges $R_{k-1}$ on the sample $\sample$ can only decrease when in conjunctions with a fixed subset $x \subseteq \sample$. 

Iterating this argument $k$ times, we obtain
\begin{align*}
|P(\sample , R_k)| \leq |P(\sample , R_1)|^k .
\end{align*}

To prove that the VC-dimension of the range space $Q_k$ is $< m$, it must hold that 
\begin{align*}
|P(\sample , R_1)|^k < 2^m .
\end{align*}
Let $d = VC(Q_1)$. 
We apply Sauer-Shelah's Lemma \cite{ShalevSBD14} obtaining 
\begin{align*}
|P(\sample , R_1)|^k \leq \pars{ \frac{e m }{d} }^{k d} .
\end{align*}
We now need to show that, for $m = 2 d k \log_2(3k)$, it holds
\begin{align*}
\pars{ \frac{e m }{d} }^{k d} < 2^m .
\end{align*}
Taking the $\log_2$ on both sides, and substituting $m = 2 d k \log_2(3k)$, we obtain
\begin{align*}
& \pars{ \frac{e m }{d} }^{k d} < 2^m \\ 
& \iff \log_2 ( 2ek \log_2(3k) ) < \log_2(9k^2) \\
& \iff 2ek \log_2(3k) < 9 k^2 \\
& \iff \log_2(3k) < \frac{9k}{2e} ,
\end{align*}
which holds for all $k \geq 1$, proving the statement. 
\end{proof}

\balance

\begin{proof}[Proof of Lemma~\ref{thm:vcboundgeneral}]
Denote a set $\sample = \{ h_1 , \dots , h_m \}$ of size $m$ that is shattered, such that
\begin{align*}
|P(\sample , R_1)| = 2^m .
\end{align*}
If $\sample$ is shattered, it means that each $h_i \in \sample$ is in the projection of $2^{m-1}$ distinct ranges;  
there must be at least $2^{m-1}$ distinct nodes $u \in V$ with $f_u(h_i) = 1$. 
This implies that 
$|h_i| \geq 2^{m-1}$. 
Since $|h| \leq b, \forall h \in \hyperg$, it also holds that 
\begin{align*}
2^{m-1} \leq b \iff m \leq  \log_2 (2b)  .
\end{align*}
Therefore, the maximum size of a set that can be shattered is at most $\lfloor \log_2 (2b) \rfloor$, obtaining the statement. 
\end{proof}

\subsection{Additional Experimental Results}
\label{sec:additionalexperiments}
In this Section we report additional results not included in the main text due to space constraints.

\textit{Comparisons between HEDGE, \algname, and VC-dimension bounds.} 
Figures \ref{fig:samplecompl-vc-hedge} and \ref{fig:samplecompl-vc-centra} compare the bounds to the Supremum Deviation obtained from the results described in Section \ref{sec:samplecomplexity}
(based on the VC-dimension), 
with results obtained by HEDGE (based on the union bound) and \algname\ (based on Monte Carlo Rademacher Averages, Section \ref{sec:boundsupdevrade}). 
To ease the comparison, we denote a modification of \algname, that we call \algname-VC, that instead of using Rademacher Averages (the results of Section~\ref{sec:boundsupdevrade}) uses the VC-dimension based bounds (from Section~\ref{sec:samplecomplexity}). 
For these experiments we consider $\delta = 0.05$, $k \in \{ 10 , 50 , 100\}$, and samples of size $m \in \{ 5 \cdot 10^4 , 10^5 , 2 \cdot 10^5 , 5 \cdot 10^5 , 10^6 \}$ (other values of $k$ produced analogous results). 
The bounds of Section \ref{sec:samplecomplexity} are computed bounding the VC-dimension $V(Q_k)$ with Lemma \ref{thm:lemmavckbound} and Lemma \ref{thm:vcboundgeneral}, using $b=B$ (the vertex diameter of the graphs, see Table~\ref{tab:graphs}). 

From Figure \ref{fig:samplecompl-vc-hedge}, comparing \algname-VC with HEDGE, we observe that \algname-VC yields generally more accurate bounds, up to $4$ times smaller for two graphs, confirming that standard techniques provide too conservative guarantees in such cases.

Figure \ref{fig:samplecompl-vc-centra} compares \algname-VC with \algname. As we may have expected, we can clearly conclude that distribution- and data-dependent bounds are significantly more accurate that 
VC-dimension based results
in almost all cases, 
improving up to a factor $3$.
This confirms the significance of the contributions at the core of \algname.

\begin{figure*}
\centering
\begin{subfigure}{.75\textwidth}
  \centering
  \includegraphics[width=\textwidth]{./figures/bounds-fixed-m-legend.pdf}
\end{subfigure} \\
\begin{subfigure}{.24\textwidth}
  \centering
  \includegraphics[width=\textwidth]{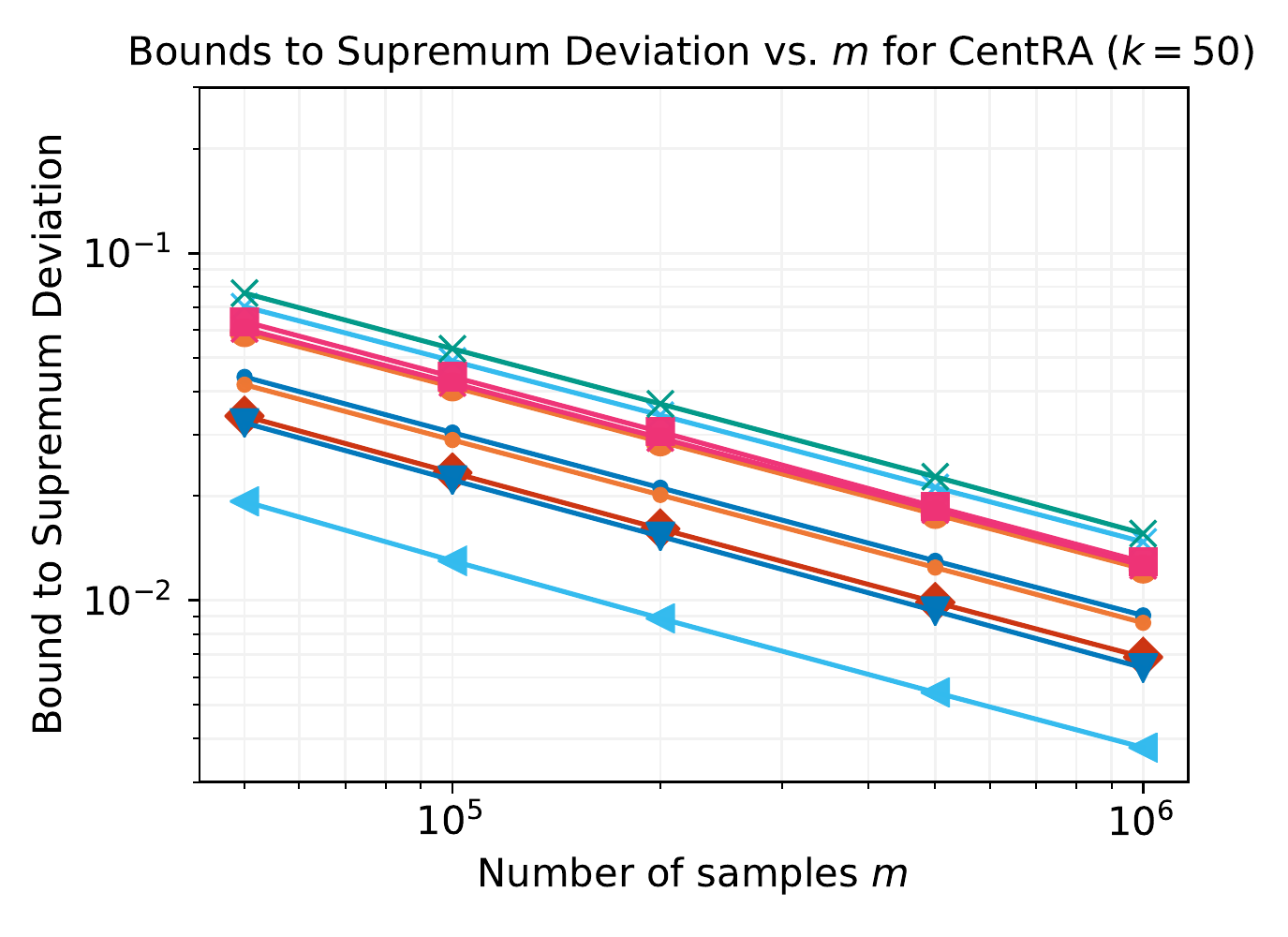}
  \caption{}
\end{subfigure}
\begin{subfigure}{.24\textwidth}
  \centering
  \includegraphics[width=\textwidth]{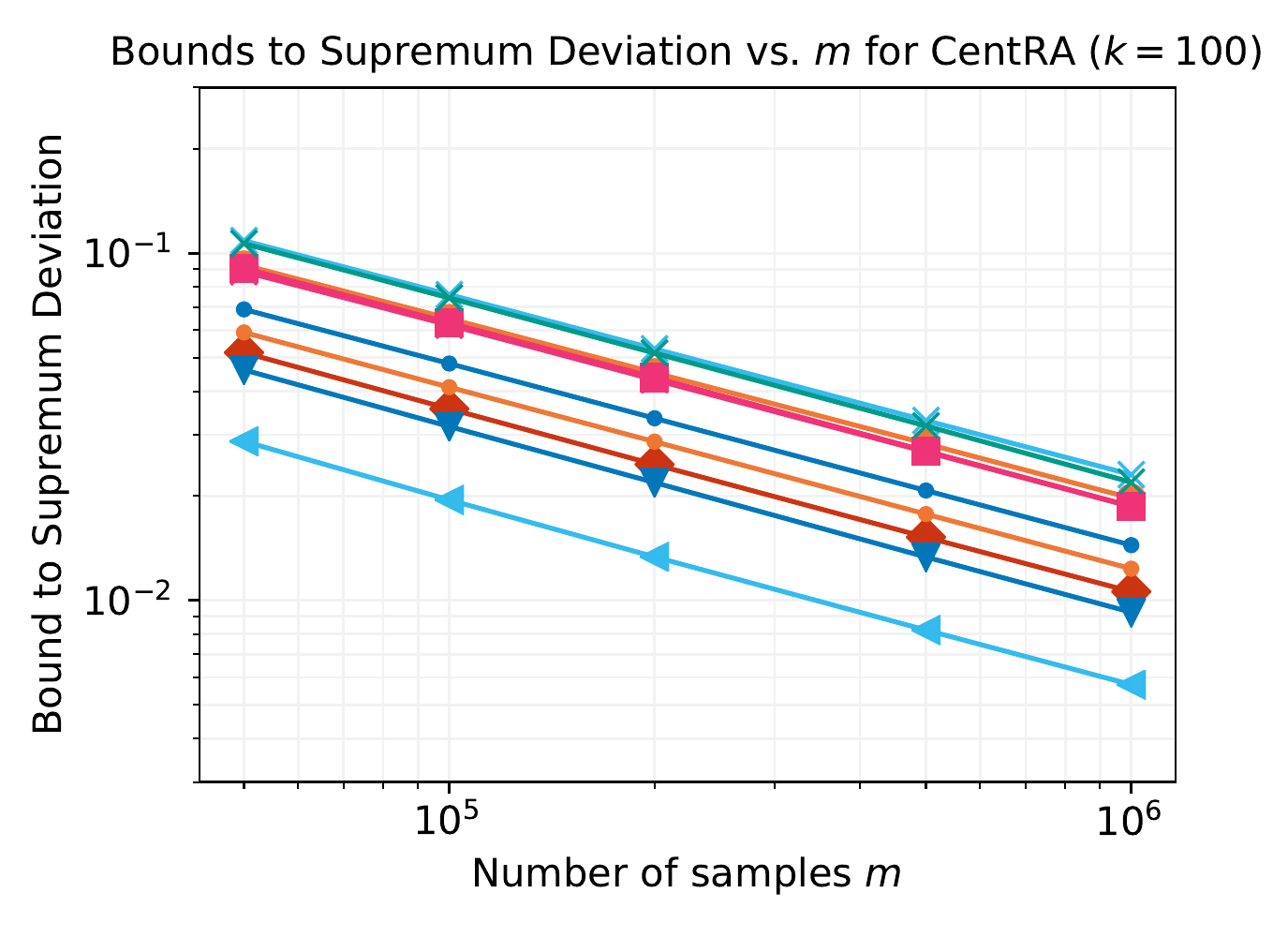}
  \caption{}
\end{subfigure}
\begin{subfigure}{.24\textwidth}
  \centering
  \includegraphics[width=\textwidth]{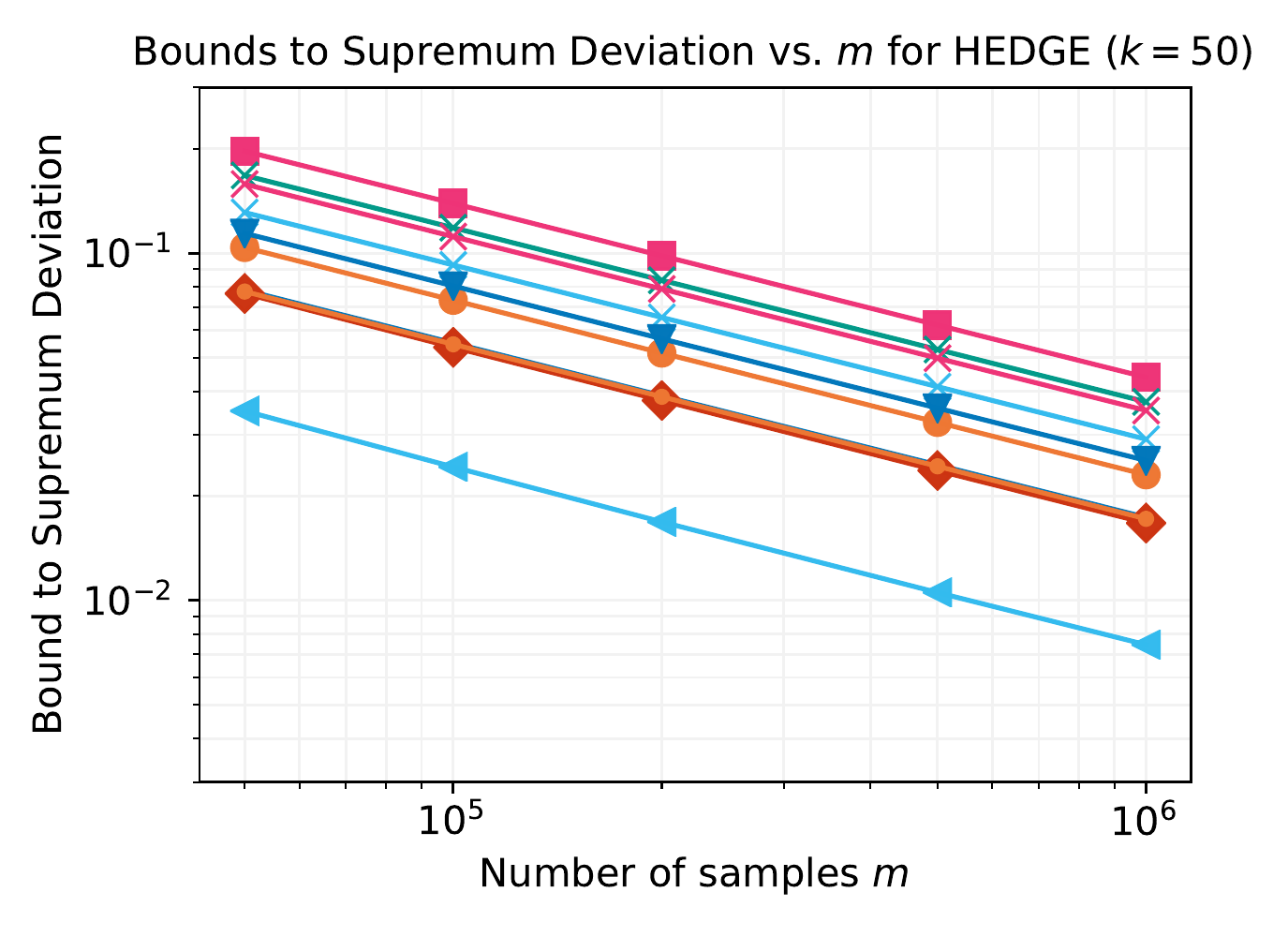}
  \caption{}
\end{subfigure}
\begin{subfigure}{.24\textwidth}
  \centering
  \includegraphics[width=\textwidth]{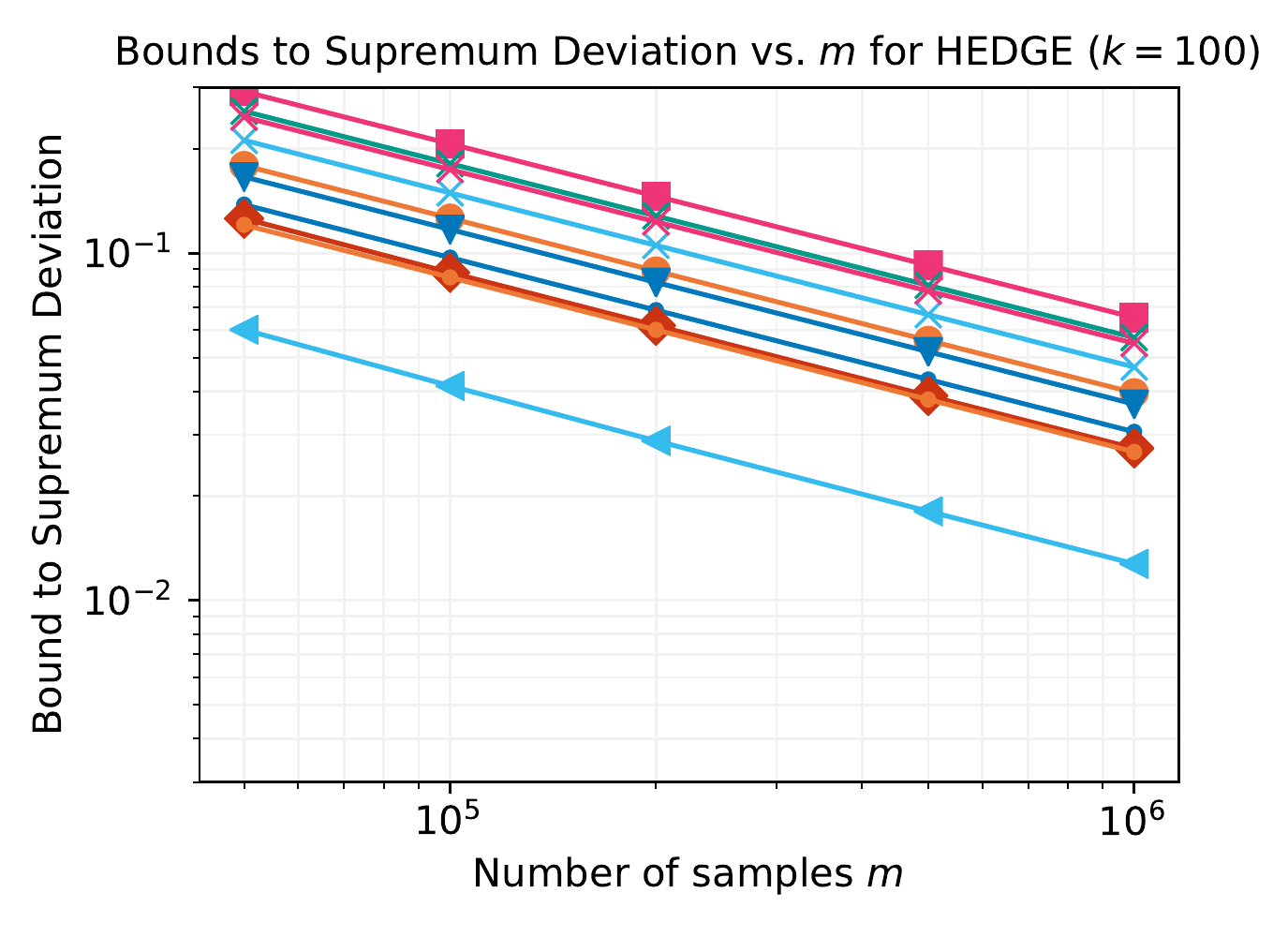}
  \caption{}
\end{subfigure}
\caption{
Bounds to the SD using \algname\ (Figures (a)-(b)) and HEDGE (Figures (c)-(d)) as functions of $m$, for $k=50$ and $k=100$.  
}
\label{fig:supdevvsmindividual}
\Description{This Figure shows the relationship between the number of samples m and the bounds to the supremum deviation obtained by CentRA and HEDGE.}
\end{figure*}

\begin{figure*}[ht]
\centering
\begin{subfigure}{.75\textwidth}
  \centering
  \includegraphics[width=\textwidth]{./figures/bounds-fixed-m-legend.pdf}
\end{subfigure} \\
\begin{subfigure}{.32\textwidth}
  \centering
  \includegraphics[width=\textwidth]{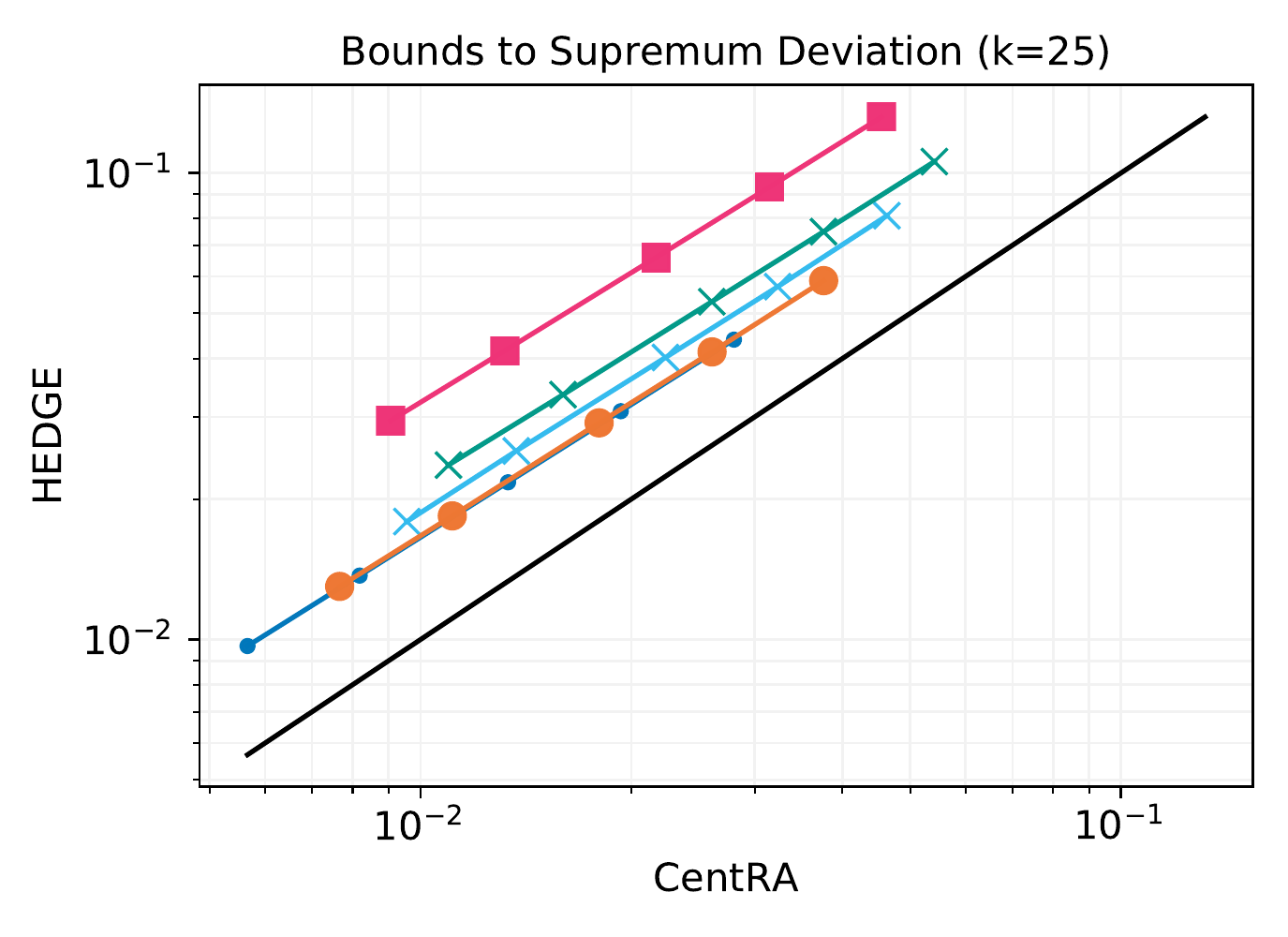}
  \caption{}
\end{subfigure}
\begin{subfigure}{.32\textwidth}
  \centering
  \includegraphics[width=\textwidth]{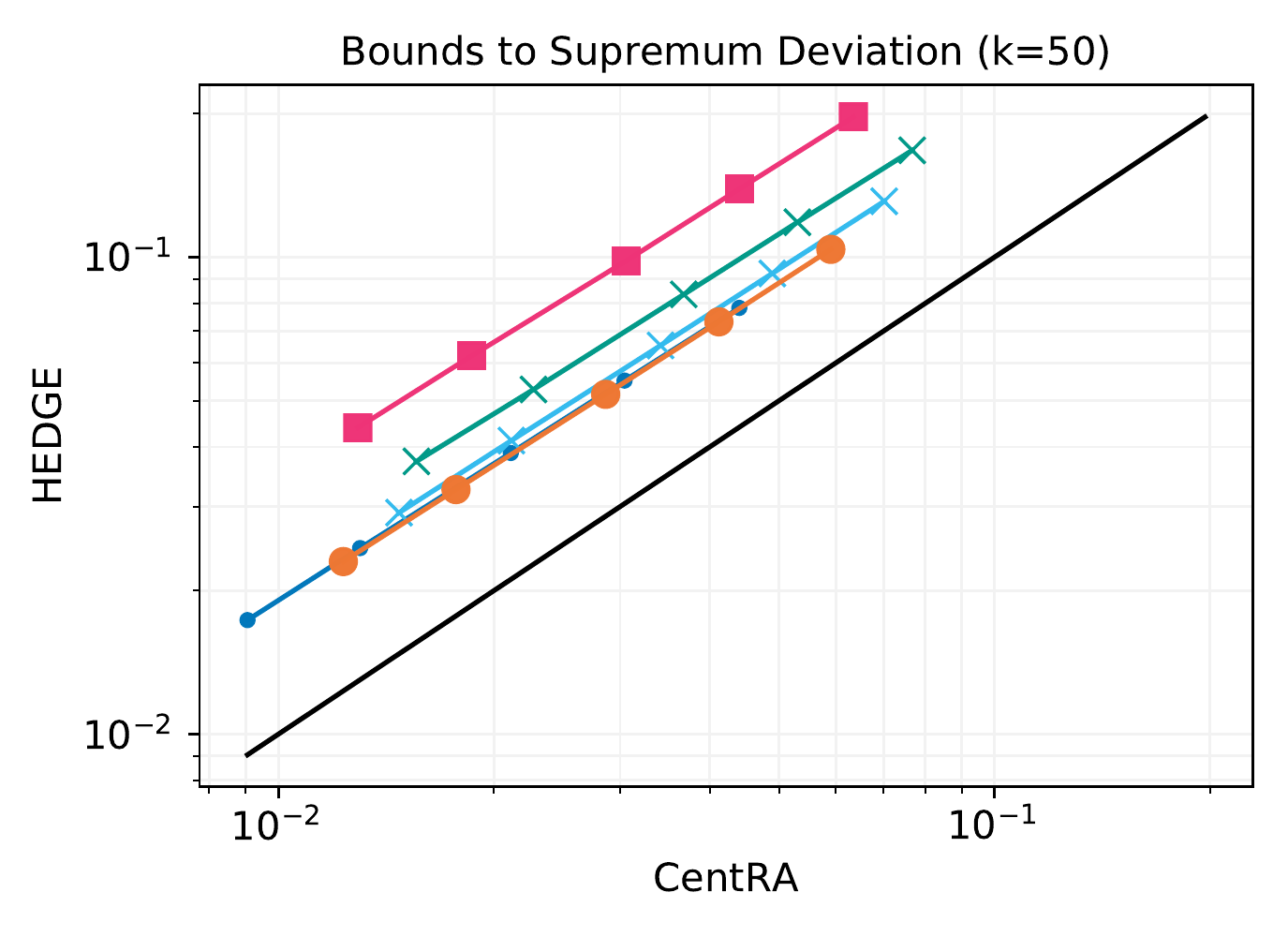}
  \caption{}
\end{subfigure}
\begin{subfigure}{.32\textwidth}
  \centering
  \includegraphics[width=\textwidth]{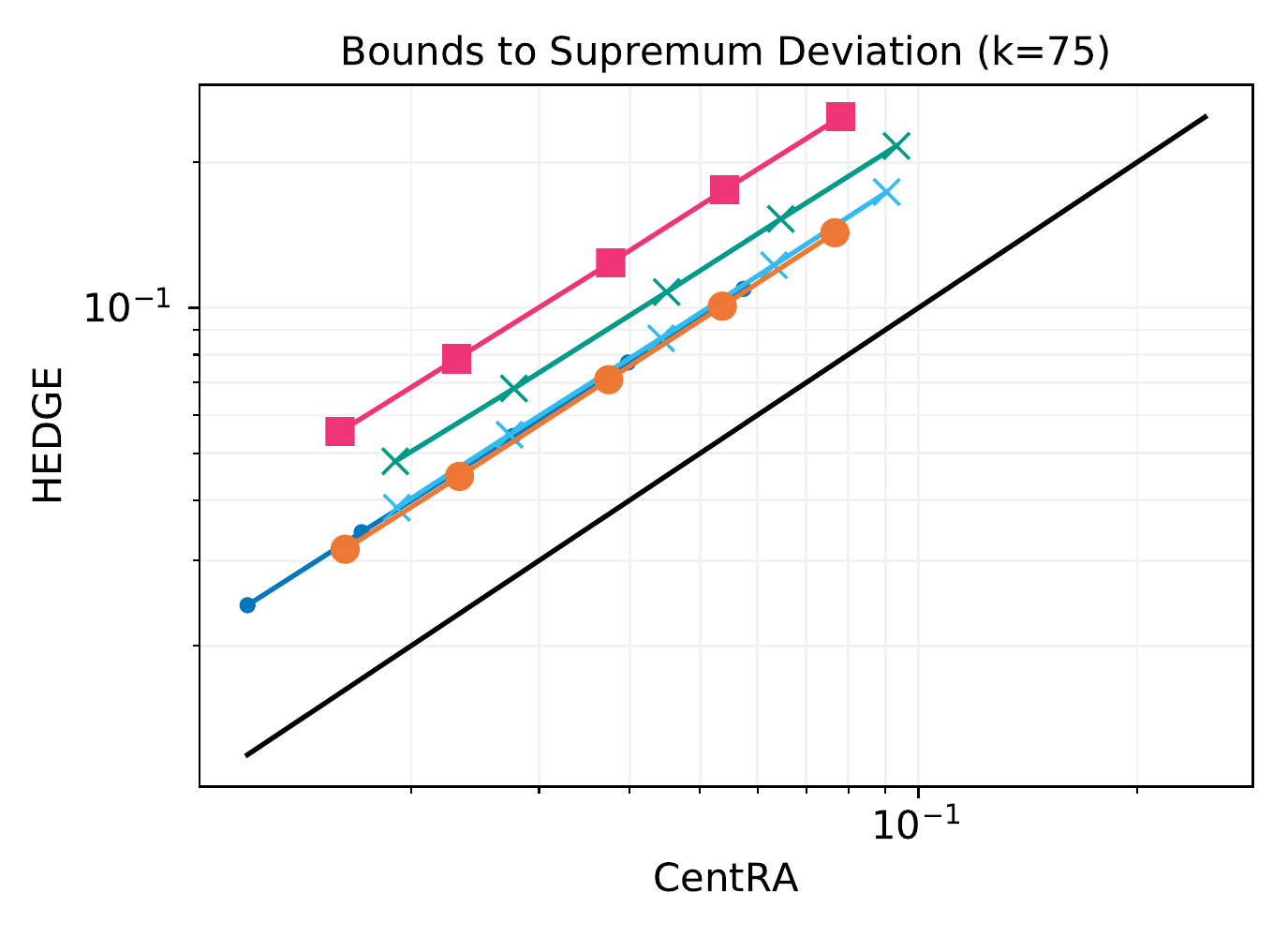}
  \caption{}
\end{subfigure} \\
\begin{subfigure}{.32\textwidth}
  \centering
  \includegraphics[width=\textwidth]{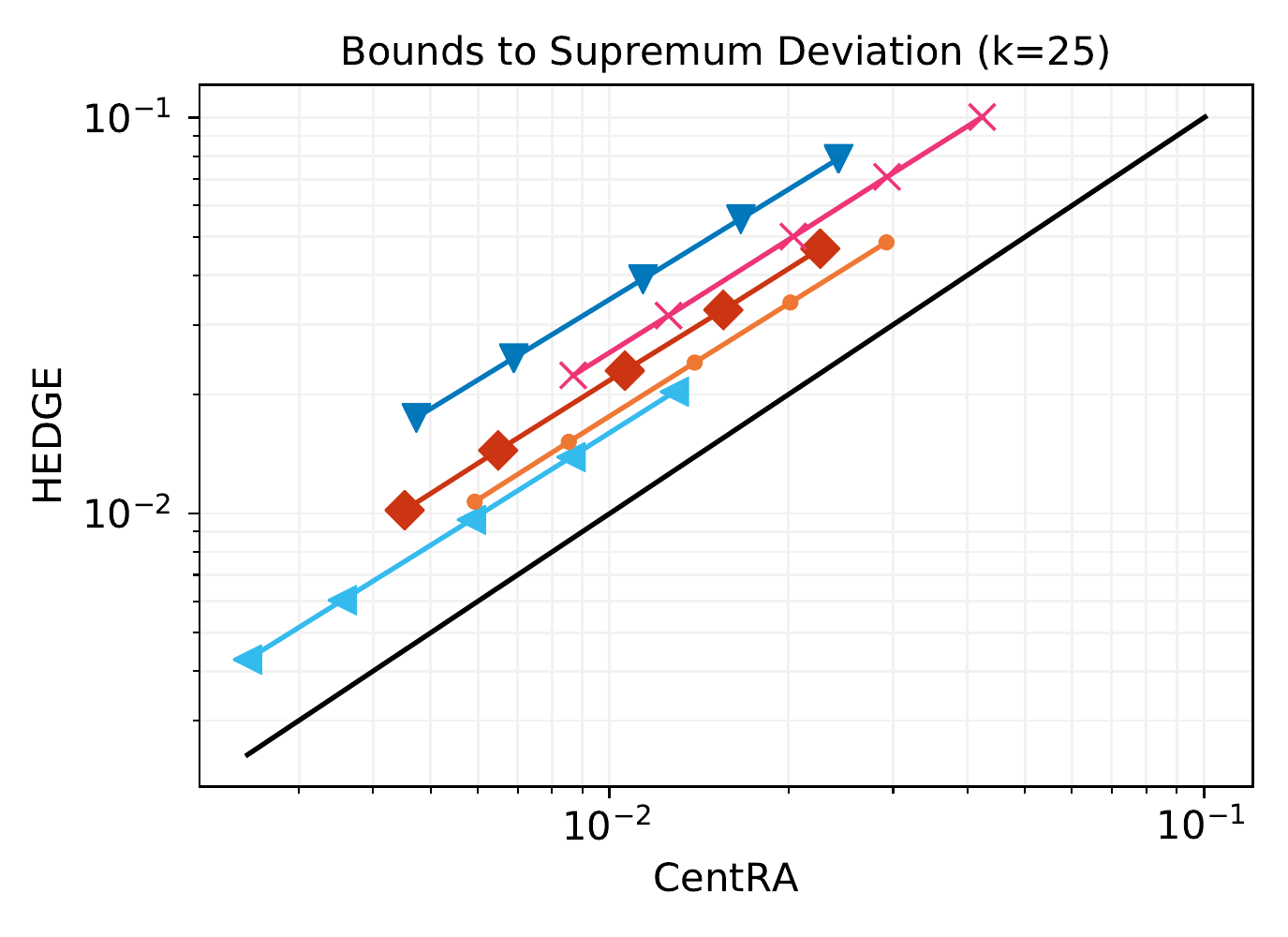}
  \caption{}
\end{subfigure}
\begin{subfigure}{.32\textwidth}
  \centering
  \includegraphics[width=\textwidth]{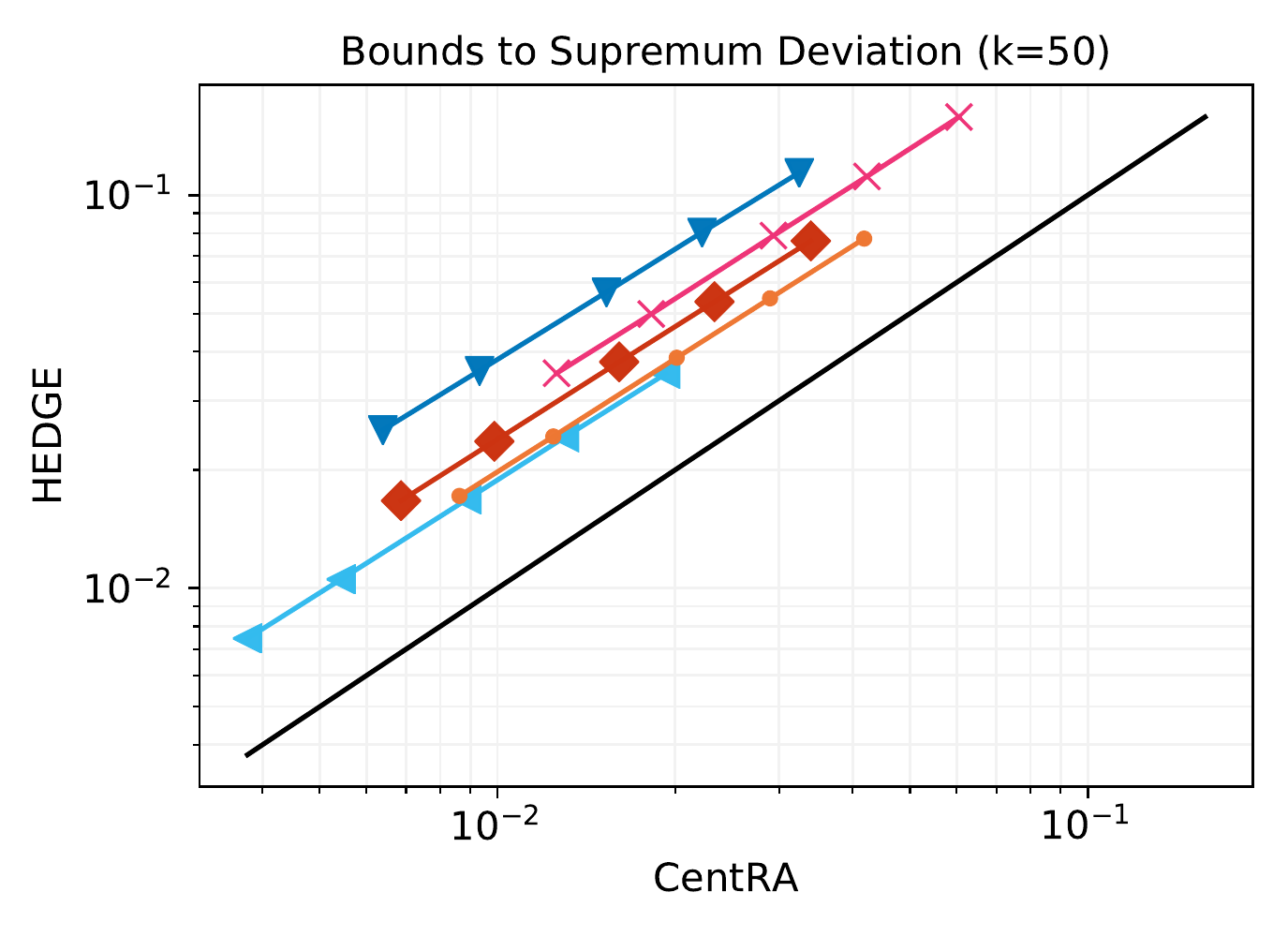}
  \caption{}
\end{subfigure}
\begin{subfigure}{.32\textwidth}
  \centering
  \includegraphics[width=\textwidth]{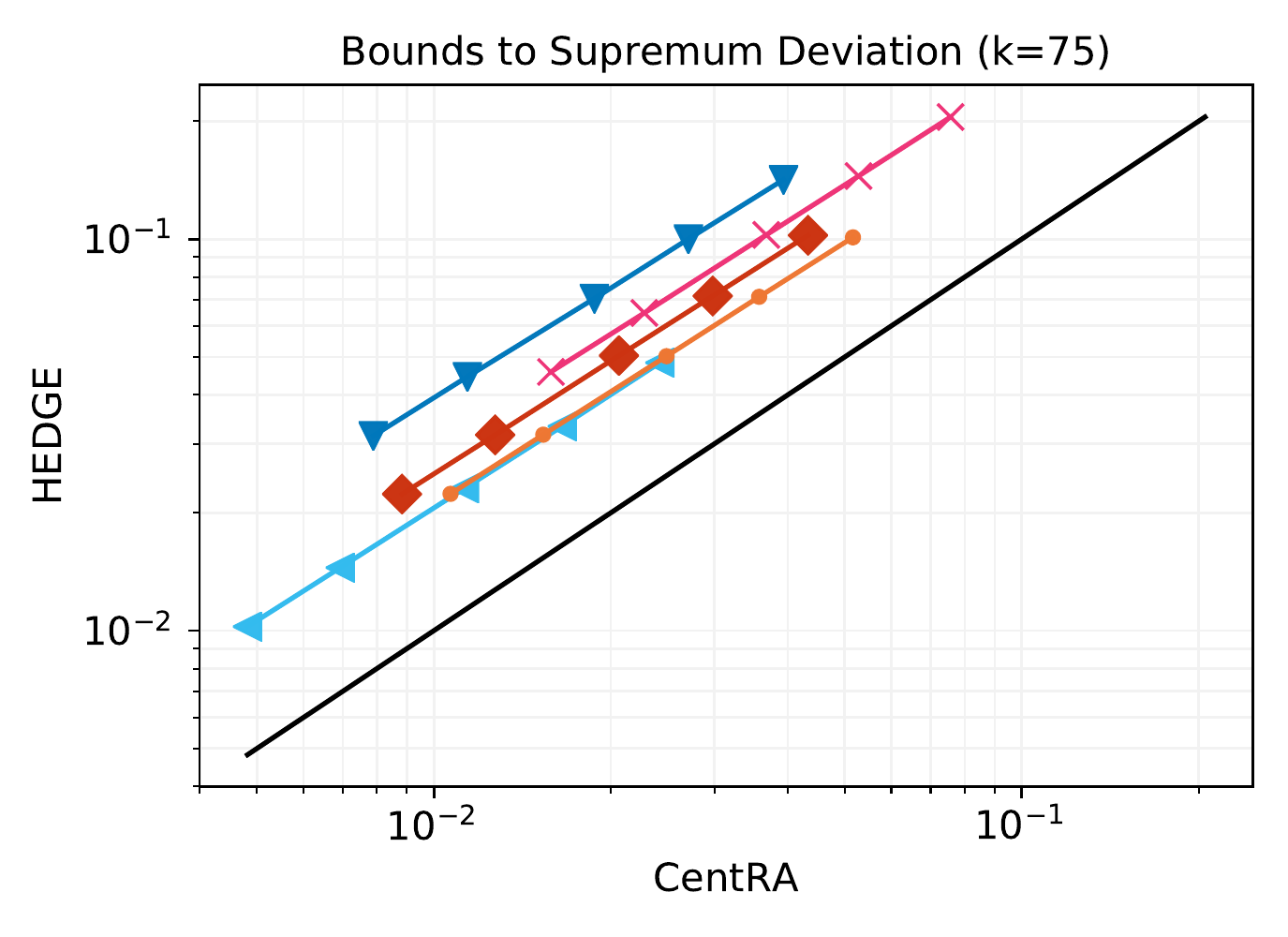}
  \caption{}
\end{subfigure} 
\caption{
Additional Figures comparing the 
bounds to the Supremum Deviation $\supdev$ obtained by HEDGE ($y$ axes, based on the union bound) and \algname\ ($x$ axes, Section~\ref{sec:algorithm}) on samples of size $m$, for $k \in \{25, 50,75\}$. 
Figures (a)-(c): undirected graphs.
Figures (d)-(f): directed graphs. 
Each point corresponds to a different value of $m$.
The black diagonal line corresponds to $y = x$.
}
\label{fig:fixedmappx}
\Description{This Figure shows additional results, for k=25, k=50, and and k=75, of the experiments shown in Figure 1.}
\end{figure*}

\begin{figure*}[ht]
\centering
\begin{subfigure}{.75\textwidth}
  \centering
  \includegraphics[width=\textwidth]{./figures/bounds-fixed-m-legend.pdf}
\end{subfigure} \\
\begin{subfigure}{.24\textwidth}
  \centering
  \includegraphics[width=\textwidth]{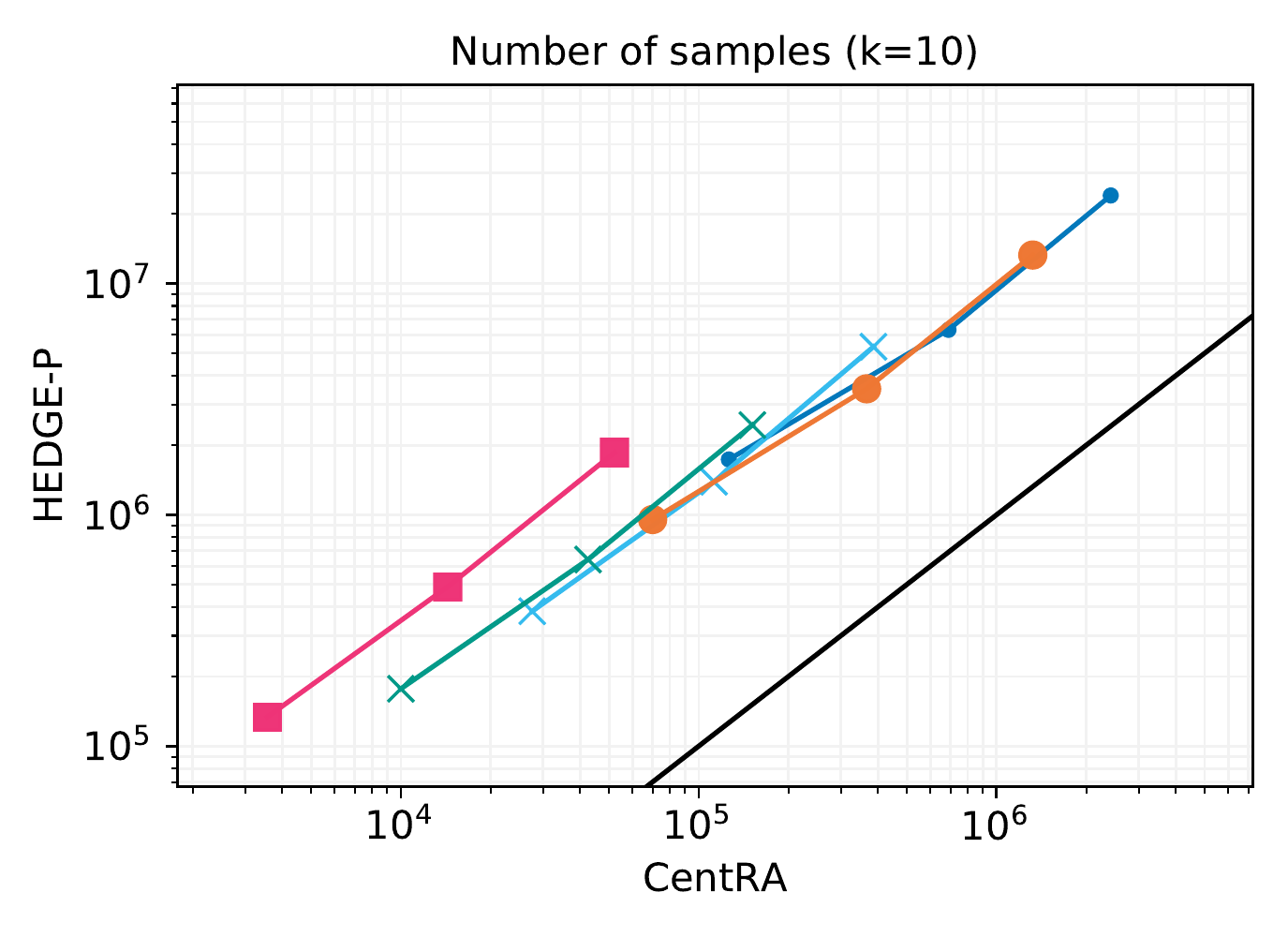}
\end{subfigure}
\begin{subfigure}{.24\textwidth}
  \centering
  \includegraphics[width=\textwidth]{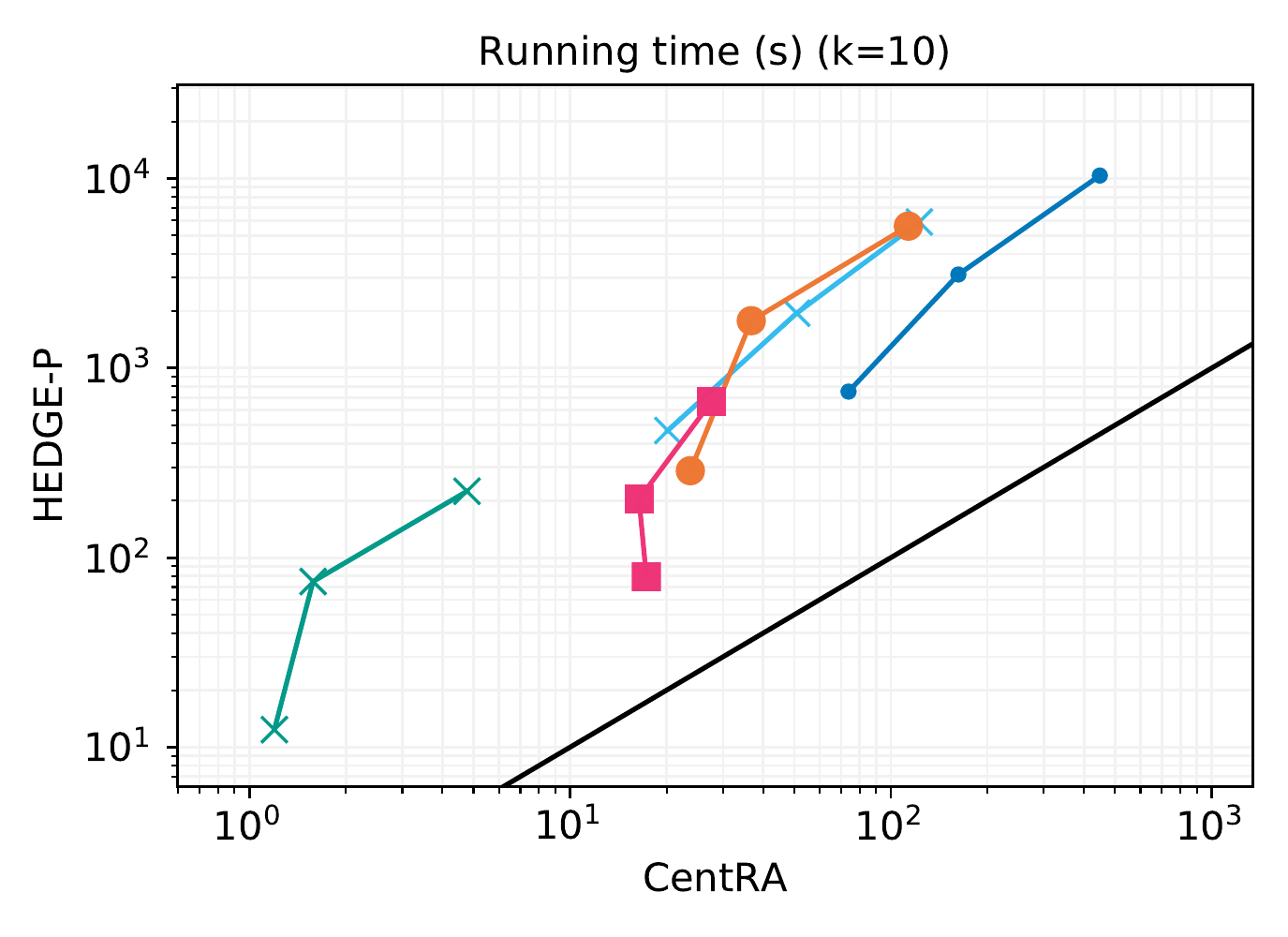}
\end{subfigure}
\begin{subfigure}{.24\textwidth}
  \centering
  \includegraphics[width=\textwidth]{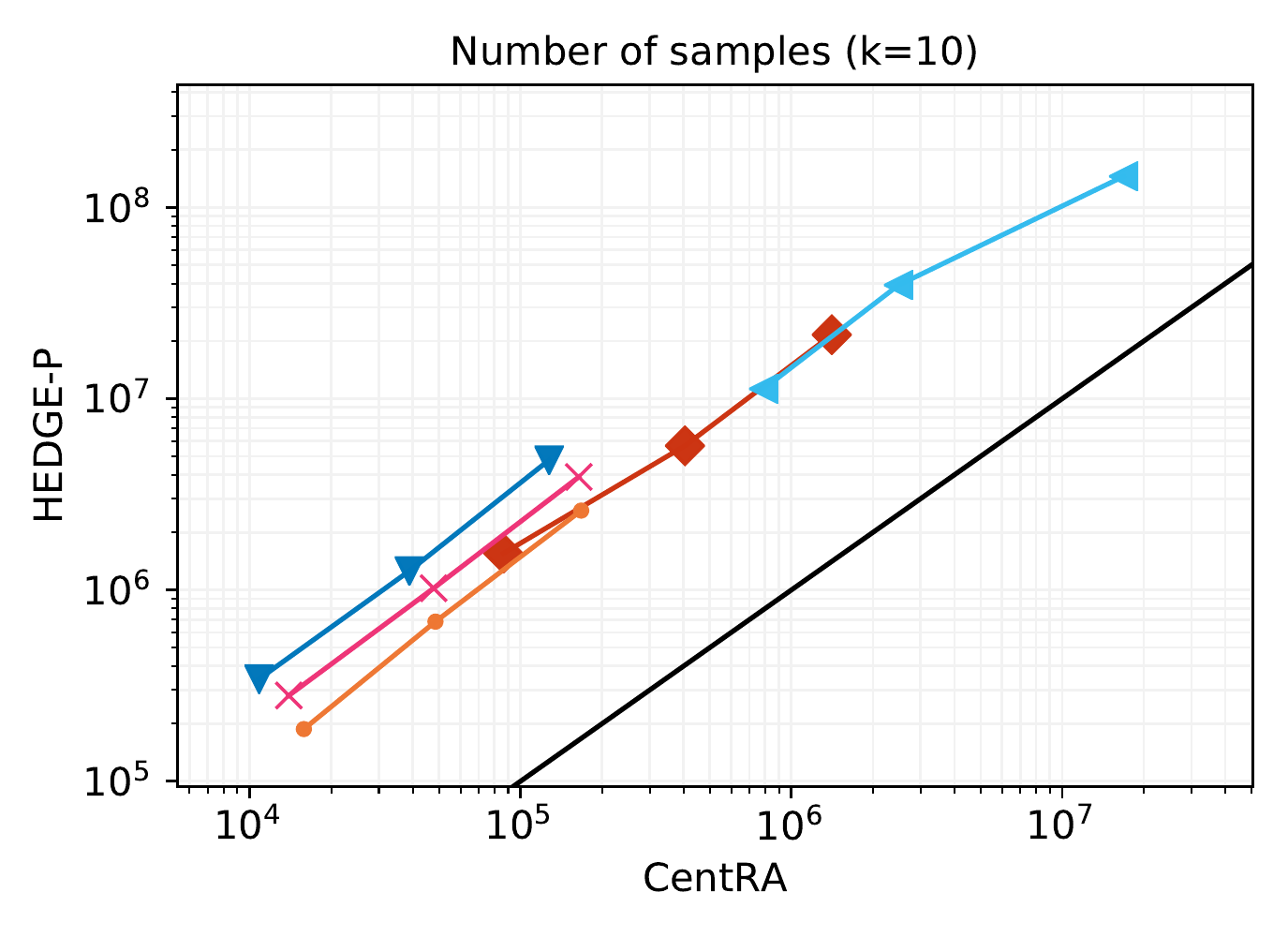}
\end{subfigure}
\begin{subfigure}{.24\textwidth}
  \centering
  \includegraphics[width=\textwidth]{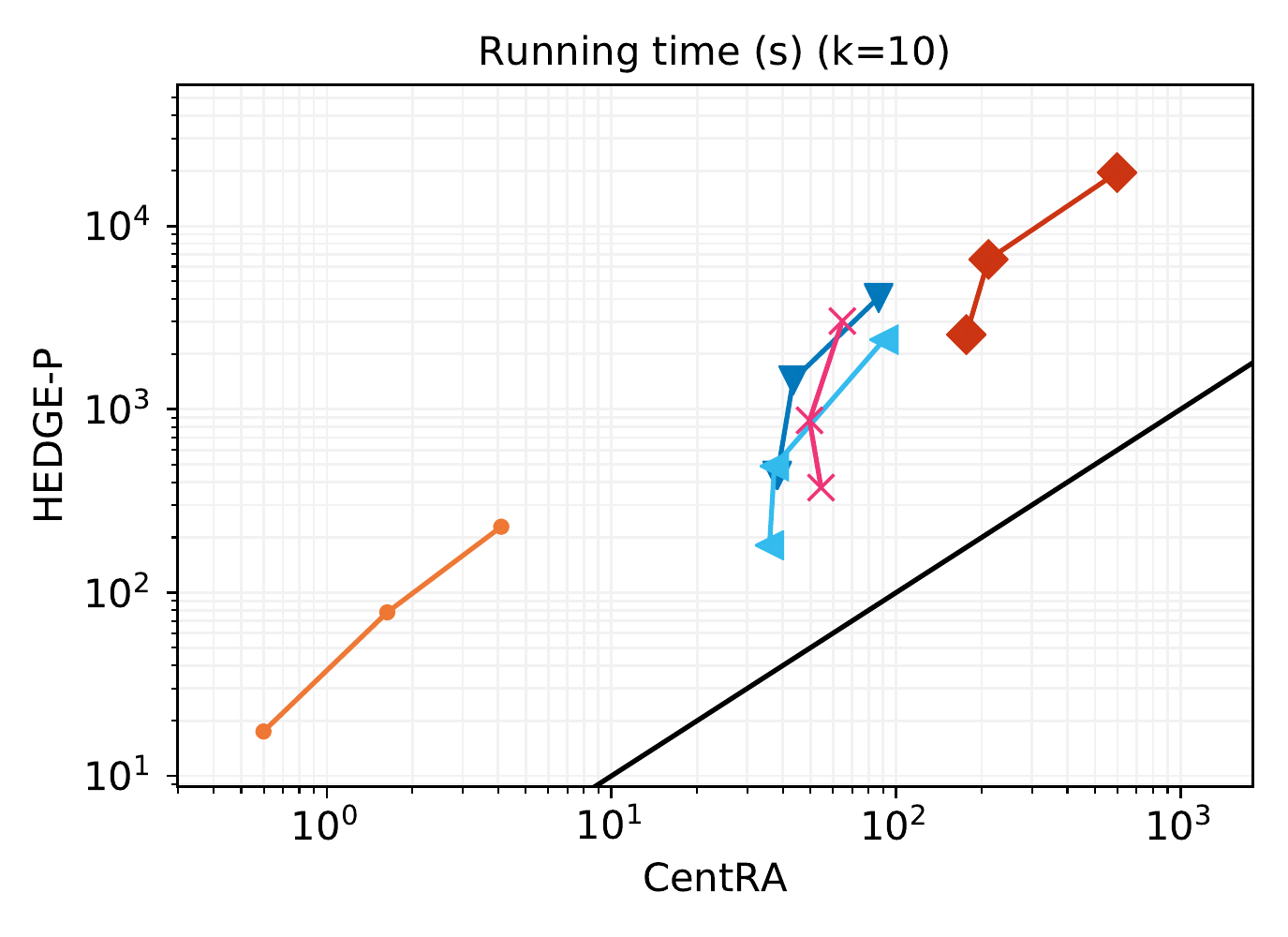}
\end{subfigure}
\begin{subfigure}{.24\textwidth}
  \centering
  \includegraphics[width=\textwidth]{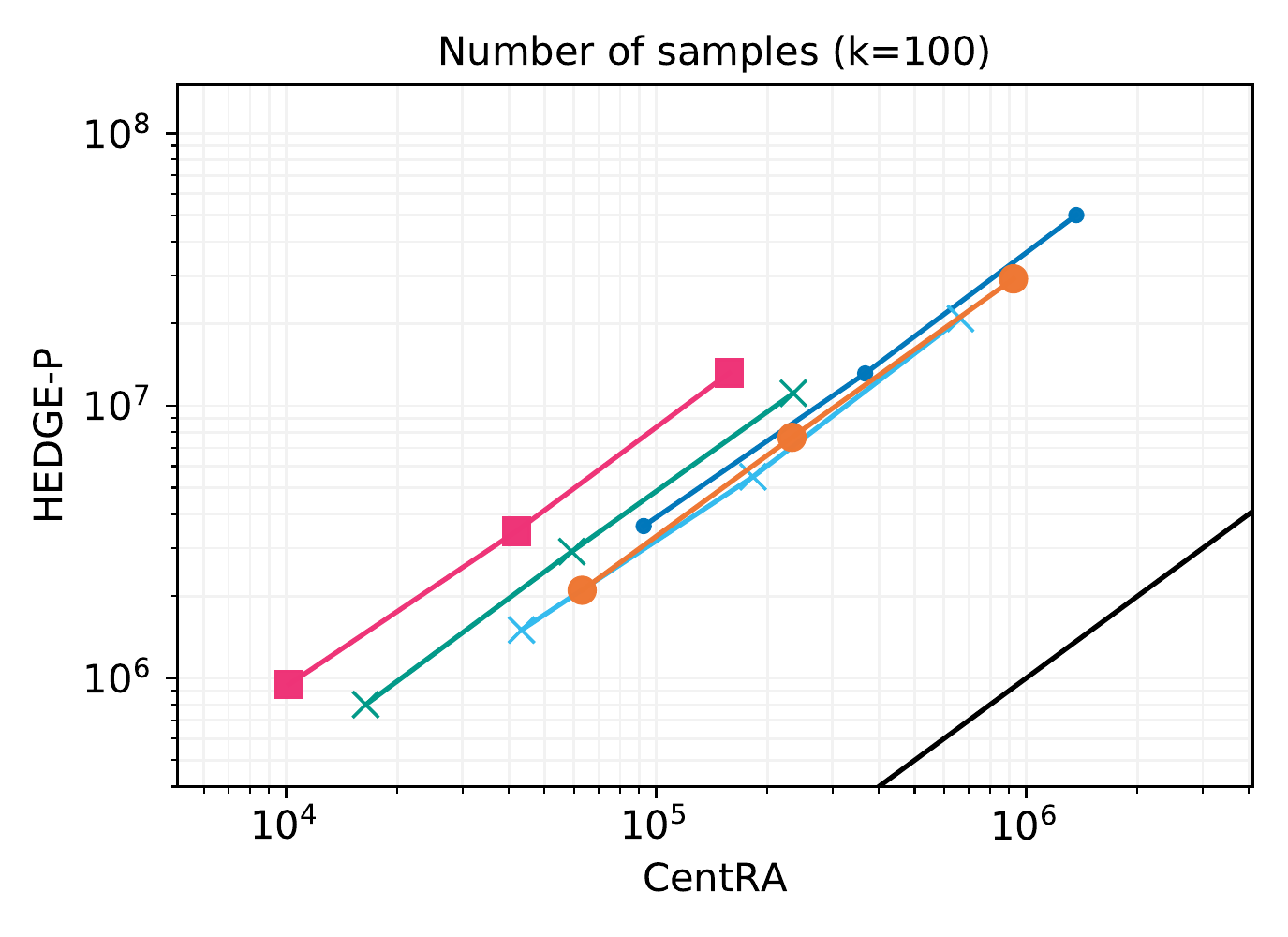}
\end{subfigure}
\begin{subfigure}{.24\textwidth}
  \centering
  \includegraphics[width=\textwidth]{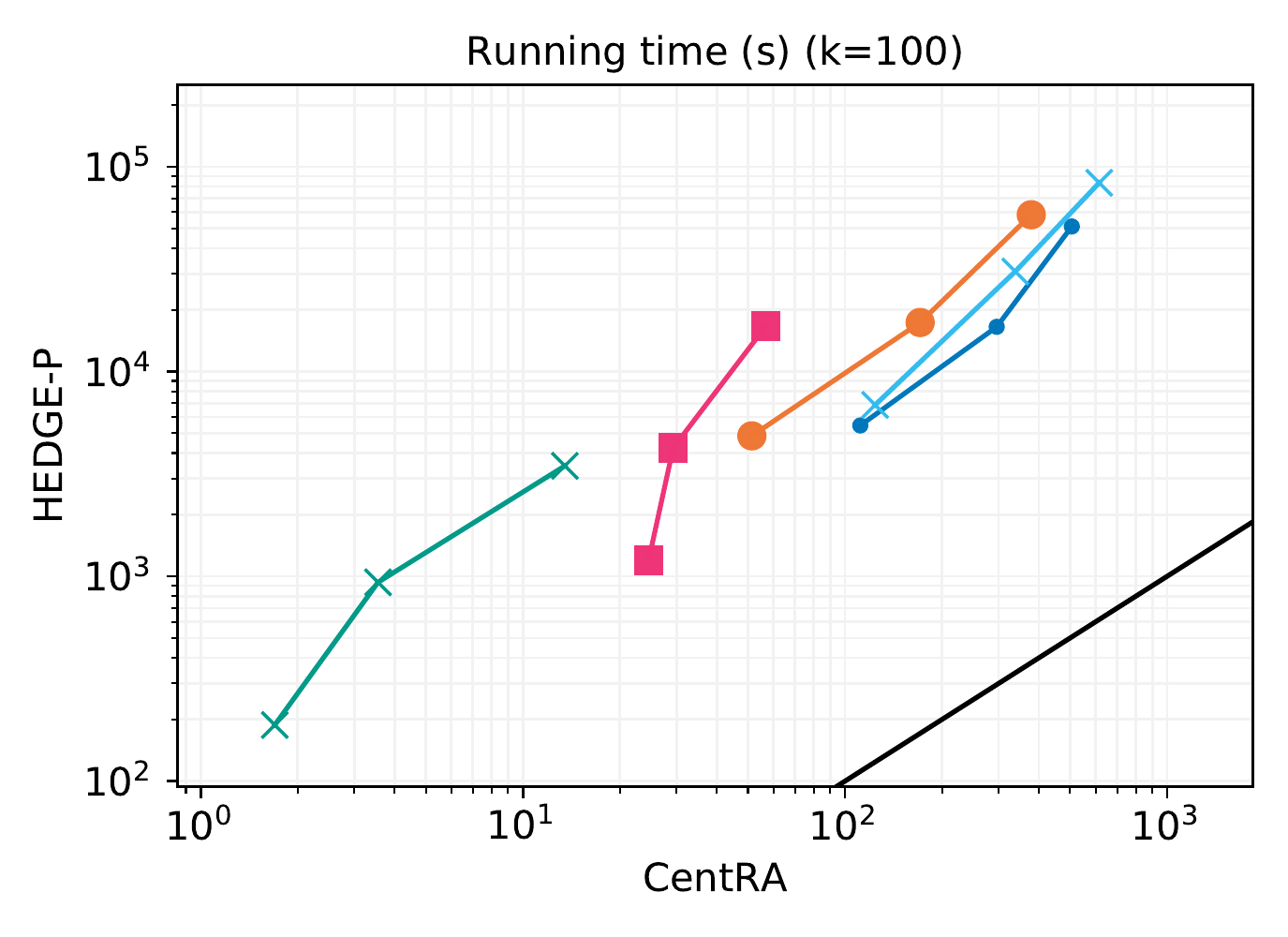}
\end{subfigure}
\begin{subfigure}{.24\textwidth}
  \centering
  \includegraphics[width=\textwidth]{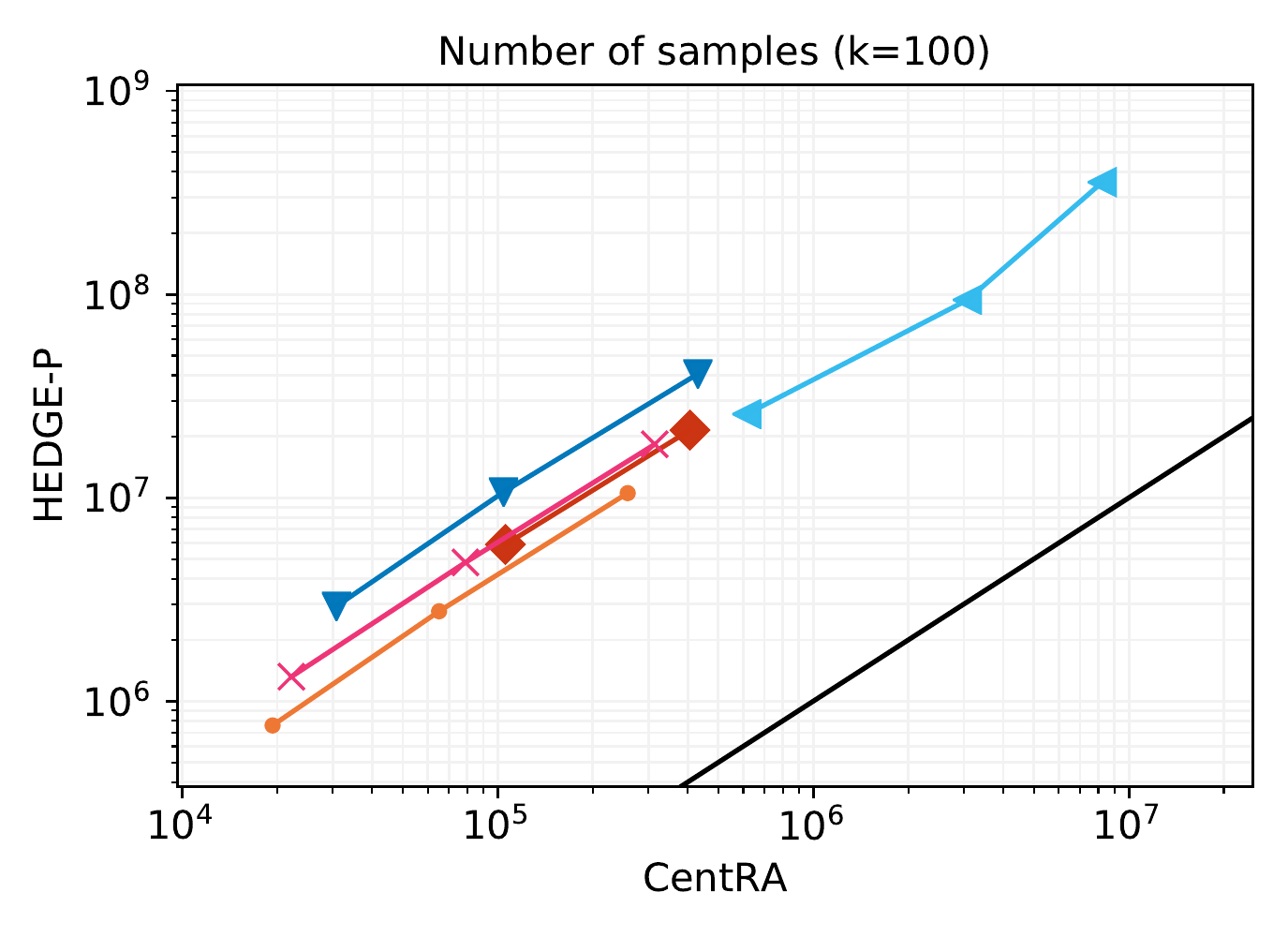}
\end{subfigure}
\begin{subfigure}{.24\textwidth}
  \centering
  \includegraphics[width=\textwidth]{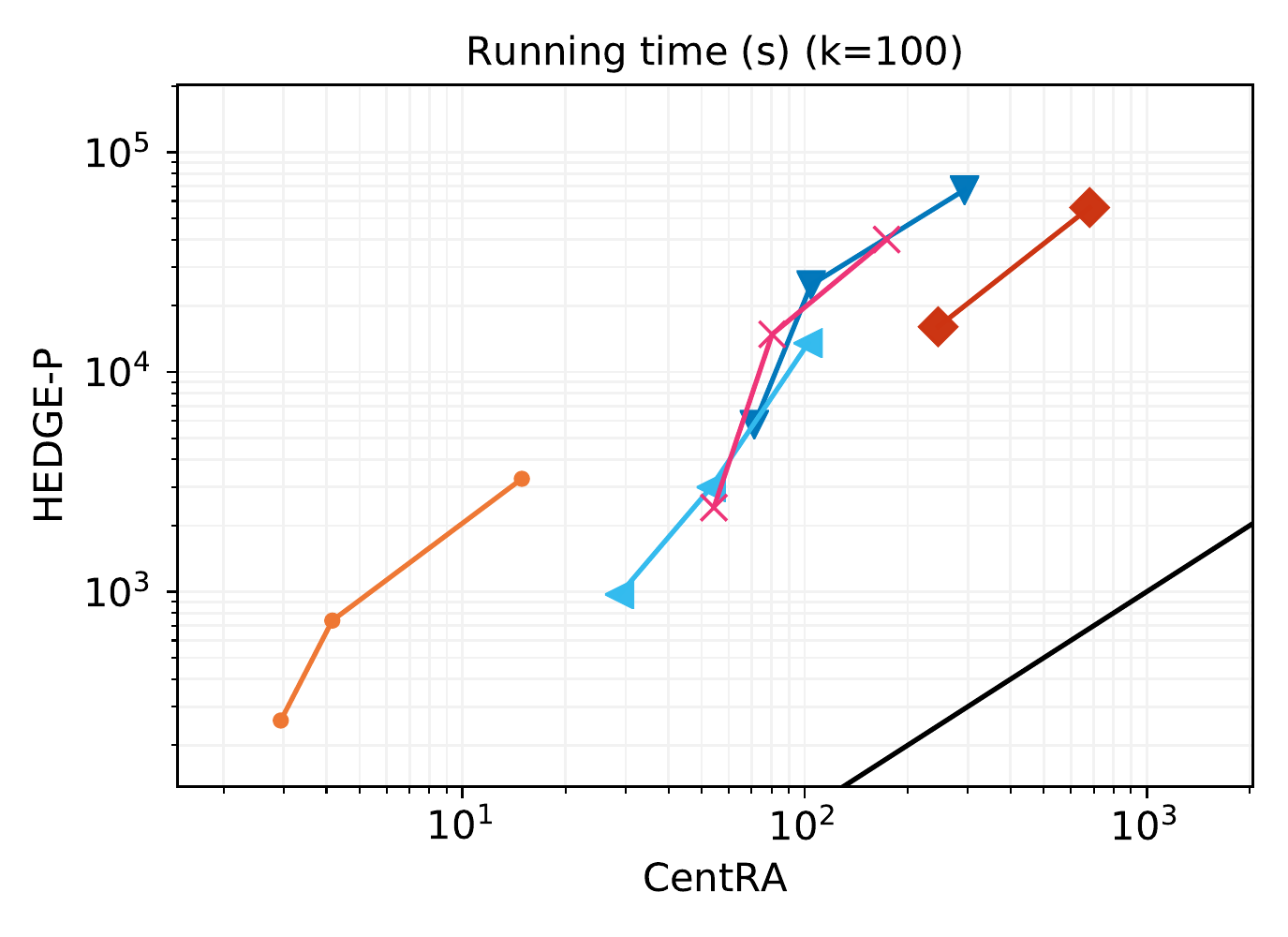}
\end{subfigure}
\caption{
Additional results (analogous to Figure \ref{fig:progsampling}) for $k=10$ and $k=100$.
}
\label{fig:progsamplingappx}
\Description{This Figure shows additional results, for k=10, and k=100, of the experiments shown in Figure 2.}
\end{figure*}

\begin{figure*}[ht]
\centering
\begin{subfigure}{.75\textwidth}
  \centering
  \includegraphics[width=\textwidth]{./figures/bounds-fixed-m-legend.pdf}
\end{subfigure} \\
\begin{subfigure}{.32\textwidth}
  \centering
  \includegraphics[width=\textwidth]{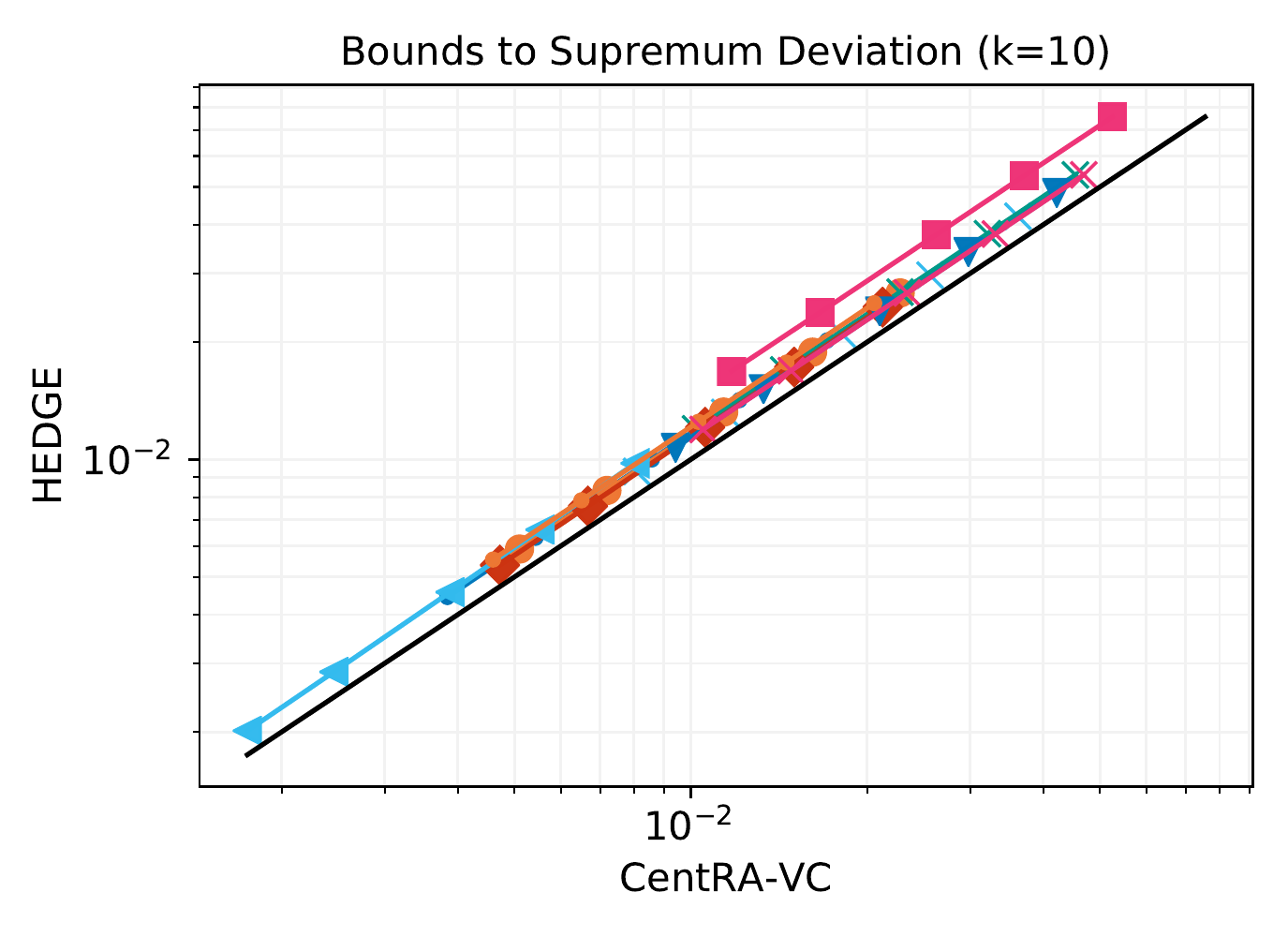}
\end{subfigure}
\begin{subfigure}{.32\textwidth}
  \centering
  \includegraphics[width=\textwidth]{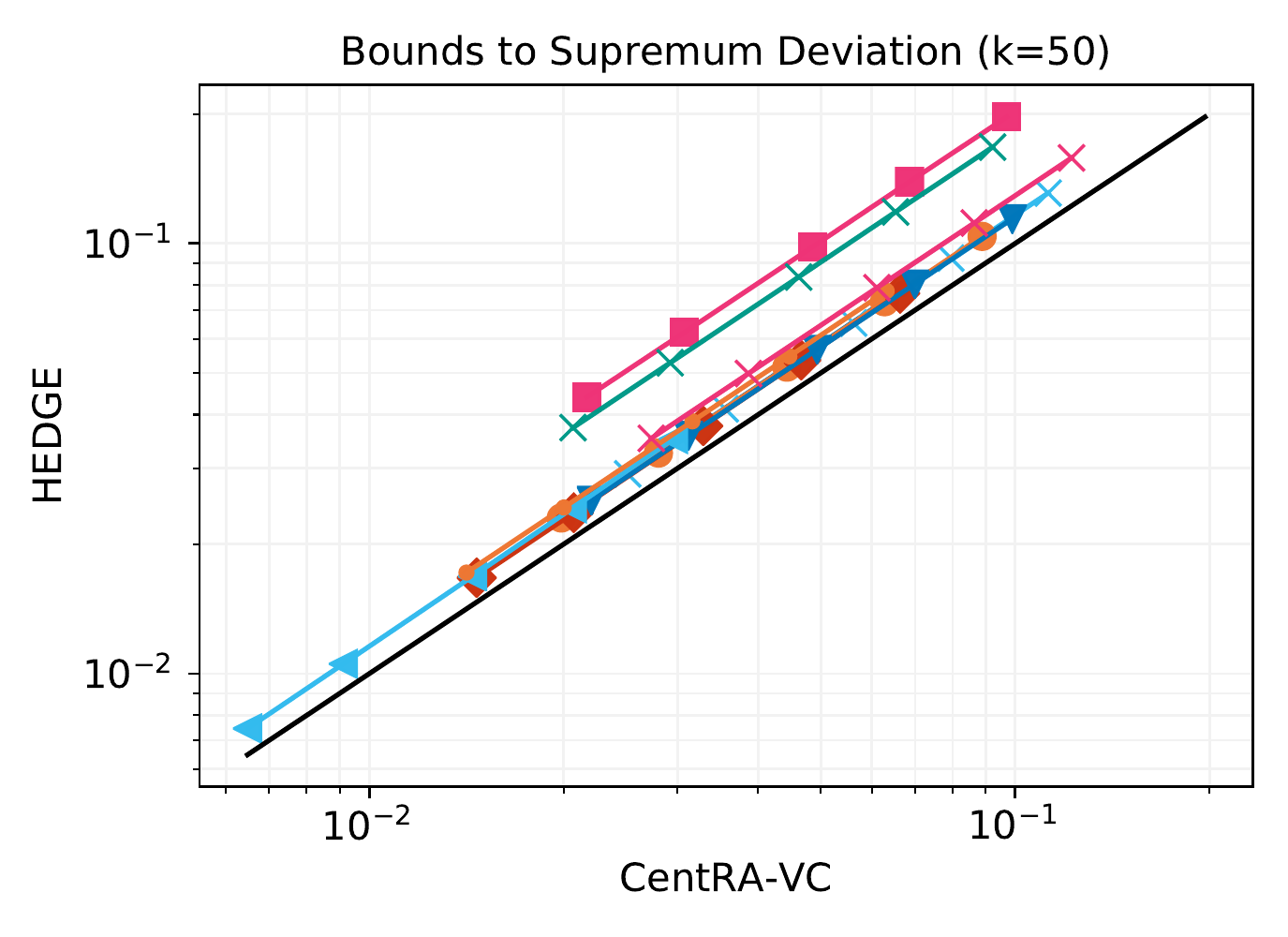}
\end{subfigure}
\begin{subfigure}{.32\textwidth}
  \centering
  \includegraphics[width=\textwidth]{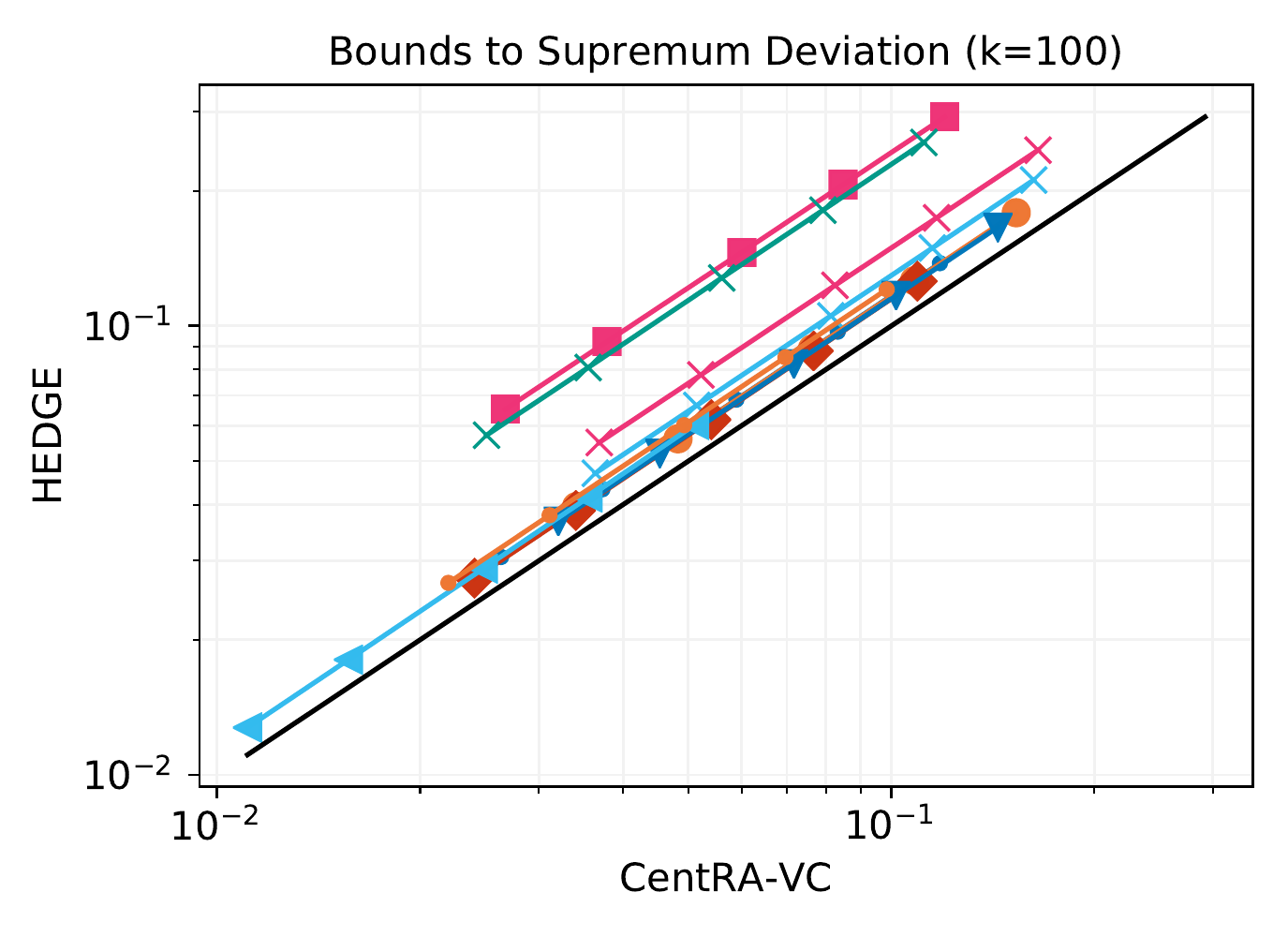}
\end{subfigure}
\caption{
Comparison between the 
bounds to the Supremum Deviation $\supdev$ obtained by HEDGE ($y$ axes, based on the union bound) and \algname-VC ($x$ axes, using the VC-dimension bounds of Section \ref{sec:samplecomplexity}) on samples of size $m \in \{ 5 \cdot 10^4 , 10^5 , 2 \cdot 10^5 , 5 \cdot 10^5 , 10^6 \}$, for $k \in \{10,50,100\}$. 
Each point corresponds to a different value of $m$.
The black diagonal line is at to $y = x$.
}
\label{fig:samplecompl-vc-hedge}
\Description{This Figure compares the bounds to the supremum deviation obtained from samples of size m using HEDGE and CentRA-VC, a variant of CentRA using the VC-dimension instead of Rademacher averages. The plots show that CentRA-VC yields more accurate bounds to the SD.}
\end{figure*}

\begin{figure*}[ht]
\centering
\begin{subfigure}{.75\textwidth}
  \centering
  \includegraphics[width=\textwidth]{./figures/bounds-fixed-m-legend.pdf}
\end{subfigure} \\
\begin{subfigure}{.32\textwidth}
  \centering
  \includegraphics[width=\textwidth]{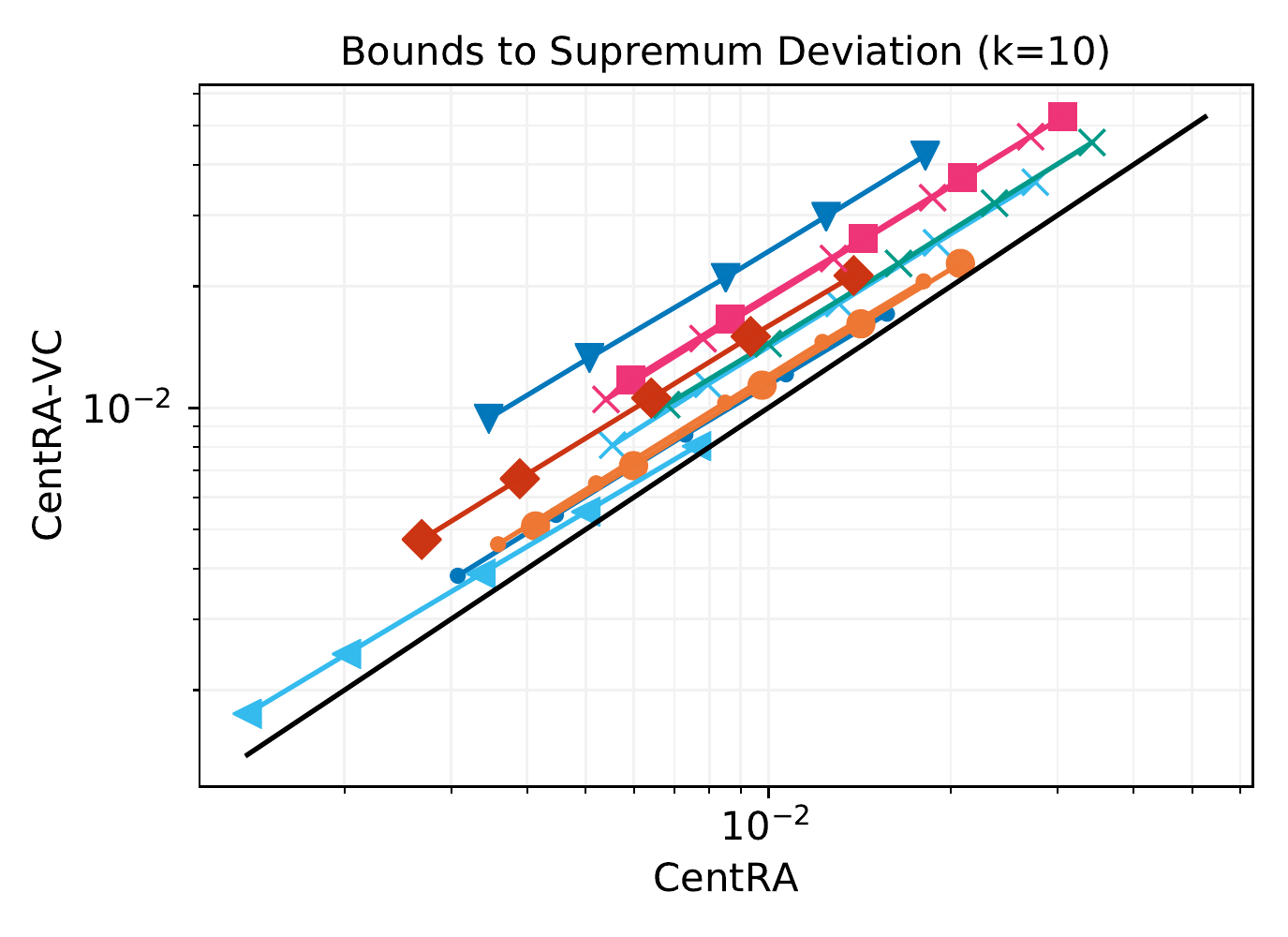}
\end{subfigure}
\begin{subfigure}{.32\textwidth}
  \centering
  \includegraphics[width=\textwidth]{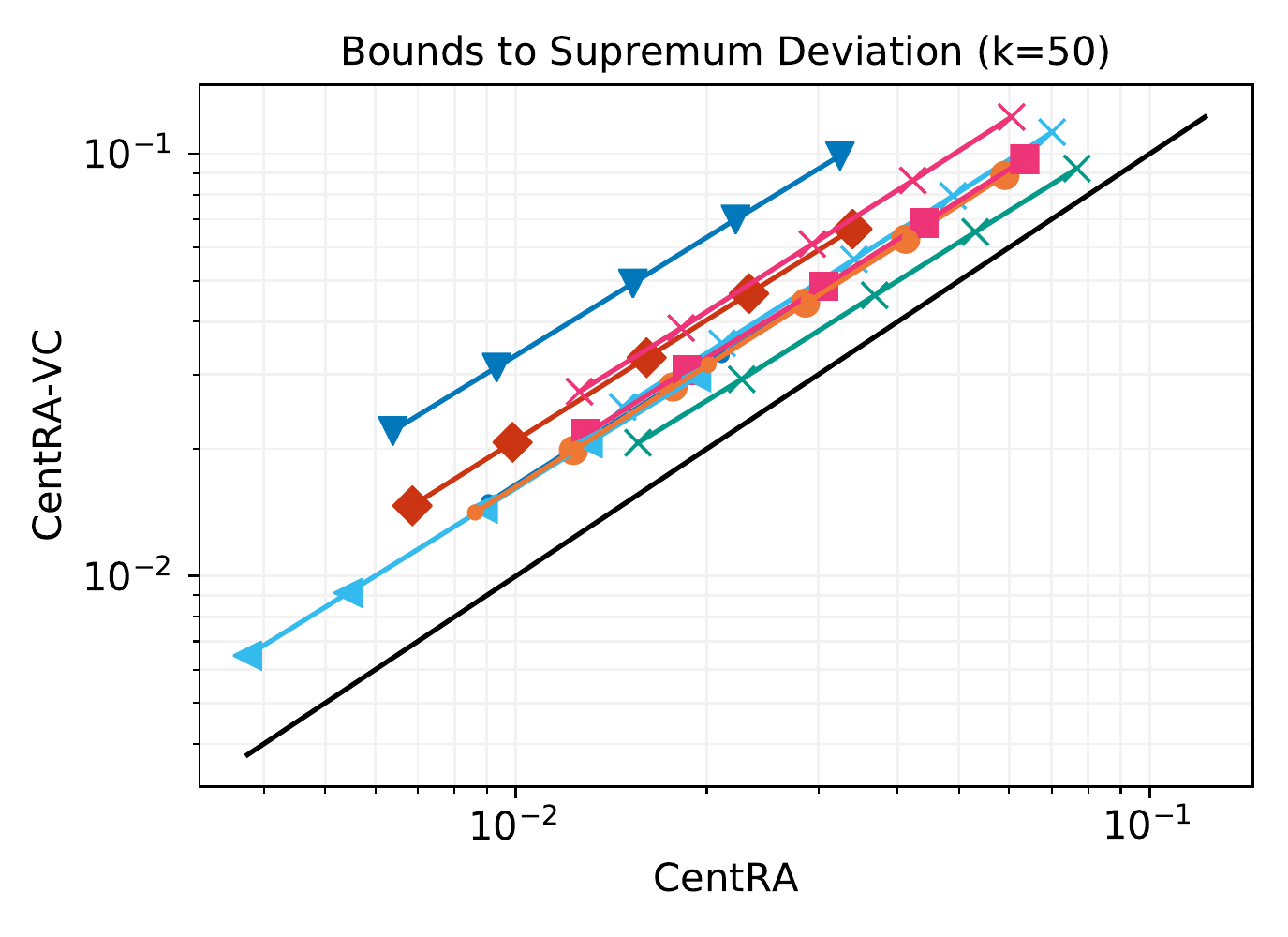}
\end{subfigure}
\begin{subfigure}{.32\textwidth}
  \centering
  \includegraphics[width=\textwidth]{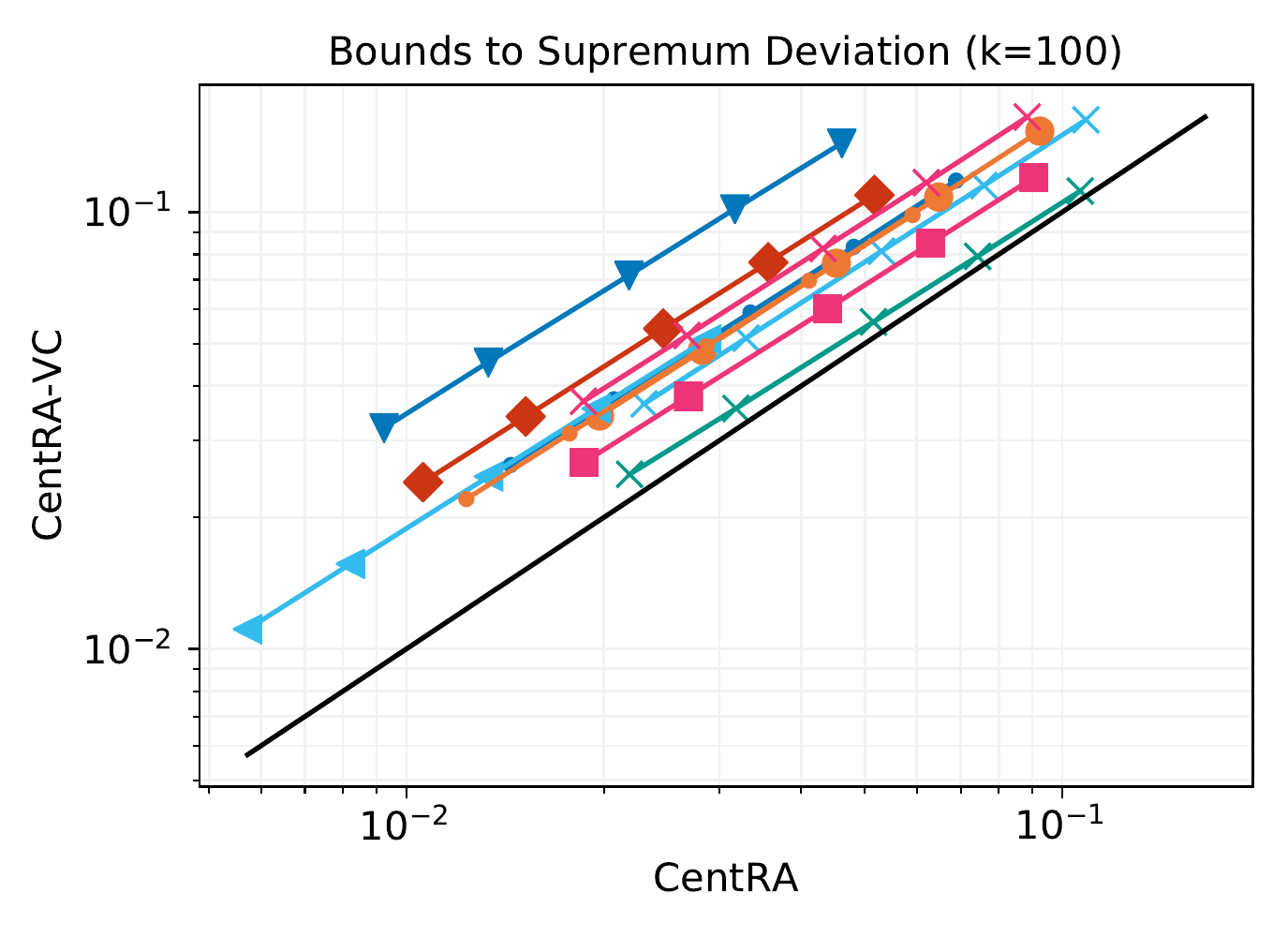}
\end{subfigure}
\caption{
Comparison between the 
bounds to the Supremum Deviation $\supdev$ obtained by 
\algname-VC ($y$ axes, using the VC-dimension bounds of Section \ref{sec:samplecomplexity}) 
and \algname\ ($x$ axes, using the Rademacher averages bounds of Section \ref{sec:boundsupdevrade}) on samples of size $m \in \{ 5 \cdot 10^4 , 10^5 , 2 \cdot 10^5 , 5 \cdot 10^5 , 10^6 \}$, for $k \in \{10,50,100\}$. 
Each point corresponds to a different value of $m$.
The black diagonal line is at to $y = x$.
}
\label{fig:samplecompl-vc-centra}
\Description{This Figure compares the bounds to the supremum deviation obtained from samples of size m using CentRA and CentRA-VC, a variant of CentRA using the VC-dimension instead of Rademacher averages. The plots show that CentRA yields more accurate bounds to the SD.}
\end{figure*}